\documentclass[10pt,reqno]{amsart}

\usepackage{geometry}
\usepackage[foot]{amsaddr}
\usepackage{xcolor}
\RequirePackage{times}
\RequirePackage[OT1]{fontenc}
\usepackage{amsmath,amsfonts,amsbsy,amsgen,amscd,mathrsfs,amssymb,amsthm}
\RequirePackage{bm}
\RequirePackage{amsthm,amsmath,amsfonts,amssymb}
\RequirePackage[numbers]{natbib}
\RequirePackage[colorlinks=true,linkcolor=blue,citecolor=blue,urlcolor=blue]{hyperref}
\RequirePackage{graphicx}
\usepackage[title]{appendix}
\usepackage{algorithm}
\usepackage{algpseudocode}
\usepackage{enumitem}
\usepackage[T1]{fontenc}
\def\Xset{\mathsf{X}}
\def\Yset{\mathsf{Y}}
\def\sigmaX{\mathcal{X}}
\def\sigmaY{\mathcal{Y}}
\def\1{\mathbbm{1}}

\def\pE{\mathbb{E}}
\def\pP{\mathbb{P}}
\def\pV{\mathbb{V}}
\def\N{\mathbb{N}}
\def\plim{\overset{\pP}{\longrightarrow}}
\def\dlim{\Longrightarrow}
\def\gauss{\mathcal{N}}

\def\GT{\mathsf{GT}}
\def\BS{\mathsf{BS}}
\def\rmd{\mathrm{d}}
\def\R{\mathbb{R}}
\def\aslim{\overset{a.s.}{\underset{N \to \infty}{\longrightarrow}}}
\def\backsum{\mathcal{T}}
\def\sett{\mathcal{S}}
\def\eve{\mathrm{E}}
\def\Bset{\mathcal{B}}
\def\zero{\mathbf{0}}
\def\bigo{\mathcal{O}}
\def\boldone{\mathbf{1}}
\def\FFBS{\mathrm{FFBS}}

\def\iid{\overset{\mathrm{iid}}{\sim}}

\def\simplelim{\underset{N \to \infty}{\lim}}
\def\boundg{G_\infty}
\def\paris{\textit{PaRIS}}
\def\Pbacksum{\widetilde{\backsum}} 
\def\bigH{\textbf{H}}
\def\bigF{\textbf{F}}
\def\pos{\mathsf{Pos}}

\newcommand{\Const}[1]{\mathcal{C}_{N,b,#1}}
\newcommand{\Gfilt}[2]{\mathcal{G}^N _{#2}}

\newcommand{\card}[1]{\mathrm{Card}(#1)}
\newcommand{\bigprod}[1]{\overline{\Lambda}^{1,2} _{#1}}
\newcommand{\intersect}[2]{\mathrm{I}_{#1}(#2)}
\newcommand{\joint}[1]{\gamma _{#1}}
\newcommand{\pred}[1]{\eta _{#1}}
\newcommand{\filter}[1]{\phi _{#1}}
\newcommand{\F}[1]{\mathcal{F}^N _{#1}}
\newcommand{\GG}[1]{\mathrm{\mathbf{G}}_{#1}}

\newcommand{\smooth}[2]{\phi_{#1 | #2}}
\newcommand{\smoothN}[3]{\phi^{#3} _{#1 | #2}}
\newcommand{\particletraj}[2]{\xi^{#1 _{0:#2}} _{0:#2}}
\newcommand{\particlecloud}[1]{\bm \xi^{1:N} _{0:#1}}
\newcommand{\bwker}[1]{\mathrm{\mathbf{B}} _{\phi_{#1}}}
\newcommand{\bwpath}[1]{\mathrm{\mathbf{T}} _{#1}}
\newcommand{\transition}[1]{\ifthenelse{\equal{#1}{0}}{\mathrm{M}_0}{\mathrm{M}_{#1}}}

\newcommand{\Q}[1]{\mathrm{\mathbf{Q}}_{#1}}
\newcommand{\Qmarg}[1]{\overline{\mathrm{\mathbf{Q}}}_{#1}}
\newcommand{\asymptvarest}[2]{\ifthenelse{\equal{#2}{}}{\ensuremath{\Gamma _{0: #1 \mid #1}}}{\ensuremath{\Gamma^{#2}_{0: #1 \mid #1}}}}
\newcommand{\asymptvarffbs}[2]{\ifthenelse{\equal{#2}{}}{\ensuremath{\Gamma^{\textrm{FFBS}} _{0: #1 \mid #1}}}{\ensuremath{\Gamma^{#2, \textrm{FFBS}}_{0: #1 \mid #1}}}}
\newcommand{\transitiondens}[1]{\ifthenelse{\equal{#1}{0}}{m_0}{m_{#1}}}
\newcommand{\bitransition}[2]{\mathcal{M}^{#1} _{#2}}
\newcommand{\mumeasure}[1]{\mathcal{Q}_{#1}}
\newcommand{\bounded}[1]{\ifthenelse{\equal{#1}{}}{\ensuremath{\mathbb{F}_b (\mathcal{X})}}{\ensuremath{\mathbb{F}_b (\mathcal{X}^{\otimes {#1}})}}}
\newcommand{\measurable}[1]{\ifthenelse{\equal{#1}{}}{\ensuremath{\mathbb{F}(\mathcal{X})}}{\ensuremath{\mathbb{F}(\mathcal{X}^{\otimes {#1}})}}}
\newcommand{\additive}[1]{\mathbb{A}_b(\mathcal{X}^{\otimes #1})}
\newcommand{\muest}[2]{\mathcal{Q}^{N,#1} _{#2}}

\newcommand{\predasymptvar}[2]{\mathcal{V}^{\infty}_{\eta, #1} \left(#2\right)}
\newcommand{\jointasymptvar}[2]{\mathcal{V}^{\infty}_{\gamma, #1} (#2)}

\newcommand{\predasymptvarestim}[2]{\mathcal{V}^{N}_{\eta, #1} (#2)}
\newcommand{\asymptvar}[1]{\mathcal{V}^{\infty} _{#1}}
\newcommand{\asymptvarestim}[2]{\mathcal{V}^{N,#2} _{#1}}

\newcommand{\tbtasymptvarestim}[2]{\overline{\mathcal{V}}^{N,#2} _{#1}}

\newcommand{\parisasymptvar}[2]{\mathsf{V}^{N,M} _{#1}}

\newcommand{\esssup}[2][]
{\ifthenelse{\equal{#1}{}}{\left\| #2 \right\|_\infty}{\left\| #2 \right\|^2_{\infty}}}
\newcommand{\eqdef}{\ensuremath{:=}}
\newcommand{\eqsp}{}
\newcommand{\weight}[2]{\omega^{#1} _{#2}}
\newcommand{\normweight}[2]{\mathcal{W}^{#1} _{#2}}
\newcommand{\particle}[2]{\xi^{#1} _{#2}}

\newcommand{\mixture}[3]{
    \ifthenelse{\equal{#3}{}}{\mathrm{Q} _{#1}(\particle{1:N}{#1 -1}, \rmd \particle{#2}{#1})}{\mathrm{Q}^{\otimes #3} _{#1}(\particle{1:N}{#1}, \rmd \particle{#2}{#1})}}
\newcommand{\ffbsasymptvar}[3]{\mathcal{V}^{#1} _{#2 | #3}}
\newcommand{\parismuest}[2]{\widetilde{\mathcal{Q}}^{N,M} _{#2}}
\newcommand{\A}[2]{\ifthenelse{\equal{#2}{}}{(\mathbf{A}{\mathbf{\ref{#1})}}}{(\mathbf{A}\mathbf{\ref{#1}:\ref{#2}})}}
\newcommand{\As}[2]{\ifthenelse{\equal{#2}{}}{(\mathbf{A}{\mathbf{\ref*{#1})}}}{(\mathbf{A}\mathbf{\ref*{#1}:\ref*{#2}})}}

\newcommand{\suppl}{appendix}

\newtheorem{thm}{Theorem}[section]
\newtheorem{corollary}[thm]{Corollary}
\newtheorem{lemma}[thm]{Lemma}
\newtheorem{proposition}[thm]{Proposition}

\theoremstyle{definition}

\theoremstyle{remark}
\newtheorem{rem}[thm]{Remark}
\newtheorem*{rem*}{Remark}

\newtheorem{example}[thm]{Example}
\numberwithin{equation}{section}

\newcounter{hypA}
\newenvironment{hypA}{\refstepcounter{hypA}\begin{itemize}
  \item[({\bf A\arabic{hypA}})]}{\end{itemize}}

\usepackage{bbm}

\begin{document}

 \title[Variance estimation for SMC algorithms]{Variance estimation for Sequential Monte Carlo algorithms: a backward sampling approach}

 \author{Yazid Janati El Idrissi}
\address{SAMOVAR, TELECOM SUDPARIS. INSTITUT POLYTECHNIQUE DE PARIS, 91120 PALAISEAU FRANCE}
\email{yazid.janati\_elidrissi@telecom-sudparis.eu}

\author{Sylvain Le Corff}
\address{LPSM, Sorbonne Université, UMR CNRS 8001, 4 Place Jussieu, 75005 Paris.}
\email{yohan.petetin@telecom-sudparis.eu}

\author{Yohan Petetin}
\email{sylvain.le\_corff@sorbonne-universite.fr}

\begin{abstract}
  In this paper, we consider the problem of online asymptotic variance estimation for particle filtering and smoothing. Current solutions for the particle filter rely on the particle genealogy and are either unstable or hard to tune in practice. We propose to mitigate these limitations by introducing a new estimator of the asymptotic variance based on the so called backward weights. The resulting estimator is weakly consistent and trades computational cost for more stability and reduced variance. We also propose a more computationally efficient estimator inspired by the $\paris$ algorithm of \cite{paris}. As an application, particle smoothing is considered and an estimator of the asymptotic variance of the Forward Filtering Backward Smoothing estimator applied to additive functionals is provided.
\end{abstract}

\maketitle

{\hypersetup{linkcolor=black}
\tableofcontents
}
\section{Introduction}
Sequential Monte Carlo (SMC) methods offer a flexible framework for the approximation of posterior distributions in the context of Bayesian inference, for instance in \emph{Hidden Markov Models} (HMM).  These models presuppose that the observations are defined using an unobserved process assumed to be a Markov chain. In such a setting, we are particularly interested in estimating the law of a hidden state given all past observations referred to as the filtering distribution and the laws of sequences of states given past and future observations, referred to as smoothing distributions. These distributions can be approximated by weighted empirical measures associated with random samples, usually known  as {\em particles}. All SMC methods are based on successive importance sampling and resampling steps. When a new observation is available, new particles are sampled 
according to an importance distribution and then they are 
weighted to match the target distribution. Finally, through a resampling scheme,  particles with large weights are duplicated while low weighted particles are discarded.
This general procedure has been used in a wide range of applications such as signal processing, target tracking, econometrics, biology, see  \cite{infhidden,douc2014nonlinear,Chopin_2020} and the references therein. 

 The quantification of the Monte Carlo error of SMC estimators is a major challenge. For a variety of SMC methods such as the \emph{bootstrap filter} \cite{gordon1993novel} or the \emph{Forward Filtering Backward Smoothing} (FFBS) \cite{Tanizaki-smoothing} algorithms, \emph{Central Limit Theorems} (CLT) with theoretical expressions of the asymptotic variance (in the number of particles) have been derived \cite{del1999central, chopin2004central,kunsch2005recursive, doucmoulines}. However, these expressions are not computable in practice and give rise to the natural question of their estimation. 
Since this problem appears in an online context, a critical constraint is that the samples
produced by the original SMC algorithm should be recycled to compute such estimates.

This problem  has received some satisfactory  solutions 
recently: a simple estimator
of the asymptotic variance when multinomial resampling is used has
been proposed in \cite{chan}
for the bootstrap filter algorithm
and has since been refined in \cite{leewhiteley,olssonvar, asmc}.
The computation of the associated asymptotic variance estimator at time $t$ is based on tracing the genealogy of each particle down to time $0$. Although it has been shown to be consistent as the number of particles $N$ grows to infinity,
it is particularly prone to instability when $t$ is large, as the successive resampling steps lead to the well known path degeneracy issue.
After a few time steps, particles are likely to share the same ancestor at time $0$ which in turn makes the asymptotic variance estimates collapse. To overcome this degeneracy issue, \cite{olssonvar} proposed to only trace a part of the genealogy of each particle according to a fixed-lag parameter following the fixed-lag smoothing approach introduced in \cite{olsson2008sequential}. As long as the number of particles is balanced with the chosen lag, the bias introduced by considering only the most recent ancestors of each particle can be controlled as shown in \cite{olssonvar}. However, while this alternative estimator remains easily computable, 
choosing an optimal lag is a non trivial task which makes this approach hard to tune in practice.  We thus address the limitations of current asymptotic variance  estimators in this paper.

The first contribution of this paper is to propose a parameter free estimator of the asymptotic variance associated with the bootstrap particle filter with multinomial resampling that trades computational cost for stability and reduced variance. The construction of our
estimator starts from the observation that
the aforementioned degeneracy 
is similar to that of classical SMC-based smoothing algorithms \cite{doucet2000sequential, fearnhead2010sequential, poyiadjis2011particle} or Particle Markov Chain  Monte Carlo algorithms such as the
Particle Gibbs sampler \cite{andrieu2010particle}.
In both cases, a backward sampling step which aims at diversifying the particle trajectories has shown to be a reliable workaround that decreases the (theoretical) variance of the estimators at the expense of higher computational cost \cite{godsill2004monte,doucmoulines,paris,del2010backward,lindstenBS,chopin_PG}.
We thus aim at introducing such a mechanism in the estimation of the asymptotic variance.
The construction of our estimator 
relies on the analysis conducted in \cite{leewhiteley} in which it is shown that the estimator of \cite{chan} can be interpreted as a conditional expectation
 with respect to the indices that retrace the genealogy of the particles, 
 given all the particles and ancestors.
We  show that this construction still holds when the distribution of  the indices relies on the backward importance weights. The resulting estimator is computed by averaging 
auxiliary statistics that are very similar to those of the \emph{forward implementation} of the FFBS for additive functionals \cite{delmoral:doucet:singh:2010} and can be thus updated online. The time complexity per update of our estimator is of order $N^3$. 
Driven by the efficient implementation of the FFBS for additive functionals developed in \cite{paris}, we show that the computational cost of our estimator  can be reduced from $\bigo(N^3)$ to $\bigo(N^2)$ by means of additional Monte Carlo simulation while remaining as competitive in terms of bias and variance.

We next focus on the FFBS algorithm for the estimation of smoothing estimators.
Despite the fact that a CLT has been obtained for estimators based on the FFBS, no variance estimator has been proposed in the literature. We show that our previous construction enables us to 
fill this gap and we thus provide 
a consistent estimator in the case of additive functionals  which are
particularly critical, for instance in the Expectation Maximization framework. Again, this estimator
can be computed online and in the
particular case of 
marginal smoothing
its computational cost can be drastically 
reduced.

The paper is organized as follows. In Section \ref{sec:SMC} we briefly review the SMC framework and discuss the current estimators of the asymptotic variance proposed in \cite{chan,leewhiteley,olssonvar}. In Section \ref{sec:variance_estimation}, we introduce
our estimator based on the backward weights, propose an online implementation and establish its asymptotic properties. In Section \ref{sec:varFFBS}, we extend our derivations to the FFBS algorithm and provide a consistent asymptotic variance estimator. We finally
validate our results with numerical experiments in Section \ref{sec:experiments}. Notably, we show empirically that our novel estimator for the filter has a favourable dependence on the time horizon $t$ in comparison with the existing estimators. All additional proofs and  discussions can be found in the \suppl.

\section{Notations}
\label{sec:notations}
For any measurable space $(\mathsf{E}, \mathcal{E})$, we denote by $\mathbb{F}(\mathcal{E})$ the set of $\mathbb{R}$ valued, $\mathcal{E}$-measurable functions,  by $\mathbb{F}_b(\mathcal{E})$ the subset of $\mathbb{F}(\mathcal{E})$ of bounded functions on $\mathsf{E}$, and by $\mathbb{M}(\mathsf{E})$ the set of measures on $\mathsf{E}$. For any $\mu \in \mathbb{M}(\mathrm{E})$ and $h \in \mathbb{F}(\mathcal{E})$, we write 
\[
\mu(h) \eqdef \int_{\mathsf{E}} h(x) \mu(\rmd x) \eqsp.
\]
 For any transition kernel $M$ from $(\mathsf{E}, \mathcal{E})$ to another measurable space $(\mathsf{G}, \mathcal{G})$, define 
\[
M[h](x) \eqdef \int_{\mathsf{G}} M(x, \rmd y) h(y) \eqsp, \quad \forall h \in \mathbb{F}(\mathcal{G})\eqsp, \quad \forall x \in \mathsf{E} \eqsp.
\]
and write $\mu M$ the measure defined on $\mathsf{G}$ by
\[
\mu M(A) \eqdef \int_{\mathsf{E}} \mu(\rmd x) M(x, A)\eqsp, \quad\forall A \in \mathcal{G} \eqsp.
\]
If $M_1$ is a transition kernel from $(\mathsf{E}, \mathcal{E})$ to $(\mathsf{G}, \mathcal{G})$ and $M_2$ a transition kernel from $(\mathsf{G}, \mathcal{G})$ to a measurable space  $(\mathsf{H}, \mathcal{H})$, then $M_1 M_2$ is the transition kernel from $(\mathsf{E}, \mathcal{E})$ to $(\mathsf{H}, \mathcal{H})$ defined by
\[ 
M_1M_2(x,A) = \int_{\mathsf{G}} M_1(x, \rmd y) M_2(y, A)\eqsp,\quad \forall x \in \mathsf{E}\eqsp, \quad \forall A \in \mathcal{H} \eqsp.
\]
In addition, $M_1 \otimes M_2$ is the transition kernel from $(\mathsf{E}, \mathcal{E})$ to $(\mathsf{G} \times \mathsf{H}, \mathcal{G} \otimes \mathcal{H})$ defined by
\[
M_1 \otimes M_2 (x, A ) \eqdef \int \1_{A}(y,z)M_1(x,\rmd y) M_2(y, \rmd z)\eqsp, \quad \forall x \in \mathsf{E}, \quad A \in \mathcal{G} \otimes \mathcal{H} \eqsp.
\]
In particular, for all $N\geq 1$,we will write $M^{\otimes N}$ for $\bigotimes_{i = 1}^N M$. For two $\mathcal{E}$-measurable functions $f,g$ the tensor product is defined as
\[
f \otimes g : \mathsf{E} ^2 \ni (x,y) \mapsto f(x)g(y) \eqsp.
\]
The sets $\N$ and $\N^{*}$ are respectively the sets of natural numbers and positive natural numbers. The $\mathbf{L}_q$ norm of a random variable $X$ is $\| X \|_q \eqdef \pE \big[ |X|^q \big]^{1/q}$. The supremum norm of $f \in \mathbb{F}_b (\mathsf{E})$ is denoted by $|f |_\infty$. The unit function $\boldone$ is such that $\boldone(x) = 1$ for all $x \in \mathsf{E}$.
 The transpose of a matrix $A$ is denoted $A^\top$. If $A, B \in \R^{M \times N}$ are two matrices, then the Hadamard product $A \odot B$ is the element-wise product, i.e for all $1\leq i \leq M$ and $1\leq j \leq N$, $(A \odot B) _{i,j} = a_{i,j} b_{i,j}$.
If $A \in \R^{N \times N}$, then $\mathrm{Diag}(A)$ is the $N \times N$ diagonal matrix such that for all $1\leq i\leq N$, $\mathrm{Diag}(A) _{i,i} = A_{i,i}$ and if $\bm x \in \R^{N \times 1}$ then $\mathrm{Diag}(\bm{x})$ is the $N \times N$ diagonal matrix such that $\mathrm{Diag}(\bm{x})_{i,i} = \bm{x} _i$.  For $(a, b) \in \N^2$, $[a:b] \eqdef \N \cap [a,b]$ and  $[b] \eqdef [1:b]$. If $f$ is a mapping from $[N]^2$ to $\R$, we denote by $\bm{f}$ the associated $N \times N$ matrix such that $\bm{f}_{i,j} = f(i,j)$. Finally, we adopt the following conventions. Given some set $\{ \xi^i _s\} _{s \in [0:t], i \in [N]}$, we write $\particle{1:N}{t} \eqdef (\particle{1}{t}, \cdots, \particle{N}{t})$, $\particletraj{k}{t} \eqdef (\particle{k_0}{0}, \cdots, \particle{k_t}{t})$, $\particlecloud{t} \eqdef \big\{\particletraj{k}{t} \big\} _{k_{0:t} \in [N]^{t+1}}$.

\section{Sequential Monte Carlo}
\label{sec:SMC}
In this section, we first review the bootstrap particle filter methodology and recall the main asymptotic results associated
with the estimates produced by this algorithm. A state of the art and a discussion of the estimation of the asymptotic variances 
related to this algorithm are also presented. 
\subsection{Definitions}
Let $(\mathsf{X}, \sigmaX)$ be a general measurable space. Let $\transition{0}$ and $(\transition{t})_{t \in \mathbb{N}}$ be a probability measure on $(\Xset, \mathcal{X})$ and a sequence of Markov transition kernels on $\Xset\times \mathcal{X}$, respectively. Consider also a family $(g_t)_{t \in \mathbb{N}}$ of non-negative $\sigmaX$-measurable functions, referred to as potentials. Throughout this paper, we make the following assumptions on $\{\transition{t} \}_{t \in \N}$ and $\{ g_t \}_{t \in \N}$.
\begin{hypA}
\label{assp:A}
The probability measure $\transition{0}$ admits $\transitiondens{0}$ as probability density with respect to some reference measure $\nu \in \mathbb{M}(\mathcal{X})$.  For all $t \in \N$ and $x_t \in \Xset$, $\transition{t+1}(x_t,.)$  admits $\transitiondens{t+1}(x_t,.)$  as probability density with respect to $\nu$.
\end{hypA}
\begin{hypA}
\label{assp:B}
There exists a constant $\boundg > 0$ such that for all $t \in \N$ and $x \in \Xset$,  $0 < g_t (x) \leq \boundg \eqsp.$
\end{hypA}
Define the sequence of unnormalized transition kernels $(\Q{t+1})_{t \in \N}$ where, for all $t \in \N$, $x_t \in \Xset$ and $A\in \sigmaX$,
\begin{equation*}
    \label{def:Qdef}
    \Q{t+1}(x_t, A) \eqdef g_t(x_t) \transition{t+1}(x_t, A) \eqsp,
\end{equation*}
and, for any $s,t \in \N^2$,
\begin{equation*}
    \Q{s:t} \eqdef \begin{cases} \Q{s} \otimes \cdots \otimes \Q{t} & \text { if } s \leq t \eqsp,
         \\ \text { Id } & \text { otherwise} \eqsp .\end{cases}
\end{equation*}
Let $\Qmarg{s:t}$ denote its marginal with respect to the variable $x_t$, i.e. for all $x_{s-1}\in\Xset$ and all measurable set $A$,
\[ 
    \Qmarg{s:t}(x_{s-1},A) := \int_{\Xset^{t-s+1}} \Q{s:t}(x_{s-1}, \rmd x_s, \dotsc, \rmd x_t)\1_A(x_t).
\]
Define recursively the sequence of measures $(\joint{0:t}) _{t \in \N}$ by
\begin{equation}
\label{eq:predjoint}
    \joint{0}(\rmd x_0) \eqdef \transition{0}(\rmd x_0) \eqsp, \quad 
    \joint{0:t}(\rmd x_{0:t}) \eqdef \joint{0:t-1}(\rmd x_{0:t-1}) \Q{t}(x_{t-1}, \rmd x_t)
    \eqsp,
\end{equation}
and let $\joint{t}(A) = \int_{\Xset^{t+1}} \joint{0:t}(\rmd x_{0:t})\1_A(x_t)$ for all measurable set $A$. Sequential Monte Carlo algorithms aim at solving recursively the filtering problem, i.e. at estimating the sequence of probability measures defined as
\begin{equation}
    \label{def:posteriors}
    \pred{t}(\rmd x_t) \eqdef \joint{t}^{-1}(\boldone) \joint{t}(\rmd x_t), \quad \filter{t}(\rmd x_t) \eqdef g_t(x_t) \pred{t}(\rmd x_t) / \pred{t}(g_t) \eqsp,
\end{equation}
respectively called the \textit{predictive} and \textit{filtering} measures. Note that $\pred{t}$ can be computed recursively using 
\begin{equation}
\label{eq:recurseta}
    \pred{t}(\rmd x_t) = \int \filter{t-1}(\rmd x_{t-1}) \transition{t}(x_{t-1}, \rmd x_t) \eqsp.
\end{equation}
We motivate these definitions with the following example.
\begin{example}
\label{ex:hmm}
\textit{Hidden Markov models} consist of an unobserved state process $\{X_t\}_{t \in \N}$ and observations $\{Y_t\}_{t \in \N}$. 
They respectively evolve in two general measurable spaces $(\Xset, \sigmaX)$ and $(\Yset, \sigmaY)$. 
It is assumed that $\{ X_t\}_{t \in \N}$ is a Markov chain with transition kernels $(\transition{t+1})_{t \in \N} $ and 
initial distribution $\transition{0}$. Given the states $\{X_t\}_{t \in \N}$, the observations $\{Y_t\}_{t \in \N}$ are independent
 and for all $t \in \N$, the conditional distribution of the observation $Y_t$ only depends on the current state $X_t$. This  distribution is written $G_t(X_t, .)$ and admits the potential $g_t(x_t,.)$ as density (the dependency in $Y_t$ is made implicit and we drop the second argument). Given an observation record $Y_{0:t} $, the predictive and filtering 
  distributions \eqref{def:posteriors} are then the distributions of $X_t$ given $Y_{0:t-1}$ and $X_t$ given $Y_{0:t}$ respectively.

These two distributions are of considerable interest in Bayesian filtering as they enable the estimation of the hidden states through the observed data record. Unfortunately only in a few cases, such as discrete state spaces or linear and Gaussian HMM, can they be obtained in closed form, see \cite{infhidden,Chopin_2020} for a complete overview.
\end{example}
\subsection{Particle filter}
\label{sec:PF}
We now illustrate how to obtain empirical estimates of $\pred{t}$ and $\filter{t}$ in an online manner through Monte Carlo simulation.
Assume that at time $t$ the empirical measure $\pred{t}^N (\rmd x_t) \eqdef N^{-1} \sum_{i = 1}^N \delta_{\particle{i}{t}} (\rmd x_t)$ based on random samples $\{\particle{i}{t}\}_{1\leq i\leq N}$ approximates $\pred{t}(\rmd x_t)$. Plugging $\pred{t}^N $ in \eqref{def:posteriors} provides an approximation of $\filter{t}(\rmd x_t)$, 
\begin{equation*}
    \filter{t}^N (\rmd x_t) \eqdef \sum_{i = 1}^N \normweight{i}{t} \delta_{\particle{i}{t}} (\rmd x_t) \eqsp,
\end{equation*}
where 
 $
 \normweight{i}{t} \eqdef \Omega^{-1} _t \weight{i}{t}$, $\weight{i}{t} \eqdef g_t(\particle{i}{t})$ and $\Omega_{t} \eqdef \sum_{i = 1}^N \weight{i}{t}\eqsp.
 $
Replacing $\filter{t}$ by $\filter{t}^N$ in \eqref{eq:recurseta}, we obtain the mixture $ \filter{t}^N \transition{t+1}$ which allows to construct $\pred{t+1}^N$ by drawing $N$ samples from it.  We first sample for all $1\leq i \leq N$ an ancestor index $A^i_t \sim \mathrm{Categorical}(\normweight{1:N}{t})$, and then sample $\particle{i}{t+1} \sim \transition{t+1}(\particle{A^i_t}{t}, \cdot)$. The algorithm is initialized with the approximation of $\pred{0} = \transition{0}$,  $\eta^N _0 \eqdef N^{-1} \sum_{i = 1}^N \delta_{\particle{i}{0}}(\rmd x_t)$ where $\particle{1:N}{0} \sim \transition{0} ^{\otimes N}$ and coincides with the \emph{bootstrap} algorithm with multinomial resampling \cite{gordon1993novel}. 
 Note that this mechanism has been extended in many directions in the past decades \cite{pitt1999filtering, douc2005comparison, Chopin_2020}.
\\Alongside $\pred{t}^N$ and $\filter{t}^N$, the particle approximation of the unnormalized marginal $\joint{t}(\rmd x_t)$ is given by
\begin{equation}
    \label{eq:jointdef}
\joint{t}^N (\rmd x_t) \eqdef \bigg\{ \prod_{s = 0}^{t-1} N^{-1} \Omega_s \bigg\} \pred{t}^N (\rmd x_t) \eqsp, \quad \forall t > 0 \eqsp,
\end{equation}
and $\joint{0}^N(\rmd x_0) = \pred{0}^N(\rmd x_0) \eqsp.$ In particular, $\joint{t}^N(h)$ is unbiased for any measurable function $h$ \cite{del2004feynman}.

In the remainder of this paper, we denote by $\F{t}$ the $\sigma$-field containing all the particles and ancestors up to time $t$, i.e. $\F{t} \eqdef \sigma \left( \particlecloud{t}, \bm{A}^{1:N} _{0:t-1}\right)$.
\subsection{Asymptotic variance estimation in particle filters}
\label{sec:relatedwork} 
The particle filter described above yields consistent estimators, see for instance \cite{infhidden,douc2014nonlinear,Chopin_2020, liu2001sequential, del2004feynman} 
and references therein for a complete overview. Indeed, for a test function $h \in \bounded{}$
and under assumption 
$\A{assp:B}{}$, the SMC estimators satisfy a Strong Law of Large
 Numbers when the number of particles $N$ goes to infinity, i.e.
\begin{equation}
    \label{eq:asconv}
    \joint{t}^N (h) \aslim \joint{t}(h), \quad \pred{t}^N (h) \aslim \pred{t}(h) \eqsp, \quad \filter{t}^N (h) \aslim \filter{t}(h)\eqsp.
\end{equation}
Under the same assumptions, CLTs for $\joint{t}^N(h)$, $\pred{t}^N(h)$ and $\filter{t}^N(h)$ are also available \cite{del1999central, chopin2004central}: 
\begin{align}
    \label{eq:CLTs}
    \begin{cases}
    \sqrt{N}\big( \joint{t}^N(h) - \joint{t}(h) \big) & \underset{N \to \infty}{\dlim} \gauss\big(0,\jointasymptvar{t}{h} \big)\eqsp,\\
    \sqrt{N}\big( \pred{t}^N(h) - \pred{t}(h) \big) & \underset{N \to \infty}{\dlim} \gauss\big(0,\predasymptvar{t}{h} \big)\eqsp, \\
    \sqrt{N}\big( \filter{t}^N(h) - \filter{t}(h) \big) & \underset{N \to \infty}{\dlim} \gauss\big(0,\asymptvar{\phi, t}(h)\big)\eqsp,
    \end{cases}
\end{align}
where $\dlim$ denotes weak convergence and
\begin{align}
\label{eq:asymptvarjoint}
\asymptvar{\gamma,t}(h) & =  \sum_{s = 0}^{t} \big\{ \joint{s} (\mathbf{1}) \joint{s}\big(\Qmarg{s+1:t}[h]^2\big) - \joint{t}(h)^2 \big\} \eqsp, \\
\label{eq:asymptvarpred}
\asymptvar{\eta, t}(h) & = \sum_{s = 0}^t \frac{\joint{s}(\mathbf{1})\joint{s}\big( \Qmarg{s+1:t}[h - \pred{t}(h) ]^2 \big)}{\joint{t}(\mathbf{1})^2} \eqsp, \\
\label{eq:asymptvarfilter}
\asymptvar{\phi, t}(h) & = \sum_{s = 0}^t \frac{\joint{s}(\mathbf{1})\joint{s}\big( \Qmarg{s+1:t}[g_t\{h - \filter{t}(h)\} ]^2 \big)}{\joint{t+1}(\mathbf{1})^2} \eqsp.
\end{align}
An intuitive derivation of \eqref{eq:asymptvarjoint} is proposed in Section~\ref*{apdx:asymptvar} of the \suppl.
The authors of \cite{chan} propose to estimate \eqref{eq:asymptvarjoint} online using the samples produced by the particle filter described above.
 Their estimator is based on the genealogy of the particle system induced by the successive resampling steps
  of the particle filter.
  From the indices $A^i _{t}$,
  it is possible to trace back the ancestors of each particle and deduce the corresponding ancestor at time  $t=0$.
  More interestingly, these ancestors 
  can be computed in a forward way 
  by introducing the Eve indices $\eve^i _{t,0}$. For all $i \in [1:N]$, $\eve^i _{t,0}$ describes the index of the ancestor at time $0$ of particle $\particle{i}{t}$ and can be 
  computed from 
  \begin{equation}
\label{eq:evedef}
\eve^i _{t,0} = \eve^{A^i _{t-1}} _{t-1,0} \1_{ t > 0} + i \1_{t = 0}\eqsp.
\end{equation}
The asymptotic variance estimator of $\pred{t}^N(h)$ obtained in \cite{chan} reads (see Section~\ref*{sec:altexpr} of the supplementary for a proof):
\begin{equation}
\label{eq:CLE}
    \predasymptvarestim{t}{h} \eqdef  - N^{-1} \sum_{i,j \in [N]^2} \1_{\eve^i _{t,0} \neq \eve^j _{t,0} } \big\{ h(\particle{i}{t}) - \pred{t}^N (h) \big\} \big\{ h(\particle{j}{t}) - \pred{t}^N (h) \big\}  \eqsp.
\end{equation}
 We sometimes refer to $\predasymptvarestim{t}{h}$
 as the CLE (Chan \& Lai Estimator).
Note that  it can be computed online in a remarkably simple way since the Eve indices \eqref{eq:evedef} are computed recursively. However, the counterpart of its computational simplicity is that it degenerates as soon as the ancestral paths coalesce. Indeed, it is widely known in the SMC
 literature that all lineages eventually end up with the same ancestor when $t$ is large enough with respect to the number of samples $N$ (see e.g. \cite[Section 2.2]{fearnhead2010sequential} for a more detailed explanation).
   This means that for a fixed $N$, as $t$ grows and $s \ll t$, $E^i _{s,0} = E^j _{s,0}$ for all $(i,j) \in \N^2$ and $\predasymptvarestim{t}{h} = 0$ for any test function $h$. 
    
    The degeneracy problem concerning \eqref{eq:CLE} is partially addressed in \cite{olssonvar} by truncating the genealogy of the particle system. Denoting 
     $\lambda \in [t]$ the lag and $\eve^i _{t, t - \lambda}$ the ancestor of $\particle{i}{t}$ at time $t - \lambda$, their estimator reads
    \begin{equation}
        \label{eq:truncCLE}
        \asymptvarestim{\eta, t}{\lambda}(h) \eqdef - N^{-1} \sum_{i,j \in [N]^2} \1_{\eve^i _{t, t - \lambda} \neq \eve^j _{t, t - \lambda}} \big\{ h(\particle{i}{t}) - \pred{t}^N (h) \big\} \big\{ h(\particle{j}{t}) - \pred{t}^N(h) \big\} \eqsp.
    \end{equation}
    In the regime where \eqref{eq:CLE} degenerates, \eqref{eq:truncCLE} can be made stable provided that the lag $\lambda$ is chosen such that there is little asymptotic bias. However, besides the strong mixing case for which the authors propose a heuristic, the practical choice of such a $\lambda$, although crucial, is a non trivial task. 
 
In \cite{leewhiteley}, the CLE is revisited using different techniques based on \cite{cerou2011, andrieu2010particle, andrieu2018}. These tools enable them to derive  a weakly consistent term
 by term estimator of \eqref{eq:asymptvarjoint} based on
 the unbiased estimation of $\joint{t}(h)^2$ and of each $\joint{s} (\mathbf{1}) \joint{s}\big(\Qmarg{s+1:t}[h]^2\big)$, 
 for all $s \in [0:t]$.
 The construction of this second estimator is appealing and insightful in that it helps identifying the deep root of the degeneracy in the CLE. They indeed show that \eqref{eq:CLE} and  each $\joint{s}(\boldone) \joint{s}\big( \Qmarg{s+1:t}[h]^2 \big)$ can be interpreted as a conditional expectation with respect to particle indices that retrace the ancestral paths.
 Indeed, by introducing discrete
 random variables $K^1_{0:t}$ and 
 $K^2_{0:t}$ such that conditionally on $\F{t}$, $K^1 _t$ and $K^2 _t$ are distributed uniformly on $[N]$ and such that for any $s \in [0:t-1]$, 
\begin{equation*}
    K^1 _{s}  = A^{K^1 _{s+1}} _s,  \quad 
    K^2 _s = \1_{K^1 _{s+1} \neq K^2 _{s+1}} A^{K^2 _{s+1}} _s + \1_{K^1 _{s+1} = K^2 _{s+1}} C_s \eqsp,  
\end{equation*}
where $C_s \sim \mathrm{Categorical}(  \normweight{1:N}{s})$, then for example
 \[
\predasymptvarestim{t}{h} = - N \pE \left[ \prod_{s = 0}^t \1_{K^1 _s \neq K^2 _s} \{ h(\particle{K^1 _t}{t}) 
- \pred{t}^N (h) \} \{ h(\particle{K^2 _t}{t}) - \pred{t}^N (h) \} \bigg| \F{t} \right] \eqsp.
\]

An intuitive extension is to replace the deterministic assignments $K^1 _s = A^{K^1 _{s+1}} _s$ and $K^2 _s = A^{K^2 _{s+1}} _s$ by random ones based on backward sampling. Essentially, the backward kernel 
samples at time $s$ an index $i$ with probability proportional to $\weight{i}{s} \transitiondens{s+1}(\particle{i}{s}, \particle{K^1 _{s+1}}{s+1})$ (resp. $\weight{i}{s} \transitiondens{s+1}(\particle{i}{s}, \particle{K^2 _{s+1}}{s+1})$) and thus allows to consider relevant trajectories which are not necessarily ancestral trajectories.
\section{Variance estimation with backward sampling}
\label{sec:variance_estimation}
In this section we present three variance estimators for the bootstrap particle filter. In Section \ref{subsec:tbt}, we lay out our methodology and derive a term by term variance estimator; its computation
is detailed
in Section \ref{subsec:compute}. In Sections \ref{subsec:disj} and \ref{subsec:paris}, we provide two additional estimators that have a lower computational cost. 
All estimators and justifications are provided for the distribution \eqref{eq:jointdef}. We give the expressions for the variance estimators of the predictor and filter and provide their justification in Section~\ref*{subsec:predfiltervariance} of the \suppl.
\subsection{Term by term variance estimator}
\label{subsec:tbt}
For any $t \in \N$, let $\Bset_t \eqdef \{0,1\}^{t+1}$. Denote by $\zero$ the null vector in $\Bset_t$ and $e_s$ the vector
 with 1 at position $s$ and $0$ elsewhere. 
Let $(X_s, X' _s)_{s \in [0:t]}$ be a bivariate Markov chain in $(\Xset^2, \sigmaX^{\otimes 2})$ and depending on $b \in \Bset_t$  with initial distribution $\bitransition{b_0}{0}$ and transition kernels $\bitransition{b_{t}}{t}$, $t\geq 1$, where
\begin{equation}
\label{eq:bivariatechain}
\begin{alignedat}{2}
    \bitransition{b_0}{0} (\rmd x_0,\rmd x'_0 ) & \eqdef \transition{0}(\rmd x_0) \{ \1_{b_0 = 0} \transition{0}(\rmd x' _0) + \1_{b_0 = 1} \delta_{x_0}(\rmd x'_0)\} \eqsp,\\
    \bitransition{b_{t}}{t}(x, x'; \rmd z, \rmd z') & \eqdef \transition{t}(x, \rmd z)  \{ \1_{b_{t} = 0} \transition{t}(x', \rmd z') + \1_{b _{t} = 1} \delta_{z}(\rmd z') \} \eqsp, \quad \forall t \geq 1 \eqsp.
    \end{alignedat}
\end{equation}
Define also for any $b \in \Bset_t$ the measure $\mumeasure{b,t}$ by $\mumeasure{b,0}(\rmd x_{0}, \rmd x'_{0} ) = \bitransition{b_0}{0}(\rmd x_0, \rmd x'_0 )$ and for $t\geq 1$:
\begin{equation}
    \label{eq:defmu1}
    \mumeasure{b,t}\big( \rmd x_{0:t}, \rmd x'_{0:t} \big) \eqdef \bitransition{b_0}{0}\big(\rmd x_0, \rmd x'_0 \big) \prod_{s = 0}^{t-1} g_s ^{\otimes 2}(x_s, x'_s)  \prod_{s = 1}^t \bitransition{b_s}{s}\big(x_{s-1}, x'_{s-1}; \rmd x_s, \rmd x'_s \big) \eqsp.
\end{equation}
The measure $\mumeasure{b,t}$ is the joint distribution \eqref{eq:predjoint} of the Feynman-Kac model defined by the initial distribution $\bitransition{b_0}{0}$, the transition kernels $\{\bitransition{b_s}{s}\}_{s\in [1:t]}$ and by the potential functions $\{g_s ^{\otimes 2}\}_{s\in [0:t-1]}$.
 Remark that for any $h \in \measurable{}$, writing $h_t: x_{0:t} \mapsto h(x_t)$, we have that 
$\mumeasure{\zero,t}(h_t ^{\otimes 2}) = \joint{t}(h)^2$ and 
\begin{equation}
    \label{eq:identityQes}
    \mumeasure{e_s, t}(h_t ^{\otimes 2}) = \joint{s}(\mathbf{1}) \joint{s}(\Qmarg{s+1:t}[h]^2) \eqsp.
\end{equation}
A generalization of \eqref{eq:identityQes} is proved in Proposition~\ref{prop:identityQes}. Consequently, for $h \in \bounded{}$ $\asymptvar{\gamma, t}(h)$ in \eqref{eq:asymptvarjoint} can be rewritten as 
\begin{equation}
    \label{eq:scdestimator}
     \asymptvar{\gamma, t}(h) = \sum_{s = 0}^t \big\{ \mumeasure{e_s, t}(h_t ^{\otimes 2}) - \mumeasure{\zero, t}(h_t ^{\otimes 2}) \big\}  \eqsp,
    \end{equation}
where $h_t : x_{0:t} \mapsto h(x_t)$.
Following this observation,
for a given $b$,
an estimator of $\mumeasure{b,t}(h_t)$  could be obtained with a single run of a
bootstrap particle filter in augmented dimension (i.e. relying on the
the bivariate transition $\bitransition{b_t}{t}$ at time $t$ and on the weighing of the associated particles with $g^{\otimes 2} _t$). As a direct 
extension of \eqref{eq:jointdef},
this estimator would be unbiased and
as a byproduct, we would get an unbiased estimator of \eqref{eq:scdestimator}.
However, this procedure is not in line with our initial objective in the sense
that we aim at estimating \eqref{eq:scdestimator} with the  particles and indices already available.


In order to motivate and introduce our approach, let us consider the static situation where $t = 0$, $b_0 = 0$ and let $(h,f) \in \bounded{}^2$. Then,
\begin{equation}
    \label{eq:unbiasedest}
    \frac{1}{N(N-1)} \sum_{i, j \in [N]^2} \1_{i \neq j} h(\particle{i}{0})f(\particle{j}{0}), \quad \particle{1:N}{0} \iid \transition{0}
\end{equation}
is an unbiased and almost sure convergent estimator of $\mumeasure{0, 0}(h \otimes f) = \transition{0}^{\otimes 2}(h \otimes f)$ and only relies on i.i.d. samples from $\transition{0}$ rather than $\transition{0}^{\otimes 2}$. Note that for $h = f$, we thus get an unbiased estimator of $\gamma_{0}(h)^2$. If $b_0 = 1$, then $N^{-1} \sum_{i = 1}^N h(\particle{i}{0})f(\particle{i}{0})$ is an unbiased and consistent estimator of $\mumeasure{1,0}(h \otimes f)=\transition{0}(hf)$.

Taking advantage of the fact that the particles at time $t$ are i.i.d. conditionally on $\F{t-1}$ as we use multinomial resampling for the particle filter, we can carry these simple observations to the sequential case in two directions as we now detail. Define for any $t \in \N^{*}$ the functional version of the backward weights:
\begin{equation}
    \label{def:functionalweight}
    \beta^N _t(x,y) \eqdef \frac{g_{t-1}(y) \transitiondens{t}(y, x)}{\sum_{\ell = 1}^N \omega^\ell _{t-1} \transitiondens{t}(\particle{\ell}{t-1}, x)}, \quad \forall (x,y) \in \Xset^2 \eqsp.
\end{equation}
Assume that $t = 1$ and define for any $(k^1 _0, k^2 _0) \in [N]^2$ the following random variables that involve the backward weights
\begin{align}
    \label{eq:EBSdisj}
    \mathcal{E}^{\BS} _0 (k^1 _0, k^2 _0) & \eqdef \frac{\Omega_0 ^2}{N(N-1)} \sum_{k^{1:2} _1 \in [N]^2 } \1_{k^1 _1 \neq k^2 _1} \beta^N _1(\particle{k^1 _1}{1}, \particle{k^1 _0}{0}) \beta^N _1(\particle{k^2 _1}{1}, \particle{k^2 _0}{0}) h(\particle{k^1 _1}{1}) f(\particle{k^2 _1}{1}) \eqsp, \\
    \label{eq:EBSinters}
    \mathcal{E}^{\BS} _1 (k^1 _0, k^2 _0) & \eqdef \frac{\Omega_0 ^2}{N} \sum_{k^{1:2} _1 \in [N]^2 } \1_{k^1 _1 = k^2 _1} \beta^N _1(\particle{k^1 _1}{1}, \particle{k^1 _0}{0}) \normweight{k^2 _0}{0} h(\particle{k^1 _1}{1}) f(\particle{k^2 _1}{1}) \eqsp,
\end{align}
and also the following which involve the ancestors
    \begin{align}
    \label{eq:EGTdisj}
    \mathcal{E}^{\GT} _0(k^1 _0, k^2 _0) & \eqdef \frac{\Omega_0 ^2}{N(N-1)} \sum_{k^{1:2} _1 \in [N]^2 } \1_{A^{k^1 _1} _0 = k^1 _0, A^{k^2 _1} _0 = k^2 _0, k^1 _1 \neq k^2 _1}  h(\particle{k^1 _1}{1}) f(\particle{k^2 _1}{1}) \eqsp, \\
    \label{eq:EGTinters}
    \mathcal{E}^{\GT} _1 (k^1 _0, k^2 _0) & \eqdef \frac{\Omega_0 ^2}{N} \sum_{k^{1:2} _1 \in [N]^2 } \1_{k^1 _0 = A^{k^1 _1} _0, k^1 _1 = k^2 _1} \normweight{k^2 _0}{0} h(\particle{k^1 _1}{1}) f(\particle{k^2 _1}{1}) \eqsp,
\end{align}
Here, $\BS$ and $\GT$ correspond to backward sampling and genealogy tracing, respectively. Consider also Lemma~\ref{lem:BSGTidentity} which states a crucial identity involving the backward weights.
\begin{lemma}
   \label{lem:BSGTidentity}
    For all $t \in \N^{*}$, $(x,y) \in \Xset^2$, 
    \begin{equation}
        \label{lem:identity}
        \beta^N _{t}(x,y) \filter{t-1}^N \transition{t} (\rmd x) = \frac{g_{t-1}(y)}{\Omega_{t-1}}  \transition{t}(y, \rmd x) \eqsp.
    \end{equation}
    and for any $(k _{t-1}, k _t) \in [N]^2$ and $h \in \measurable{}$, 
    \begin{equation}
        \pE \big[ \beta^N _t(\particle{k _{t}}{t}, \particle{k _{t-1}}{t-1}) h(\particle{k _t}{t}) \big| \F{t-1} \big] = \pE \big[ \1_{k _{t-1} = A^{k _t} _{t-1}} h(\particle{k _t}{t}) \big| \F{t-1} \big] =  \normweight{k _{t-1}}{t-1} \transition{t}[h](\particle{k _{t-1}}{t-1}) \eqsp.
    \end{equation}
\end{lemma} 
The proof is postponed to Section~\ref*{proof:BSGTidentity} in the \suppl. Applying Lemma \ref{lem:BSGTidentity} to \eqref{eq:EBSdisj} and \eqref{eq:EGTdisj} and using that given $\F{0}$ the particles at $t = 1$ are i.i.d., 
we get 
\begin{equation*}
    \pE \big[ \mathcal{E}^\BS _0(k^1 _0, k^2 _0) \big| \F{0} \big] = \pE \big[ \mathcal{E}^\GT _0(k^1 _0, k^2 _0) \big| \F{0} \big] = g^{\otimes 2} _{0}(\particle{k^1 _0}{0}, \particle{k^2 _0}{0}) \bitransition{0}{1}[h \otimes f](\particle{k^1 _0}{0}, \particle{k^2 _0}{0}) \eqsp,
\end{equation*}
and
\begin{align*}
    \pE \bigg[\sum_{k^{1:2} _0 \in [N]^2 } \1_{k^1 _0 \neq k^2  _0} \mathcal{E}^\BS _0(k^1 _0, k^2 _0) \bigg]  & = \pE \bigg[  \sum_{k^{1:2} _0 \in [N]^2 } \1_{k^1 _0 \neq k^2 _0} \pE \big[ \mathcal{E}^\BS _0(k^1 _0, k^2 _0) \big| \F{0} \big] \bigg] \\
    & = \pE \bigg[ \sum_{k^{1:2} _0 \in [N]^2 } \1_{k^1 _0 \neq k^2 _0} \pE \big[ \mathcal{E}^\GT _0(k^1 _0, k^2 _0) \big| \F{0} \big] \bigg] \\
    &= N(N-1) \mumeasure{\zero, 1}(h \otimes f) \eqsp.
\end{align*}
Similarly, 
\begin{equation*}
    \pE \big[ \mathcal{E}^\BS _1(k^1 _0, k^2 _0) \big| \F{0} \big] = \pE \big[ \mathcal{E}^\GT _1(k^1 _0, k^2 _0) \big| \F{0} \big] = g^{\otimes 2} _{0}(\particle{k^1 _0}{0}, \particle{k^2 _0}{0}) \bitransition{1}{1}[h \otimes f](\particle{k^1 _0}{0}, \particle{k^2 _0}{0}) \eqsp;
\end{equation*}
and
\begin{equation*}
    \pE \bigg[\sum_{k^{1:2} _0 \in [N]^2 } \1_{k^1 _0 \neq k^2  _0} \mathcal{E}^\BS _1(k^1 _0, k^2 _0) \bigg] = N(N-1) \mumeasure{e_1, 1}(h \otimes f) \eqsp.
\end{equation*}
Therefore, it is possible to derive
unbiased estimators of $\mumeasure{\zero,1}(h \otimes f)$ and $\mumeasure{e_1,1}(h \otimes f)$ (but also $\mumeasure{e_0, 1}(h \otimes f)$ and $\mumeasure{(1,1), 1}(h \otimes f)$)
with a single run of the particle filter in two different ways: either by using the backward weights and \eqref{eq:EBSdisj}-\eqref{eq:EBSinters}, or by using directly the ancestry of the particles and \eqref{eq:EGTdisj}-\eqref{eq:EGTinters}. The asymptotic variance estimators proposed in \cite{chan,leewhiteley,olssonvar,asmc} are all based on the latter solution, while in this paper we instead focus on estimators based on the backward weights.


We now generalize the derivations performed in the case $t = 1$. Denote by $\Lambda_{1,t}$ and $\Lambda_{2,t}$ the discrete measures conditioned on $\F{t}$ and defined by
\begin{align}
    \label{def:biglambda}
    \Lambda_{1,t}( k_{0:t}) & \eqdef N^{-1} \prod_{s=1}^t \beta_s(k_s, k_{s-1}) \eqsp, \\
    \Lambda_{2,t}(k^1 _{0:t}; k^2 _{0:t}) & \eqdef N^{-1} \prod_{s = 1}^t \left\{\1_{k^2 _s = k^1 _s} \normweight{k^2 _{s-1}}{s-1} + \1_{k^2 _s \neq k^1 _s} \beta_s(k^2 _s,k^2 _{s-1}) \right\}\eqsp.
\end{align}
Specific choices of kernels $\{\beta_s\}_{s = 1} ^t$ are, for all $(k,\ell) \in [N]^2$,
\begin{equation*}
    \beta^{\GT} _s(k,\ell)  \eqdef \1_{\ell = A^k _{s-1}},\quad \beta^{\BS}_s(k,\ell)  \eqdef \frac{\weight{\ell}{s-1} m_{s}(\particle{\ell}{s-1}, \particle{k}{s})}{\sum_{j = 1}^N \weight{j}{s-1} m_{s}(\particle{j}{s-1}, \particle{k}{s})} = \beta^N _s(\particle{k}{s}, \particle{\ell}{s-1}) \eqsp,
\end{equation*}
where here again $\GT$ stands for \emph{genealogy tracing} and $\BS$ for \emph{backward sampling}. When the conditional distribution given $\F{t}$ is  
 $\Lambda_{1,t} ^\BS \otimes \Lambda_{2,t} ^\BS$ (resp. $\Lambda_{1,t} ^\GT \otimes \Lambda_{2,t} ^\GT$) we write  $\pE_\BS[\cdot | \F{t} ]$ (resp. $\pE_\GT[\cdot | \F{t} ]$). Define also for any $b \in \Bset_t$ the coalescence function:
 \begin{equation}
    \label{eq:Idef}
     \mathrm{I}_{b, s}: ([N]^{s+1})^2 \ni (k^1 _{0:s}, k^2 _{0:s}) \mapsto \prod_{\ell = 0}^{s} \{ \1_{k^1 _\ell = k^2 _\ell} \1_{b_\ell = 1}
      + \1_{k^1 _\ell \neq k^2 _\ell}\1_{b_\ell = 0} \} \eqsp, \quad \forall s \in [0:t] \eqsp,
 \end{equation} 
for any $h \in \measurable{2(t+1)}$  the random variable 
\begin{equation} 
    \label{eq:mu_estimator}
    \muest{\BS}{b,t} (h) \!\eqdef \! \prod_{s = 0 }^t N^{b_s} 
    \bigg(\frac{N}{N-1}\bigg)^{1 - b_s} \!\!\! \joint{t}^N(\mathbf{1})^2 
    \pE_\BS\big[ \intersect{b,t}{K^1 _{0:t}, K^2 _{0:t}} 
    h(\particletraj{K^1}{t}, \particletraj{K^2}{t}) | \F{t} \big] \eqsp,
\end{equation}
and denote by $\muest{\GT}{b,t}$ the counterpart where the expectation in the r.h.s. is $\pE _\GT$.
\begin{rem} 
    By \eqref{eq:Idef}, the random variable $\muest{\BS}{b,t}(h)$ remains defined for any $b \in \Bset_r$ with $r > t$ and $\muest{\BS}{b,t}(h) = \muest{\BS}{b_{0:t},t}(h)$ where $b_{0:t}$ is the truncation of $b$ to the $t+1$ first terms.
\end{rem}
Finally, define for any $h \in \measurable{}$, using $h_t : x_{0:t} \mapsto h(x_t)$,
\begin{equation}
    \label{eq:general_termbyterm}
    \tbtasymptvarestim{\gamma,t}{\BS}(h) \eqdef \sum_{s = 0}^t \big\{ \muest{\BS}{e_s,t}(h_t ^{\otimes 2}) - \muest{\BS}{\zero, t}(h_t ^{\otimes 2}) \big\} \eqsp.
\end{equation}
\begin{proposition}
\label{prop:mu_expression}
Let $t \in \N$. For any $b \in B_t$ and any $h \in \measurable{2(t+1)}$,
\begin{enumerate}[label=(\roman*)]
    \item \label{item:condexpect} $ \pE \big[ \muest{\BS}{b,t}(h) \big| \F{t-1} \big] = \muest{\BS}{b,t-1}\big( g^{\otimes 2} _{t-1} \bitransition{b_t}{t}[h] \big)$ for all $t \in \N^{*} \eqsp.$
    \item $\muest{\BS}{b,t}(h)$ is an unbiased estimator of $\mumeasure{b,t}(h) \eqsp.$
    \item If $h \in \measurable{}$, $\tbtasymptvarestim{\gamma,t}{\BS}(h)$ is an unbiased estimator of $\asymptvar{\gamma,t}(h)$.
\end{enumerate}
By convention, we used in (i) the following notation:
$$
g^{\otimes 2} _{t-1} \bitransition{b_t}{t}[h] : (x_{0:t-1}, x' _{0:t-1}) \mapsto g^{\otimes 2} _{t-1}(x_{t-1}, x' _{t-1}) \int h(x_{0:t}, x' _{0:t}) \bitransition{b_{t}}{t}(x_{t-1}, x' _{t-1}; \rmd x_t, \rmd x'_t).
$$
\end{proposition}
The proof is provided in Section~\ref*{proof:unbiasedest} of the \suppl. First, (ii) 
is a generalization of \cite[Lemma 2]{leewhiteley} which states that $\muest{\GT}{b,t}(h)$ is also an unbiased estimator $\mumeasure{b,t}(h)$.
Its proof, see \cite[Supplementary]{leewhiteley}, is based on a doubly conditional SMC argument \cite{andrieu2018} and while this scheme can be replicated to our estimator based on backward weights, we rather propose an alternative and elementary proof that also extends straightforwardly to $\GT$ and which is based on our previous discussion. From \eqref{eq:scdestimator}, (iii) is a direct consequence of (ii) and provides an estimator of \eqref{eq:scdestimator} based on a single particle run. 

Theorem~\ref{thm:conv} deals with the convergence of $\muest{\BS}{b,t}(h)$ for bounded $h$. The convergence in $\mathbf{L}_2$ is stated under $\A{assp:B}{assp:boundup}$ which are standard and the $1/\sqrt{N}$ convergence rate is obtained under the additional assumption $\A{assp:boundbelow}{}$. The equivalent result for $\muest{\GT}{b,t}$ is stated in \cite{leewhiteley} and is proved under $\A{assp:B}{}$ alone. From a technical point of view, this is possible because the use of indicators instead of backward weights allows for cancellations that simplify the analysis significantly. Assumption  $\A{assp:boundup}{}$ enables us to show that the additional terms that come with the use of backward weights go to zero. 
\begin{hypA}
    \label{assp:positive}
    For all $t > 0$ and $(x,x') \in \Xset^2$, $m_t(x',x) > 0$.
\end{hypA}
\begin{hypA}
    \label{assp:boundup}
    There exists $\sigma_{+} > 0$ such that for all $t \geq 1$,
    $\sup_{x, x' \in \Xset} m_t(x',x) \leq \sigma_{+}$.
\end{hypA}
\begin{hypA}
    \label{assp:boundbelow}
     There exists $0 < \sigma_{-} < \sigma_{+}$ 
    such that for all $t\geq 1$,  $\inf_{x, x' \in \Xset} m_t(x',x) \geq \sigma_{-}$.
\end{hypA}
Assumption $\A{assp:boundbelow}{}$ is a strong assumption that is typically verified in models where the state space $\Xset$ is compact. This assumption, together with $\A{assp:boundup}{},$ are now classic and have been widely used to obtain quantitative bounds in the SMC literature \cite{sylvain_bernoulli,doucmoulines, lee2020coupled}.

\begin{thm}
\label{thm:conv}
Assume that $\A{assp:B}{assp:boundup}$ hold. For any $t \in \N \eqsp,$ $b \in \Bset_t$ and $h \in \bounded{2(t+1)}$,
\begin{equation}
    \label{eq:convlim}
    \simplelim \| \muest{\BS}{b,t}(h) - \mumeasure{b,t}(h) \|_2 = 0 \eqsp.
\end{equation}
In addition, if $\A{assp:boundbelow}{}$ holds the convergence rate is $\bigo(1/\sqrt{N})$.
\end{thm}
\begin{rem}
The dependence on the time horizon $t$ of the $\mathbf{L}_2$ bound is difficult to analyze and we did not undertake it in the proof. Adapting the proofs of the existing analysis \cite{delmoral:doucet:singh:2010, sylvain_bernoulli} is not trivial as our smoothing estimators are non standard. Furthermore, the time dependence of the $\GT$ counterpart has not been analyzed neither, which renders the comparison with our approach even more difficult.
\end{rem}
The proof can be found in Section~\ref*{proof:conv} of the \suppl. As a straightforward consequence, the term by term estimator \eqref{eq:general_termbyterm} of the asymptotic variance is weakly consistent. 
It remains to detail how it can be computed. The next section is devoted to the exact computation of the estimators $\muest{\BS}{\zero,t}(h)$ and $\muest{\BS}{e_s, t}(h)$ that appear in its expression.

\subsection{Computation for $b = 0$ and $b = e_s$}
\label{subsec:compute}
We now derive practical expressions of $\muest{\BS}{\zero,t}(h)$ and $\muest{\BS}{e_s,t}(h)$ in the practical case where 
 $h: x_{0:t}, x^\prime _{0:t} \mapsto h(x_t, x^\prime _t) \in \measurable{2(t+1)}$. Define, for any $b \in \Bset_t$ and any $t\geq 0$,   
\begin{equation}
\label{eq:def_tau}
    \backsum^b _t (K^1 _t, K^2 _t)  \eqdef \pE_{\BS} \big[ \intersect{b,t}{K^1 _{0:t} , K^2 _{0:t}} \big| \F{t}, K^1 _t, K^2 _t \big] \eqsp.
\end{equation}
Then, by the tower property, $\muest{\BS}{b, t}(h)$ in \eqref{eq:mu_estimator} can be rewritten as
\begin{equation}
    \label{eq:Qexpr}
    \muest{\BS}{b, t}(h) = \prod_{s = 0 }^t N^{b_s} 
    \bigg(\frac{N}{N-1}\bigg)^{1 - b_s} \frac{\joint{t}^N(\mathbf{1})^2}{N^2} \sum_{k, \ell \in [N]^2} \backsum^b _t(k,\ell) h(\particle{k}{t}, \particle{\ell}{t}) \eqsp.
\end{equation}
Next, define for any $t \in \N$ and $(k,\ell) \in [N]^2$,
\begin{equation}
\label{eq:S}
    S_t(k,\ell) \eqdef \sum_{s = 0}^{t}  \backsum^{e_s} _{t}(k, \ell) \eqsp.
\end{equation}
Plugging \eqref{eq:Qexpr} in \eqref{eq:general_termbyterm}, $\tbtasymptvarestim{\gamma, t}{\BS}(h)$ can be rewritten as
\begin{equation}
    \label{eq:tbt_expr}
    \tbtasymptvarestim{\gamma, t}{\BS}(h) =  \frac{N^{t-1} \joint{t}^N(\boldone)^2}{(N-1)^t} \sum_{k, \ell \in [N]^2} \left\{ S_t(k, \ell)-\frac{t+1}{N-1}\backsum^\zero _t (k,\ell) \right \} h(\particle{k}{t})h(\particle{\ell}{t}) \eqsp.
\end{equation}
The sequential computation
of $\tbtasymptvarestim{\gamma, t}{\BS}(h)$
relies on that of $S_t(k,\ell)$ in \eqref{eq:S}, and so
on that of $\backsum^{e_s} _{t}(k, \ell)$
and $\backsum^\zero _t (k,\ell)$.
By the tower property, we obtain the following recursions for $\backsum^b _t$:
\begin{equation}
        \label{lem:bigtauexpr}
        \begin{cases}
        \backsum^b _0(k, \ell) = \1_{k \neq \ell, b_0 = 0} + \1_{k = \ell, b_0 = 1} \eqsp,\\
    \backsum^b _t(k, \ell) = \1_{k \neq \ell} \sum_{i,j \in [N]^2} \beta^\BS _t (k,i) \beta^\BS _t(\ell,j) \backsum^{b} _{t-1}(i,j)  & \mathrm{if} \quad b_t = 0 \eqsp, \\
     \backsum^b _t(k,\ell) = \1_{k = \ell} \sum_{i, j \in [N]^2} \beta^\BS _t(k,i) \normweight{j}{t-1} \backsum^b _{t-1}(i,j)  & \mathrm{if} \quad b_t = 1 \eqsp,
        \end{cases}
    \end{equation}
    for all $(k, \ell) \in [N]^2$ and $t \in \N^{*}$. In particular, if $b = \zero$,
    \begin{equation}
        \label{eq:update-0}
            \backsum^{\zero} _t(k,\ell) = \1_{k \neq \ell} \sum_{i,j \in [N]^2} \beta^\BS _t (k,i) \beta^\BS _t(\ell,j) \backsum^{\zero} _{t-1}(i,j) \eqsp,
        \end{equation}
    and if $b = e_s$,
    \begin{align}
        \label{eq:update-tau-s}
            \backsum^{e_s} _t(k,\ell)& = \begin{cases}
                \1_{k \neq \ell} \sum_{i,j \in [N]^2} \beta^\BS _t(k, i) \beta^\BS _t(\ell,j) \backsum^{e_s} _{t-1}(i,j) & \quad t > s \eqsp,\\
                \1_{k = \ell} \sum_{i,j \in [N]^2} \beta^\BS _t(k,i) \normweight{j}{t-1} \backsum^\zero _{t-1}(i,j) & \quad t = s \eqsp, \\
                    \backsum^\zero _t(k, \ell) & \quad t < s\eqsp.\\
            \end{cases}
        \end{align}
        Next, combining \eqref{eq:S}-\eqref{eq:update-tau-s} we obtain the online update of $S_t$:
        \begin{equation}
            \label{eq:updatesumes}
             S_t(k,\ell) = \backsum^{e_t} _t(k,\ell) + \1_{k \neq \ell} \sum_{i,j \in [N]^2}  \beta^\BS _t(k,i) \beta^\BS _t(\ell,j)  S_{t-1}(i,j)
             \eqsp,
            \end{equation}
        for any $(k, \ell) \in [N]^2$ and $t \in \N$. We have shown that despite the sum over $s$ that appears in \eqref{eq:general_termbyterm} we are still able to update \eqref{eq:tbt_expr} at a computational cost independent of the time horizon $t$ by propagating $S_t$ and $\backsum^\zero _t$.
       Note that the algorithm provided in \cite[Algorithm 3, Supplementary]{leewhiteley} does not compute $\tbtasymptvarestim{\gamma,t}{\GT}$ sequentially since it relies on the computation of each $\muest{\GT}{e_s,t}(h_t ^{\otimes 2})$ and $\muest{\GT}{\zero,t}(h_t ^{\otimes 2})$ from scratch whenever a new observation is available. In Section~\ref*{apdx:diffBSGT} of the supplementary material we show how it can be computed online using the same ideas behind the previous derivations. 

    The computation of the estimates $\muest{\BS}{b,t}(h)$ and $\tbtasymptvarestim{\gamma,t}{\BS}(h)$ can benefit from parallelization by implementing the updates \eqref{eq:update-0}-\eqref{eq:update-tau-s} with matrix operations:
    \begin{align*}
        \begin{cases}
            \bm{\backsum}^b _t = \bm{\beta}^\BS _t \bm{\backsum}^b _{t-1} \bm{\beta}^{\BS \top} _t - \mathrm{Diag}(\bm{\beta}^\BS _t \bm{\backsum}^b _{t-1} \bm{\beta}^{\BS \top} _t) &\quad \mathrm{if} \quad  b_t = 0 \eqsp,\\
            \bm{\backsum}^b _t = \mathrm{Diag}(\bm{\beta}^\BS _t \bm{\backsum}^b _{t-1} \mathcal{W}^{1:N} _{t-1}) & \quad \mathrm{if} \quad b_t = 1 \eqsp.
        \end{cases}
    \end{align*}
\subsection{Variance estimators with reduced computational cost}
\label{subsec:disj}
In this section we derive a second estimator that relies only on the update of $\backsum^\zero _t$.
Let $h \in \bounded{}$. By \eqref{eq:CLTs}, 
$\sqrt{N}\big(\joint{t}^N(h) - \joint{t}(h) \big)$ converges in distribution; 
moreover, $N \big(\joint{t}^N(h) - \joint{t}(h) \big)^2$ is uniformly integrable, using for instance a Hoeffding type inequality (see \cite{douc2014nonlinear}). Hence $N \pE [ (\joint{t}^N(h) - \joint{t}(h))^2 ]$ converges to the asymptotic variance $\asymptvar{\gamma,t}(h)$. On the other hand, using the lack of bias of $\joint{t}^N(h) \eqsp,$
 \[ 
     N \pE \left[ \left(\joint{t}^N(h) - \joint{t}(h)\right)^2 \right] = 
     N \left( \pE \left[ \joint{t}^N(h)^2 \right] - \joint{t}(h)^2 \right) =  N \left( \pE \left[ \joint{t}^N(h)^2 \right] - \mumeasure{\zero, t}(h _t^{\otimes 2}) \right)\eqsp.
 \]
A  natural estimator of this quantity
is obtained by replacing both terms by their unbiased estimators
 $\joint{t}^N(h)^2$ and $\muest{\BS}{\zero, t}(h _t ^{\otimes 2})$ 
 \begin{equation}
 \begin{aligned}
    \label{eq:general_disjoint}
    \asymptvarestim{\gamma,t}{\BS}(h) & \eqdef N \big( \joint{t}^N (h) ^2 - \muest{\BS}{\zero,t}(h^{\otimes 2} _t)\big) \\
    & = N \joint{t}^N (\boldone)^2 \bigg( \pred{t}^N(h)^2 - \frac{N^{t-1}}{(N-1)^{t+1}} \sum_{i,j \in [N]^2} \backsum^\zero _t(i,j) h(\particle{i}{t}) h(\particle{j}{t})\bigg) \eqsp.
    \end{aligned}
\end{equation}
For the sake of completeness we also provide the estimator for the predictor and filter and defer their justification to the Section~\ref*{subsec:predfiltervariance} of the supplementary material,
\begin{align}
    \label{eq:main:eta_disjoint}
        \asymptvarestim{\eta, t}{\BS}(h) & \eqdef \frac{-N^{t}}{(N-1)^{t+1}} 
             \sum_{i,j \in [N]^2} \backsum^\zero _{t}(i,j) \big\{ h(\particle{i}{t}) - \pred{t}^N (h) \big\} \big\{ h(\particle{j}{t}) - \pred{t}^N(h) \big\}\eqsp, \\  
    \label{eq:main:phi_disjoint}
        \asymptvarestim{\phi,t}{\BS}(h) & \eqdef \frac{-N^{t+2}}{(N-1)^{t+1}} 
              \sum_{i,j \in [N]^2}  \normweight{i}{t} \normweight{j}{t}  \backsum^\zero _{t}(i,j)\big\{ h(\particle{i}{t}) - \filter{t}^N (h) \big\} \big\{ h(\particle{j}{t}) - \filter{t}^N(h) \big\} \eqsp.
\end{align}
\begin{rem}
\label{rem:stability}
It is worthwhile to note the parallel between \eqref{eq:main:eta_disjoint} and \eqref{eq:CLE} (up to a negligible term depending on $N$); the indicator is replaced by the backward statistic $\backsum^\zero _t(i,j)$ which is the conditional probability of having two disjoint backward trajectories starting from $\particle{i}{t}$ and $\particle{j}{t}$.
\end{rem}
The convergence of \eqref{eq:general_disjoint} stated in Theorem~\ref{thm:consistencyVBS} stems from the following identity  which also appears in \cite{leewhiteley, asmc} and dates back to \cite{cerou2011}: 
\begin{equation}
    \label{prop:scdmoment}
\begin{aligned}
    & \sum_{b \in \Bset _t} \bigg\{ \prod_{s = 0} ^{t}  \frac{1}{N^{b_s}}\bigg( \frac{N-1}{N}\bigg)^{1 - b_s} \bigg\} \muest{\BS}{b,t}(h^{\otimes 2} _t) \\
    & \hspace{.5cm} = \joint{t}^N(\boldone)^2 \pE_{\BS} \left[ \sum_{b \in \Bset _t} \intersect{b,t}{K^1 _{0:t}, K^2 _{0:t}} h(\particle{K^1 _t}{t}) h(\particle{K^2 _t}{t}) \middle| \F{t}\right] =  \joint{t}^N(\boldone)^2 \pred{t}^N(h)^2 = \joint{t}^N (h)^2 \eqsp.
\end{aligned}
\end{equation}
\begin{thm}
    \label{thm:consistencyVBS}
    Let $\A{assp:B}{assp:boundup}$ hold. For any $h \in \bounded{}$, $\asymptvarestim{\gamma, t}{\BS}(h)$ converges in probability to $\jointasymptvar{t}{h}$.
\end{thm}
The proof is in Section~\ref*{proof:consistVBS} of the \suppl. The main advantage of \eqref{eq:general_disjoint} w.r.t. 
\eqref{eq:tbt_expr} is the computational cost. Indeed, remark that \eqref{eq:general_disjoint} only relies
on the sequential update of
 $\backsum^\zero _t$, contrary to 
\eqref{eq:tbt_expr} which also relies
on that of $\backsum^{e_s}_t$. Consequently, 
the computational time of \eqref{eq:general_disjoint}
is roughly twice lower. In addition,
experiments show that the difference in performance is negligible so \eqref{eq:general_disjoint} is to be preferred in practice. 

\begin{rem}
This alternative estimator does not invalidate the relevance of \eqref{eq:general_termbyterm}.
Indeed, remember that \eqref{eq:general_termbyterm} is an unbiased estimator. Moreover, the asymptotic variance estimator of the FFBS algorithm that we provide in Section~\ref{sec:ffbs_est} is a term by term estimator that can be updated online in a way similar to \eqref{eq:general_termbyterm}. 
\end{rem}
\subsection{A PaRIS variance estimator}
\label{subsec:paris}
Let us discuss how the computational cost of \eqref{eq:general_disjoint} and \eqref{eq:tbt_expr} 
can be further reduced \textit{à la PaRIS} \cite{paris,gloaguen2019pseudo}. In \cite{paris}, the forward only implementation of the FFBS algorithm is sped up by replacing the backward statistics by a conditionally unbiased estimator obtained by sampling particle indices according to the backward probabilities $\beta^\BS _t$ through rejection sampling. We therefore apply the same idea here by letting $\widetilde{\backsum}^\zero _0 \eqdef \backsum^\zero _0$ and replacing $\backsum^b _t$ with
\label{eqdef:tautilde}
\begin{alignat*}{3}
& \Pbacksum^b _t(k,\ell)  && \eqdef \frac{\1_{k \neq \ell}}{M} \sum_{i = 1}^{M}\Pbacksum^b _{t-1}(J^i _{k,t-1}, J^i _{\ell,t-1}) && \quad \mathrm{if} \quad  b_t = 0 \eqsp, \\
    & \Pbacksum^b _t(k,\ell)  && \eqdef \frac{\1_{k = \ell}}{M} \sum_{i = 1}^M \sum_{j = 1}^N \normweight{j}{t-1} \Pbacksum^b _{t-1}(J^i _{k,t-1}, j) && \quad \mathrm{if} \quad b_t = 1 \eqsp,
\end{alignat*}
where for any $k \in [N]$, $J^{1:M} _{k, t-1}$ 
are i.i.d. samples according
to $\beta^\BS _t(k, .)$. For $h \in \measurable{2}$, the \textit{PaRIS} estimator of $\mumeasure{b,t}(h)$ is, for any $b \in \Bset_t$
\begin{equation}
\label{eqdef:qtilde}
    \parismuest{\BS}{b,t}(h) = \left\{\prod_{s = 0}^{t} N^{b_s} \left( \frac{N}{N-1}\right)^{1 - b_s}\right\} \frac{\joint{t}^N(1)^2}{N^2} \sum_{i,j \in [N]^2} \widetilde{\backsum}^b _t(i,j) h(\particle{i}{t}, \particle{j}{t}) \eqsp,
\end{equation}
and the $\paris$ variance estimators are 
\begin{align}
    \label{eqdef:tbtparis}
    \overline{\mathsf{V}}^{N,M} _{\gamma,t}(h) & = \sum_{s = 0}^t \big\{ \parismuest{\BS}{e_s, t}(h^{\otimes 2}) - \parismuest{\BS}{\zero, t}(h^{\otimes 2}) \big\} \eqsp, \\
    \label{eqdef:VBSparis}
    \parisasymptvar{\gamma, t}{\BS}(h) & = N\big( \joint{t}^N(h)^2 - \parismuest{\BS}{\zero,t}(h^{\otimes 2}) \big) \eqsp,
\end{align}
where $M > 1$ refers to the number of sampled indices. The computation of \eqref{eq:general_disjoint} and \eqref{eqdef:VBSparis} is summarized in Algorithm~\ref{alg:summaryalg}.
\begin{algorithm}
    \caption{Update at step $t+1$ of the variance estimators \eqref{eq:general_disjoint} and \eqref{eqdef:VBSparis} associated to  $\joint{t+1}^N(h)$}
    \label{alg:summaryalg}
    \begin{algorithmic}
    \Require $M, \weight{1:N}{t}, \particle{1:N}{t},
     \particle{1:N}{t+1}$, $\bm{\backsum}^\zero _t$ and $\joint{t}^N(\boldone)$
    \State{Compute $\bm{\beta}^\BS _{t+1}$}
    \If{$\paris$}
    \For{$k \in [1:N]$}

         Sample $J^{1:M} _{k,t} \iid \beta^\BS _{t+1}(k, .)$

    \EndFor
    \For{$(k,\ell) \in [1:N]^2$}

        Set $\backsum^\zero _{t+1}(k,\ell) = \1_{k \neq \ell} \sum_{i = 1}^{M} \backsum^\zero _{t}(
            J^i _{k,t}, J^i _{\ell,t}) / M$

    \EndFor
    \Else

    \State{Compute $\overline{\bm{\backsum}}^\zero _{t+1} = \bm{\beta}^\BS _{t+1} \bm{\backsum}^\zero _t \bm{\beta}^{\BS \prime} _{t+1}$.}
    \State{Set $\bm{\backsum}^\zero _{t+1} = \overline{\bm{\backsum}}^\zero _{t+1} - \mathrm{Diag}(\overline{\bm{\backsum}}^\zero _{t+1})$.}
    \EndIf

    \State{Compute $\bm{\mathcal{Q}} = \bm{\backsum}^\zero _{t+1} \odot \big[ h(\particle{1:N}{t+1}) h(\particle{1:N}{t+1}) ^\top \big]$.}\algorithmiccomment{$h$ is applied elementwise}\\
    \Return $N\joint{t+1}^N(\boldone)^2 \big\{ \eta^N _{t+1}(h)^2 - N^{t} \sum_{i, j \in [N]^2} \bm{\mathcal{Q}}_{i,j} / (N-1)^{t+2} \big\}, \quad \bm{\backsum}^\zero _{t+1}$.
    \end{algorithmic}
    \end{algorithm}

We are able to reduce the time complexity of computing $\backsum^b _t$ to $\bigo(MN^2)$. The key feature of the \textit{PaRIS} approach is that $M$ does not necessarily need to be large (see \cite[Section 3.1]{paris} for a discussion on this matter). We impose $M > 1$ because then in the case $b = \zero$, which is the case we are the most interested in, the support of $\muest{\BS}{\zero,t}(h)$ is made of $N^2 M^{t+1}$ terms whereas when $M = 1$ it is only $N^2$.  We show empirically in our experiments that setting $M = 3$ is sufficient to provide good results for the asymptotic variance estimation.

While it is not needed to obtain a $\bigo(MN^2)$ time complexity, the indices $J^{1:M} _{k, t-1}$ can be sampled using an accept-reject procedure with the weights $\normweight{1:N}{t}$ as proposals if the transition densities $\transitiondens{t}$ are upper bounded. This approach does not require the computation of the normalizing constant of the backward weights \eqref{def:functionalweight}. The computational time is then random but if the transition kernels are strongly mixing it can be provably further reduced \cite{doucmoulines}. 
Theorem~\ref{corr:paris} is concerned with the convergence of $\parismuest{\BS}{b,t}(h)$ for any bounded $h$ and for any \emph{fixed} $M > 1$. Its proof bears some similarity with that of Theorem~\ref{thm:conv} with the exception that the additional sampling introduces non trivial terms that need to be handled carefully. As a straightforward consequence, we obtain the convergence in probability of $\overline{\mathsf{V}}^{N,M} _{\gamma,t}(h)$ for any $h \in \bounded{}$. The weak consistency of \eqref{eqdef:VBSparis} in Theorem~\ref{thm:parisvar} is however less straightforward than that of Theorem~\ref{thm:consistencyVBS} and relies on the insight that the identity \eqref{prop:scdmoment} still holds when $\muest{\BS}{b,t}$ are replaced with their $\paris$ versions. The proofs are provided respectively in Section~\ref*{proof:thmparis} and \ref*{proof:scdmoment_paris} of the \suppl.
\begin{thm}
    \label{corr:paris}
    Assume that $\A{assp:B}{assp:boundup}$ hold. For any $t \in \N \eqsp,$ $b \in \Bset_t$, $M > 1$ and $h \in \bounded{2}$,
    \begin{equation}
        \simplelim \| \parismuest{\BS}{b,t}(h) - \mumeasure{b,t}(h) \|_2 = 0 \eqsp.
    \end{equation}
    In addition, if $\A{assp:boundbelow}{}$ holds the convergence rate is $\bigo(1/\sqrt{N})$.
\end{thm}
\begin{thm}
    \label{thm:parisvar}
    Let $\A{assp:B}{assp:boundup}$ hold. For all $t \in \N$, $M > 1$ and $h \in \bounded{}$, $\parisasymptvar{\gamma,t}{\BS}(h)$ converges in probability to $\asymptvar{\gamma,t}(h)$ when $N$ goes to infinity.
\end{thm}

\section{Application to the FFBS}
\label{sec:varFFBS}
In this section, we derive an estimator for the asymptotic variance of the \emph{Forward Filtering Backward Smoothing} algorithm. We start by giving 
a short presentation of the FFBS algorithm and we next derive an estimator of the asymptotic variance for additive functionals. 

\subsection{FFBS algorithm}
The FFBS algorithm aims at solving the well known degeneracy problem associated with the particle filter of Section \ref{sec:PF} when it is used for approximating smoothing distributions. It relies on the following backward decomposition of the joint smoothing distribution:
\begin{equation}
    \label{eq:smoother}
\smooth{0:t}{t}(h) = \int h(x_{0:t}) \filter{t}(\rmd x_t) \bwpath{t}(x_t, \rmd x_{0:t-1}) \eqsp,
\end{equation}
where $\bwpath{t}$ is the backward transition kernel from $(\Xset, \sigmaX)$ to $(\Xset^t, \sigmaX^{\otimes t})$: $\bwpath{0} \eqdef \text { Id } $ and for $t>0$,
$$
\bwpath{t} \eqdef \bwker{t-1} \otimes  \cdots \otimes \bwker{0}
$$
and $\bwker{s}$ is the backward kernel defined by 
\[ 
\bwker{s}(x_{s+1}, A) \eqdef \frac{\int m_{s+1}(x_s, x_{s+1}) \1_A (x_s) 
\filter{s}(\rmd x_s)}{\phi_s(m_{s+1}(.,x_{s+1}))}, \quad \forall A \in \sigmaX \eqsp, \forall x_{s+1} \in \Xset \eqsp.
\]
Denote by $\bwpath{t}^N$ the particle approximation of $\bwpath{t}$ where each backward kernel $\bwker{s}$ is replaced by plugging in the particle approximation of the filter. This yields for any $A \in \sigmaX$ and $x_{s+1} \in \Xset$, 
\[
\bwker{s}^N(x_{s+1}, A ) \eqdef \frac{\int m_{s+1}(x_s, x_{s+1}) 
\1_A(x_s) \filter{s}^N(\rmd x_s)}{\filter{s}^N(m_{s+1}(.,x_{s+1}))} =
 \sum_{i = 1}^N \frac{\weight{i}{s}m_{s+1}(\particle{i}{s}, x_{s+1})}{\sum_{j = 1}^N 
 \weight{j}{s} m_{s+1}(\particle{j}{s}, x_{s+1})} \1_A(\particle{i}{s}) \eqsp.
\]
Plugging this approximation and that of the filtering distribution  in \eqref{eq:smoother} yields
\begin{equation}
    \label{eq:FFBS}
\smoothN{0:t}{t}{N, \FFBS}(h) \eqdef \sum_{i_0 = 1}^N \cdots \sum_{i_t = 1}^N \widetilde{\Lambda} _t (i_{0:t}) h(\particle{i_0}{0}, \cdots , \particle{i_t}{t}) \eqsp,
\end{equation}
where $\widetilde{\Lambda}_t (i _{0:t}) \eqdef \normweight{i_t}{t} \prod_{s = 1}^t \beta^\BS _t(i _{s}, i_{s-1})$. In the following, we write $\smooth{0:t}{t}^N$ for $\smoothN{0:t}{t}{N, \FFBS}$ and if $h$ is such that $h : x_{0:t} \mapsto h(x_{s:\ell})$ with $0 \leq s \leq \ell \leq t$, we will instead write $\smooth{s:\ell}{t}^N (h)$.

The theoretical properties of the FFBS are well understood in both the asymptotic regimes of $N$ and $t$ \cite{doucmoulines, del2010forward, del2010backward, sylvain_bernoulli, paris,douc2014nonlinear}.
In particular, a Central Limit Theorem with an explicit expression of the asymptotic variance is established for any $h \in \bounded{t+1}$ under $\A{assp:B}{}$ in \cite[Theorem 8]{doucmoulines},
\begin{equation}
\label{eq:FFBSTCL}
\sqrt{N}\big( \smoothN{0:t}{t}{N}(h) - \smooth{0:t}{t}(h) \big) \dlim \gauss\big(0, \ffbsasymptvar{\FFBS}{0:t}{t}(h)\big)\eqsp,
\end{equation}
where 
\begin{equation}
    \label{eq:FFBSvar}
\ffbsasymptvar{\FFBS}{0:t}{t}(h) \eqdef
\sum_{s = 0}^t \frac{\pred{s}\big(\GG{s,t}\big[ g_t \big\{h - \smooth{0:t}{t}(h) \big\}\big]^2)}{\pred{s}(\Qmarg{s+1:t}[g_t])^2} 
\end{equation}
and $\GG{s,t}$ is the kernel that integrates $h$ forward and backward starting from $x_s$, i.e. \[ 
    \GG{s,t}[h](x_s) \eqdef \bwpath{s}\big[\Q{s+1:t}[h]\big](x_s)
    = \int h(x_{0:t}) \bwpath{s}(x_s, \rmd x_{0:s-1}) \Q{s+1:t}(x_s, \rmd x_{s+1:t}) \eqsp,\]
     for any $s \in [0:t]$ and $x_s \in \Xset$.

     Unlike the asymptotic variance of filtering algorithms, no estimator of \eqref{eq:FFBSvar} exists in the literature, even though the FFBS and its variants are of significant importance in marginal smoothing and parameter estimation in HMMs \cite{kantas2015particle}.
     In this section, we bridge this gap by providing an online estimator for additive functionals $h$ of the form 
     \begin{equation}
        \label{eq:additive}
    h_{0:t}(x_{0:t}) = \sum_{s = 0}^{t-1} \tilde{h}_s(x_s, x_{s+1}) \eqsp,
    \end{equation}
    where for $s \in [0:t-1]$, we assume that $\tilde{h}_s$ is bounded. For such functionals, the FFBS can be computed online with a $\bigo(N^2)$ time complexity per time step, i.e. whenever a new observation is processed. 
    For $0 \leq s < r \leq t$, we write $\tilde{h}_{s:r}(x_{s:r}) = \sum_{\ell = s}^{r-1} \tilde{h}_{\ell}(x_\ell, x_{\ell + 1})$.
    Expectations of functionals of the form \eqref{eq:additive} include marginal smoothing, pairwise marginal smoothing and the \emph{E}-step of the Expectation Maximization algorithm.

Before we derive our estimator, let us first recall why the FFBS can be indeed computed online in this case. For more details on the forward only implementation of the FFBS and its variants we refer the reader to \cite{douc2014nonlinear, paris}. For any $t > 0$ and  any additive functional $h_{0:t}$,
\begin{align*}
    \bwpath{t}[h_{0:t}](x_t) 
    & = \int \left\{\tilde{h}_{0:t-1}(x_{0:t-1}) + \tilde{h}_{t-1}(x_{t-1}, x_t)\right\} \bwker{t-1}(x_t, \rmd x_{t-1}) \bwpath{t-1}(x_{t-1}, \rmd x_{0:t-2}) \\
    & = \bwker{t-1}\big[\bwpath{t-1}[\tilde{h}_{0:t-1}] + \tilde{h}_{t-1}
    (., x_t) \big](x_t) \eqsp.
\end{align*}
Then, plugging in the particle approximations, we obtain the following recursion 
\begin{equation}
\label{eq:FFBSrecursion}
\bwpath{t}^N[h_{0:t}](x_t) = \sum_{i = 1}^N \frac{\weight{i}{t-1} 
m_t(\particle{i}{t-1}, x_t)}{\sum_{j = 1}^N \weight{j}{t-1} m_t(\particle{j}{t-1}, 
x_t)} \big\{ \bwpath{t-1}^N[\tilde{h}_{0:t-1}](\particle{i}{t-1}) + \tilde{h}_{t-1}(\particle{i}{t-1}, x_t)\big\} \eqsp,
\end{equation}
and then $\smoothN{0:t}{t}{N}(h_{0:t}) = \sum_{i = 1}^N \normweight{i}{t} \bwpath{t}^N[h_{0:t}](\particle{i}{t})$.  
Therefore, $\bwpath{t}^N[h_{0:t}]$ needs only to be estimated at the particle locations and smoothing estimates for additive functionals can be computed with the forward pass and has $\bigo(N^2)$ complexity per time step.

\subsection{Asymptotic variance estimator}
\label{sec:ffbs_est}
From now on we will assume that $h_{0:t}$
satisfies \eqref{eq:additive}. Our estimator is based on the following alternative 
expression of the asymptotic variance \eqref{eq:FFBSvar}
\begin{equation}
    \label{eq:reexpression}
    \ffbsasymptvar{\FFBS}{0:t}{t}(h) = \sum_{s = 0}^t \frac{\joint{s}(\boldone)\joint{s}\big(\GG{s,t}[g_t \{ h_{0:t} -
     \smooth{0:t}{t}(h_{0:t}) \}]^2)}{\joint{t+1}(\boldone)^2}\eqsp, 
\end{equation}
which is deduced 
using the definitions given in Section~\ref{sec:SMC}. This expression is motivated by Proposition~\ref{prop:identityQes} in which we express the numerators that appear in \eqref{eq:reexpression} in terms of expectations with respect to $\mumeasure{e_s,t}$. The proof is given in Section~\ref{proof:identityQes} of the \suppl.

\begin{proposition}
    \label{prop:identityQes}
For any $s \in [0:t]$ and any additive functional $h_{0:t} \in \measurable{t+1}$,
\begin{equation}
    \joint{s}(\boldone)\joint{s}\big(\GG{s,t}[h_{0:t}]^2) =  \mumeasure{e_s, t}\big( \big[\bwpath{s}[\tilde{h}_{0:s}] + \tilde{h}_{s:t}\big]^{\otimes 2}\big) \eqsp.
\end{equation}
\end{proposition}
By Theorem~\ref{thm:conv}, for any additive functional $h_{0:t}$ as in \eqref{eq:additive}, we have that $\muest{\BS}{e_s, t}([\bwpath{s}[\tilde{h}_{0:s}] + \tilde{h}_{s:t}]^{\otimes 2})$ is a consistent estimator
of $\mumeasure{e_s, t}([\bwpath{s}[\tilde{h}_{0:s}] + \tilde{h}_{s:t}]^{\otimes 2})$, but
$\bwpath{s}[h_{0:s}]$ is intractable and we only have access to its particle approximation $\bwpath{s}^N [h_{0:s}]$. Our
proposed estimator of the asymptotic variance \eqref{eq:FFBSTCL} is then
\begin{equation}
\label{eq:FFBSvarestim}
\ffbsasymptvar{N, \BS}{0:t}{t}(h_t) \eqdef \sum_{s = 0}^T 
\frac{\muest{\BS}{e_s, t}\big(\big[g_t \{ \bwpath{s}^N[\tilde{h}_{0:s}] + \tilde{h}_{s:t} - \smooth{0:t}{t}^{N}(h_{0:t}) \} \big]^{\otimes 2}\big)}{\joint{t+1}^N (\boldone)^2}\eqsp,
\end{equation}
where we have replaced $\smooth{0:t}{t}(h_t)$ by its FFBS estimator. Remark that Theorem~\ref{thm:conv} cannot be applied to $\muest{\BS}{e_s, t}\big(\big[g_t \{ \bwpath{s}^N[\tilde{h}_{0:s}] + \tilde{h}_{s:t} \big]^{\otimes 2} \big)$ because its proof relies on the fact that the function $h$ integrated by $\muest{\BS}{b,t}$ does not depend on the particles.

Theorem~\ref{thm:convFFBS} proved in Section~\ref*{proof:FFBS} of the supplementary material shows that weak consistency still holds under the assumptions of Theorem~\ref{thm:conv}. The proof proceeds in three steps. We first establish that for all $s > 0$ and additive functional $h_{0:s}$, $\bwpath{s}^N[h_{0:s}](x_s)$ converges $\pP$-a.s. to $\bwpath{s}[h_{0:s}](x_s)$ for any $x_s \in \Xset$. Then, we use it to show that at $t = s$, the distance in $\mathbf{L}_2$ between $\muest{\BS}{e_s, s}\big(\big[\bwpath{s}^N[h_{0:s}] c_s + \tilde{h}_{s} \big] \otimes \big[\bwpath{s}^N[f_{0:s}] d_s + \tilde{f}_{s} \big] \big)$ and the "idealized" consistent estimator $\muest{\BS}{e_s, s}\big(\big[\bwpath{s}[h_{0:s}] c_s + \tilde{h}_{s} \big] \otimes \big[\bwpath{s}[f_{0:s}] d_s + \tilde{f}_{s} \big] \big)$,  goes to 0. 
Finally, we extend the result to $t > s$ by induction, similarly to Theorem~\ref{thm:conv}.
\begin{thm}
    \label{thm:convFFBS}
Assume that $\A{assp:B}{assp:boundup}$ hold. For any $t \in \N$, $s \in [0:t]$, $(\tilde{h}_{s:t}, \tilde{f}_{s:t}) \in \bounded{t-s+1}^2$, $(c _t, d _t) \in \bounded{}^2$ and additive functionnals $(h_{0:s}, f_{0:s})$ \eqref{eq:additive},
\begin{multline}
    \label{eq:hypthmFFBS}
    \simplelim \big\| \muest{\BS}{e_s, t}\big( \big[\bwpath{s}^N[h_{0:s}]c _t + \tilde{h}_{s:t} \big] \otimes
    \big[\bwpath{s}^N[f_{0:s}]d _t + \tilde{f}_{s:t} \big] \big) \\
    -  \mumeasure{e_s, t}\big( \big[\bwpath{s}[h_{0:s}]c _t + \tilde{h}_{s:t} \big] \otimes
    \big[\bwpath{s}[f_{0:s}]d _t + \tilde{f}_{s:t} \big] \big) \big\|_2 = 0 \eqsp,
\end{multline}
and for any additive functional \eqref{eq:additive}, $\ffbsasymptvar{\BS}{0:t}{t}(h_{0:t})$ converges in probability to $\ffbsasymptvar{\FFBS}{0:t}{t}(h_{0:t})$.
\end{thm}
\subsection{Algorithm for marginal smoothing}
\label{subsec:margsmoothing}
We now provide an algorithm for the case $h_\ell: x_{0:t} \mapsto h_\ell(x_\ell)$ known 
as the marginal smoothing problem. For such functions \eqref{eq:FFBSvarestim} is defined for $t \geq \ell$ and simplifies to 
\begin{equation}
\begin{aligned}
    \label{eq:FFBSmarg}
    & \ffbsasymptvar{N,\BS}{\ell}{t}(h_\ell) \eqdef \frac{1}{\joint{t+1}^N (\boldone)^2}\bigg\{ \sum_{s = 0}^{\ell} \muest{\BS}{e_s, t}\big(\big[ g_t \{ h_\ell- 
    \smooth{\ell}{t}^{N}(h_\ell) \} \big]^{\otimes 2}\big) \\
    &\hspace{5cm} + \sum_{s = \ell+1}^t
    \muest{\BS}{e_s, t}\big(\big[g_t \{ \bwpath{s}^N[h_{\ell}] - 
    \smooth{\ell}{t}^{N}(h_\ell) \} \big]^{\otimes 2}\big) \bigg\} \eqsp.
    \end{aligned}
\end{equation}
When $\ell = t$ we recover the term by term estimator of the filter which is consistent with the fact that $\smooth{t}{t}^N(h) = \filter{t}^N(h)$. Using the bilinearity of $\muest{\BS}{b, t}$  yields
\begin{equation}
    \label{eq:varupdate}
    \ffbsasymptvar{\BS}{\ell}{t}(h_\ell) = R^\ell _{1,t} - \smooth{\ell}{t}^{N}(h_\ell) R^\ell _{2,t} + \smooth{\ell}{t}^{N}(h_\ell)^2 R_t \eqsp,
\end{equation}
where $R_t  \eqdef \sum_{s = 0}^t \muest{\BS}{e_s, t}(g_t ^{\otimes 2})$ and
\begin{align*}
    \begin{cases}
        R^1 _{\ell,t} & \eqdef \sum_{s = 0}^{\ell} \muest{\BS}{e_s, t}\big( [g_t h_\ell]^{\otimes 2} \big) + \sum_{s = \ell+1}^t \muest{\BS}{e_s, t}(\big[g_t \bwpath{s}^N [h_\ell] \big]^{\otimes 2}) \eqsp,  \\
        R^2 _{\ell,t} & \eqdef \big\{ \sum_{s = 0}^{\ell} \muest{\BS}{e_s, t}(g_t h_\ell \otimes g_t) + \muest{\BS}{e_s, t}(g_t \otimes g_t h_\ell)\big\} \\
        & \hspace{3cm} + \big\{ \sum_{s = \ell+1}^{t} \muest{\BS}{e_s, t}(g_t \bwpath{s}^N [h_\ell] \otimes g_t) + \muest{\BS}{e_s, t}(g_t \otimes g_t \bwpath{s}^N [h_\ell]) \big\} \eqsp.
    \end{cases}
\end{align*}
Mirroring \eqref{eq:update-tau-s}, define for any $t \in \N$, $n \in [0:t]$ and $f_n : x_{0:t}, x' _{0:t} \mapsto f(x_n, x' _n)$ the random variable 
\begin{equation}
    \backsum^{e_s} _t [f_n](K^1 _t, K^2 _t) \eqdef \pE \big[ \intersect{e_s, t}{K^1 _{0:t}, K^2 _{0:t}} f _n(\particle{K^1 _n}{n}, \particle{K^2 _n}{n}) \big| \F{t-1}, K^1 _t, K^2 _t \big]
\end{equation}
and also write $S^1 _{t, t}(K^1 _t, K^2 _t) = S_t(K^1 _t, K^2 _t) h^{\otimes 2} _t (\particle{K^1 _t}{t}, \particle{K^2 _t}{t})$, $S^2 _{t, t}(K^1 _t, K^2 _t) =S_t(K^1 _t, K^2 _t) h^{\oplus 2} _t (\particle{K^1 _t}{t}, \particle{K^2 _t}{t})$ and for any $t > \ell$,
\begin{align*}
    S^1 _{\ell, t}(K^1 _t, K^2 _t) &= \sum_{s = 0}^\ell \backsum^{e_s} _t [h^{\otimes 2} _\ell](K^1 _t, K^2 _t) + \sum_{s = \ell + 1}^t \backsum^{e_s} _t \big[\bwpath{s}^N [h_\ell]^{\otimes 2} \big](K^1 _t, K^2 _t),\\
    S^2 _{\ell, t}(K^1 _t, K^2 _t) &= \sum_{s = 0}^\ell \backsum^{e_s} _t [h^{\oplus 2} _\ell](K^1 _t, K^2 _t) + \sum_{s = \ell + 1}^t \backsum^{e_s} _t \big[\bwpath{s}^N [h_\ell]^{\oplus 2} \big](K^1 _t, K^2 _t),
\end{align*}

where for any $f_\ell$, $f^{\oplus 2} _\ell: x_\ell, x' _\ell \mapsto f_\ell(x_\ell) + f_\ell(x' _\ell)$ and $S_\ell$ is defined in \eqref{eq:S}. Applying the tower property we get that 
\begin{equation*}
    \ffbsasymptvar{\BS}{\ell}{t}(h_t) = N \left( \frac{N}{N-1} \right)^t \sum_{i, j \in [N]^2} \normweight{i}{t} \normweight{j}{t} \big\{ S^1 _{\ell, t}(i,j) - \smooth{\ell}{t}^{N}(h_\ell) S^2 _{\ell,t}(i,j) + \smooth{\ell}{t}^{N}(h_\ell)^2 S_t(i,j)\big\} \eqsp.
\end{equation*}
The quantities $S^1 _{\ell, t+1}$, $S^2 _{\ell, t+1}$ may be updated online using the following recursions which are again obtained by applying the tower property
\begin{multline}
    \label{eq:updateS1}
    S^1 _{\ell, t+1} (i,j) \eqdef \backsum^{e _{t+1}} _{t+1}(i,j) \bwpath{t+1}^N[h_\ell](\particle{i}{t+1}) \bwpath{t+1}^N[h_\ell](\particle{j}{t+1}) \\ + \1_{i \neq j} \sum_{m,n \in [N]^2} \beta^\BS _{t+1}(i,m) \beta^\BS _{t+1}(j,n) S^1 _{\ell, t}(m,n) \eqsp,
\end{multline}
and 
\begin{multline}
    \label{eq:updateS2}
    S^2 _{\ell,t+1}(i,j) \eqdef \backsum^{e _{t+1}} _{t+1}(i,j) \big\{ \bwpath{t+1}^N[h_\ell](\particle{i}{t+1}) + \bwpath{t+1}^N[h_\ell](\particle{j}{t+1}) \big\} \\ + \1_{i \neq j}\sum_{m,n \in [N]^2} \beta^\BS _{t+1}(i,m) \beta^\BS _{t+1}(j,n) S^2 _{\ell, t}(m,n)  \eqsp.
\end{multline}
The updates of 
 $S^1 _{\ell, t+1} (i,j)$ and 
of  $S^2 _{\ell, t+1} (i,j)$ are thus similiar to that of $S_t$ in \eqref{eq:updatesumes}. The computation
of the variance estimator is described in Algorithm~\ref{alg:algFFBS}.
\begin{algorithm}

    \caption{Update at step $t+1$ of the variance estimator for marginal smoothing}
    \label{alg:algFFBS}
    \begin{algorithmic}
    \Require $\bm{\mathcal{W}}^{1:N} _{t+1}, \bm{\mathcal{W}}^{1:N} _t, \bm{\beta}^\BS _{t+1},
     \bwpath{t+1}^N [h_\ell], \bm{\backsum}^\zero _t, \bm{S}^1 _{\ell,t}, \bm{S}^2 _{\ell,t}, \smooth{\ell}{t+1}^{N}(h_\ell)$ 
    \State{Compute $\bm{\overline{\backsum}}^{e_{t+1}} _{t+1} = \bm\beta^\BS _{t+1} \bm\backsum^\zero _t \bm{\mathcal{W}} _t$, $\quad \bm{\overline{\backsum}}^{\zero} _{t+1} = \bm\beta^\BS _{t+1} \bm\backsum^\zero _t \bm\beta^{\BS \top} _{t+1}$, $\quad \widetilde{\bm{S}}_{t+1} = \bm\beta^\BS _{t+1}  \bm{S}_{t} \bm\beta^{\BS \top} _{t+1}$}
    \State{Set $\bm{\backsum}^{e _{t+1}} _{t+1} = \mathrm{Diag}(\bm{\overline{\backsum}}^{e_{t+1}} _{t+1})$, $\quad \bm{\backsum}^\zero _{t+1} = \bm{\overline{\backsum}}^{\zero} _{t+1} - \mathrm{Diag}(\bm{\overline{\backsum}}^{\zero} _{t+1})$, $\quad \bm{S}_{t+1} = \widetilde{\bm{S}}_{t+1} - \mathrm{Diag}(\widetilde{\bm{S}}_{t+1}) + \backsum^{e_{t+1}} _{t+1}$}
    \For{i $\in \{1,2\}$}
    \State{Compute $\bm{\widetilde{S}}^i _{\ell,t+1} = \bm\beta^\BS _{t+1}\bm{S}^i _{\ell,t} \bm\beta^{\BS \top} _{t+1}$ \eqsp.}
    \State{Set $\bm{\widetilde{S}}^i _{\ell,t+1} = \bm{\widetilde{S}}^i _{t+1} - \mathrm{Diag}(\bm{\widetilde{S}}^i _{t+1})$}
    \EndFor
    \State{Set $\bm{S}^1 _{\ell,t+1} = \bm{\widetilde{S}}^1 _{\ell,t+1} + \bm\backsum^{e_{t+1}} _{t+1} \odot \big[\bwpath{t+1}[h_\ell] \bwpath{t+1}[h _\ell]^\top ]$}
    \State{Set $\bm{S}^2 _{\ell,t+1} = \bm{\widetilde{S}}^2 _{\ell,t+1} + \bm\backsum^{e_{t+1}} _{t+1} \odot \big[\bwpath{t+1}[h_\ell] + \bwpath{t+1}[h _\ell]^\top ]$}
    \State{Set $\overline{\bm{S}}_{\ell, t+1} = \bm{S}^1 _{\ell,t+1} - \smooth{\ell}{t+1}^{N}(h_\ell) \bm{S}^2 _{\ell,t+1} + \smooth{\ell}{t+1}^{N}(h_\ell)^2\bm{S}_{t+1}$}
    \\
    \Return $- N^{t+2} / (N-1)^{t+1} \sum_{i, j \in [N]^2} \normweight{i}{t+1} \normweight{j}{t+1} \overline{S} _{\ell, t+1}(i,j) , \backsum^\zero _{t+1}, \bm{S}^1 _{\ell,t+1}, \bm{S}^2 _{\ell,t+1}, \bm{S}_{t+1}$.
    \end{algorithmic}
    \end{algorithm}

\section{Numerical simulations}
\label{sec:experiments}
We now demonstrate our estimators on particle filtering and smoothing examples in HMMs (see Ex. \ref{ex:hmm}). We assume in this section that $\Xset = \Yset = \R$ and that the dominating measure is the Lebesgue measure. The model considered is the \emph{stochastic volatility model} with, for all $n\geq 1$,
\begin{equation}
    \label{model:stovol}
    X_{n+1}  = \varphi X_n +  \sigma U_{n+1} \quad\mathrm{and}\quad 
    Y_n  = \beta \exp(X_n / 2) V_n \eqsp,
\end{equation}
with $(\varphi, \beta, \sigma) = (.975, .641, .165)$, 
 $\{U_n\}_{n \in \N}$ and  $\{V_n\} _{n \in \N}$ 
 are two sequences of independent standard Gaussian noises and $U_n$ is independent of $V_m$ for all 
 $(n,m) \in \N^2$.
 The state process $\{X_n\}_{n \in \N}$ is initialized with a Gaussian distribution with  zero mean and variance $\sigma^2 / (1 - \varphi^2)$. These are the exact values and initialization used in \cite{olssonvar}. 
 The assumptions on the model under which \cite{olssonvar} conduct their theoretical analysis and $\A{assp:A}{assp:boundup}$ are satisfied for this model. All the simulations are run on GPU and the implementations using matrix operations are available at \url{https://github.com/yazidjanati/asymptoticvariance}. 
\subsection{Asymptotic variance of the predictor} We are interested in the estimation of the asymptotic variance of the predictor $\pred{t}^N (\rm Id)$ \emph{at each time step} $t$. The estimator is given in Section~\ref*{subsec:predfiltervariance} of the  \suppl.
We use synthetic datasets sampled from \eqref{model:stovol}. The real asymptotic variances are intractable and we estimate them by repeating independently and a thousand times the computation of each predictor mean $\pred{t}^N(\rm Id)$ with $N = 10000$ and then multiplying the sample variance by $N$.

We first investigate how the three variance estimators based on backward sampling $\asymptvarestim{\eta,t}{\BS}$, $\parisasymptvar{\eta,t}{\BS}$ and $\tbtasymptvarestim{\eta,t}{\BS}$ behave in terms of computational time, bias and variance. The $\paris$ estimator is used with $M = 3$ without rejection sampling. For this first experiment, we sampled 750 observations from \eqref{model:stovol} and ran $50$ particle filters with $N = 3000$ from which we obtained 50 replicates of each asymptotic variance estimate. The results are reported in Figure~\ref{fig:comparison3VBS}. 
The three estimators exhibit approximately the same variance but the term by term version becomes slightly more biased at $t$ increases. Strikingly, the $\paris$ estimator $\parisasymptvar{\eta,t}{\BS}$ behaves similarly to $\asymptvarestim{\eta,t}{\BS}$ with a much lower computational time and complexity as can be seen on the left plot of Figure~\ref*{fig:runtime} in the supplementary material.  

We now compare $\parisasymptvar{\eta,t}{\BS}$ with $M = 3$ to \cite{chan,leewhiteley,olssonvar} on two different observation records of different length.
For the lag size parameter $\lambda$, we found that $\lambda = 20$ has the best bias-variance trade-off by comparing the obtained fixed lag estimates with the crude asymptotic variance estimator. 
Note that in realistic situations choosing the right $\lambda$ is non trivial (besides the strong mixing case, as argued in \cite{olssonvar}) and for this reason we conducted the experiments with two additional lag values, $\lambda \in \{100,200\}$. For moderately long observation records $(t \in [750])$ and with $N = 3000$, $\parisasymptvar{\eta, t}{\BS}$ compares favorably in terms of bias-variance trade-off with the best lagged estimator and even has similar computational time on GPU. The results are reported in Figure~\ref{fig:vbsvslag}. 
The five different estimators are computed with the same particle cloud and replicated 50 times. As expected, when the lag is increased the fixed lag estimates exhibit more variance because of the particle degeneracy. In the extreme case where the lag is set to $750$ (CLE) bias and variance both increase significantly, as showcased in the fifth plot.
\begin{figure}
    \centering
    \includegraphics[width = 12cm]{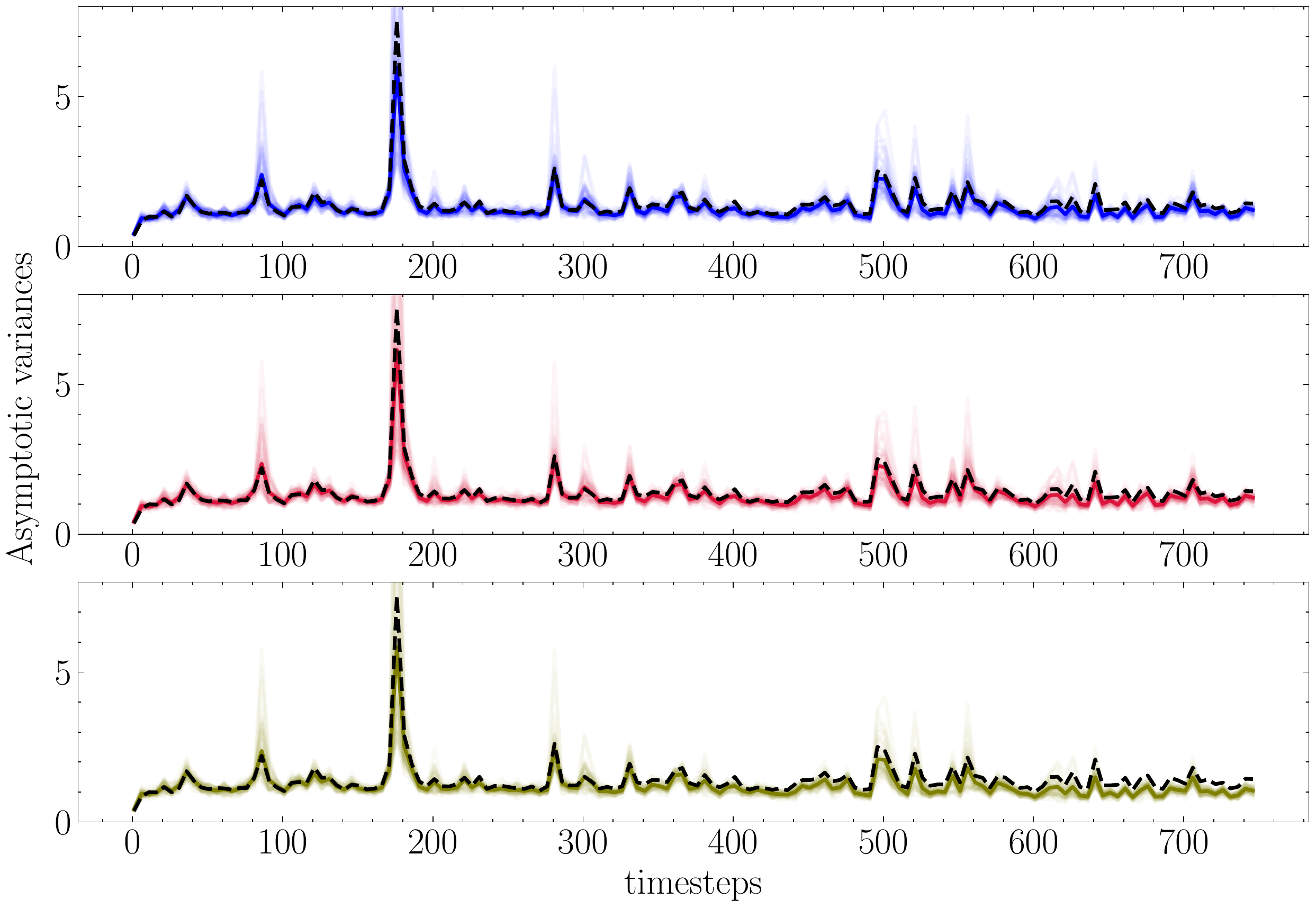}
    \caption{Long-term behavior of $\asymptvarestim{\eta,t}{\BS}$ (top), $\parisasymptvar{\eta,t}{\BS}$ with $M = 3$ (middle) and $\tbtasymptvarestim{\eta,t}{\BS}$ (bottom). The black dashed line is the asymptotic variance estimated using brute force. The number of particles is set to $N = 2000$. }
    \label{fig:comparison3VBS}
\end{figure}

\begin{figure}
    \centering
    \includegraphics[width = 12cm]{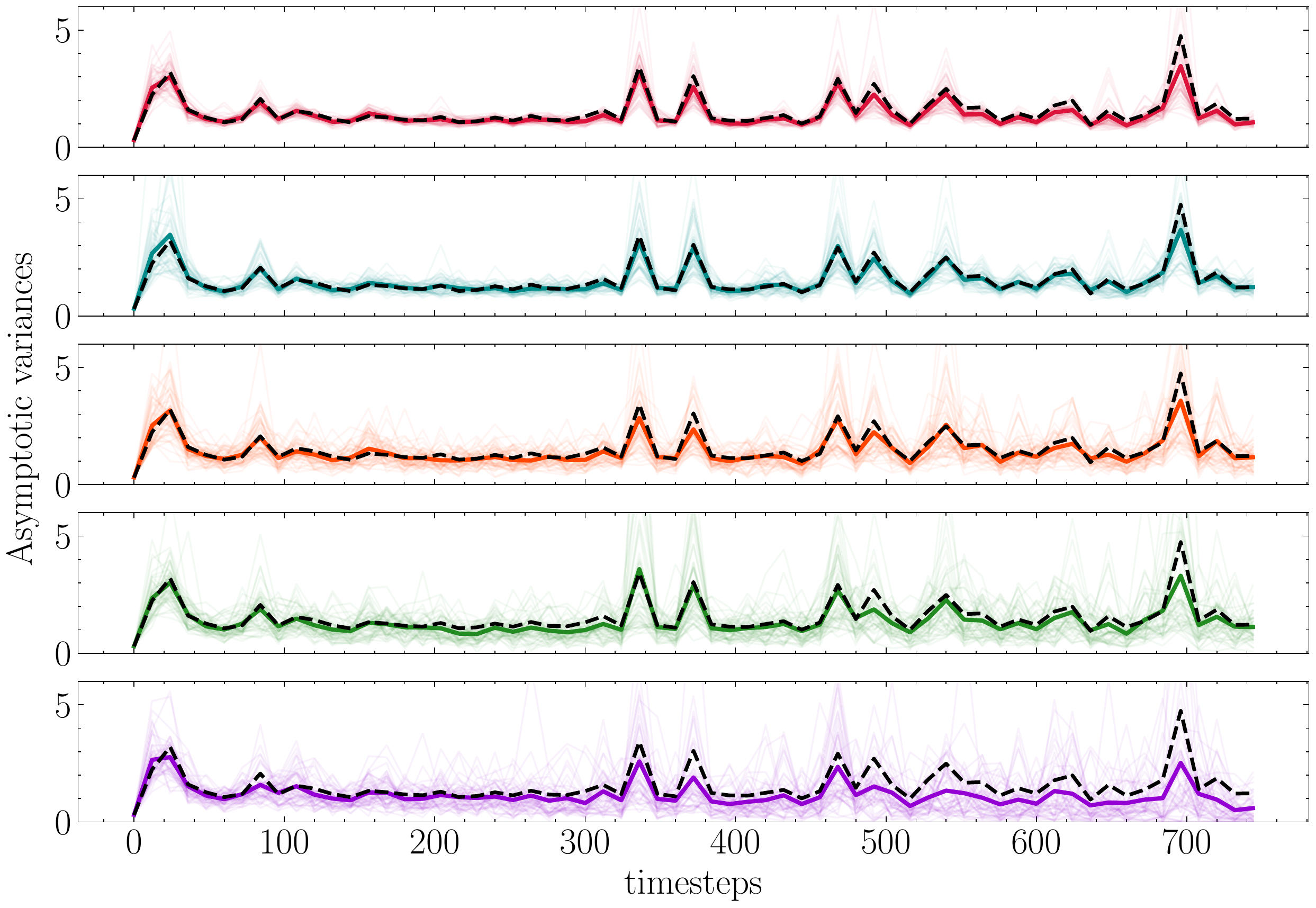}
    \caption{Long-term behavior of the asymptotic variance estimators up to $t = 750$. From top to bottom: $\paris$ version of $\asymptvarestim{\eta, t}{\BS}$ with $M = 3$, lagged estimators with (in order) $\lambda \in \{20,100,200,750\}$. The case $\lambda = 750$ corresponds to the CLE estimator. For each estimator, the blurred colored lines represent each run out of fifty runs and solid colored lines correspond to their average. The black dashed line is the asymptotic variance obtained by brute force. The number of particles $N$ is set to $2000$.
    }
    \label{fig:vbsvslag}
\end{figure}
\begin{figure}
    \centering
    \includegraphics[width = 13cm]{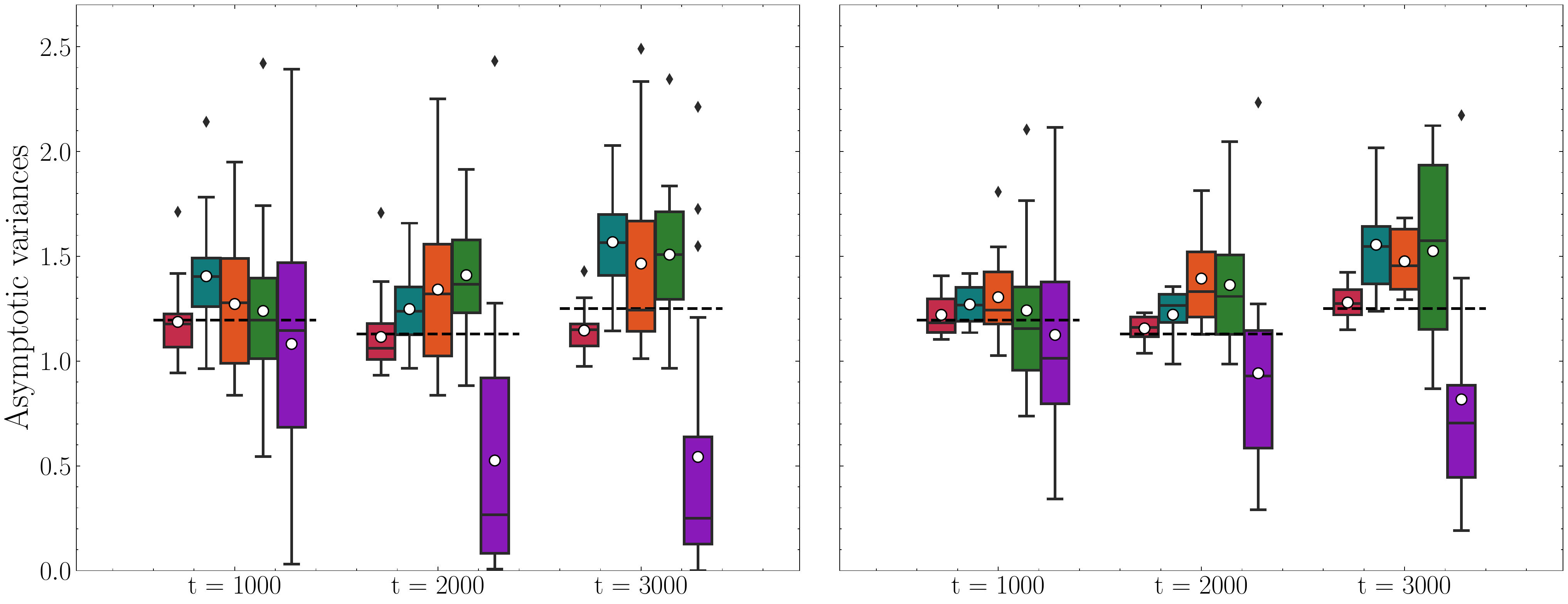}
     \caption{Long-term behavior of the asymptotic variance estimates up to $t = 3000$. White dots represent the average of the asymptotic variance estimates of each algorithm. The dashed black lines correspond to the asymptotic variances estimated by brute force. $N$ is set to $5000$ on the left boxplot and $10000$ on the right one. The boxplots at each time step from left to right are: $\parisasymptvar{\eta, t}{\BS}$ with $M = 3$ and then the lagged CLEs with $\lambda \in \{20,100,200,3000\}$.}
    \label{fig:boxplots}
\end{figure}
For the longer time horizon $t \in [3000]$ we set $N = 5000$ and picked three time steps in order to monitor the bias and variance closely, see Figure~\ref{fig:boxplots}. The variance of $\parisasymptvar{\eta,t}{\BS}$ remains steady while the bias increases gradually but slowly. This is attributed to the fact that our estimator is a ratio of two estimators with increasing bias and variance. Nonetheless, our estimator remains competitive with the best fixed lag estimator. Doubling the number of particles decreases both bias and variance, as highlighted by the  right plot. The computational cost of $\parisasymptvar{\eta,t}{\BS}$ is approximately twice larger than that of the genealogy tracing estimators when $N = 5000$ as shown in Figure~\ref*{fig:runtime} in the supplementary material. However, we are able to maintain a small variance and bias without having to tune any other parameter besides the sample size. In more complicated models or realistic scenarios, we do not know in advance which lag size is suitable and using an inappropriate lag might yield poor estimates. 
We further investigate the stability of our estimator with respect to $t$ by comparing
\[ D^\BS _N(t) \eqdef \sum_{i, j \in [N]^2} \backsum^\zero _t(i,j)  / N(N-1), \quad D^\GT _N(t) \eqdef \sum_{i, j \in [N]^2} \1_{\eve^i _{t,0} \neq \eve^j _{t,0}} / N(N-1) \eqsp,\]
which are central in the expression of the variance estimators (see Remark~\ref{rem:stability}). We also compare $E^\BS _N(t) \eqdef \big| (\mathcal{V}^{N, \BS} _{\eta, t}(\rm Id) / \asymptvar{t}(\rm Id)) -1 \big|$ and $E^\GT _N(t)$ which is defined in an analogous way.
For the CLE \eqref{eq:CLE}, although it is expected to collapse to $0$ after $\bigo(N)$ timesteps following  \cite{koskela2020asymptotic}, the estimator starts to exhibit high bias and variance much before  as the set of time $0$ ancestors depletes at a fast rate. This is illustrated on the left plot of Figure~\ref{fig:time_dep} where we fix $N$ to $1000$ and vary $t$ between $0$ and $3000$. We see that $D^\GT _N(t)$ decreases much faster than $D^\BS _N(t)$ and this in turn translates into longer stability for our estimator as can be seen on the right plot of the same figure. 

\begin{figure}
    \centering
    \includegraphics[width = 14cm]{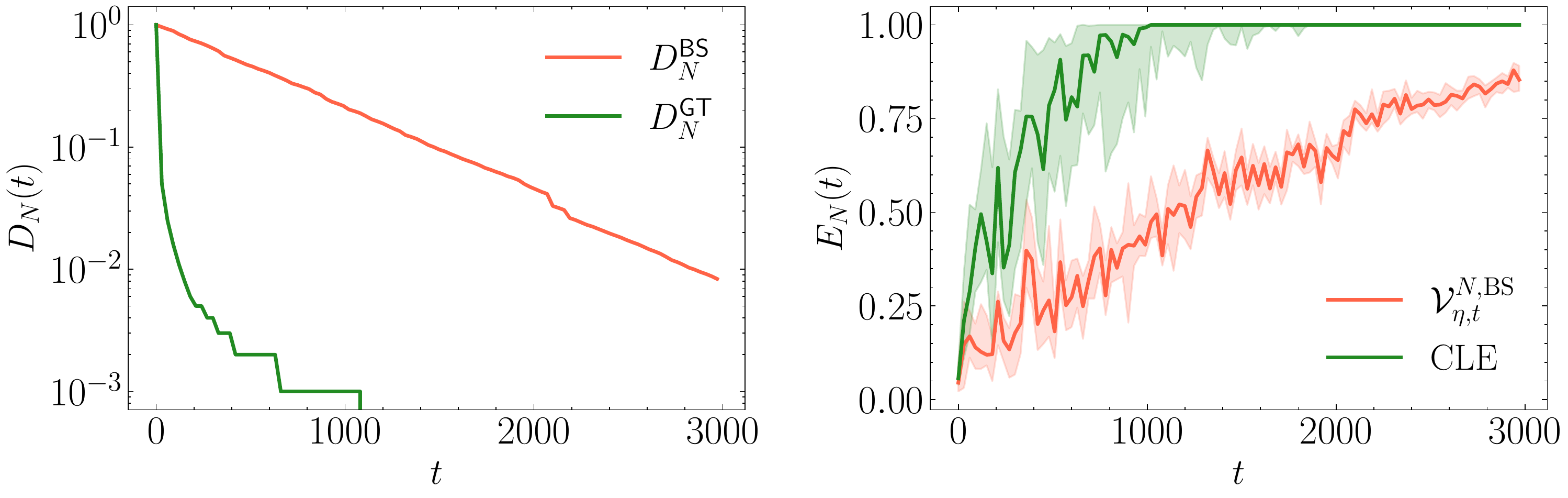}
    \caption{Dependency on the time $t$ of the variance estimators. The right plot displays the empirical error $E_N(t)$ for both $\BS$ and $\GT$ with $N$ fixed to $1000$. We display the median and the interquartile range over $30$ runs. The left plot displays the median of $D_N(t)$ associated with the $\BS$ and $\GT$ variance estimators used on the right plot.}
    \label{fig:time_dep}
\end{figure}

\subsection{Asymptotic variance of the smoother}
Here we are interested in the estimation of the asymptotic variance associated to the  FFBS estimates of the marginal means $\smooth{\ell}{t}(\mathrm{Id})$ with $\ell$ fixed and $t \geq \ell$ varying using Algorithm~\ref{alg:algFFBS}. For this example we sampled four different observation records of length $160$ each and $\ell$ is set to $100$. 
The real asymptotic variances of each $\smooth{\ell}{t}(\mathrm{Id})$ are intractable and they are estimated using $1000$ independent replicates of the marginal means $\smooth{\ell}{t}^N(\mathrm{Id})$ with $N = 10000$. We then multiply the obtained sample variance by $N$. The results are reported in Figure~\ref{fig:ffbs}. As expected, the crude estimates of the asymptotic variances all stagnate after some time $t$ due to the incoming observations becoming less and less informative as $t$ grows and thus no longer influencing the value of $\smooth{\ell}{t}(\mathrm{Id})$.

The estimator proposed in \ref{subsec:margsmoothing} captures well the behavior of the asymptotic variance with little variance and also stagnates at the same time. We observed in our experiments that, in comparison with the variance estimators of filtering algorithms, for this estimator to provide good performance more samples are required for shorter time horizons. 
Nonetheless, the increased computational time incurred by the increase of the number of particles is to be compared with the time it takes to compute the crude variance estimates up to $t = 160$, which is about $1$ hour when running on GPU. In comparison, one run of our estimator takes $3$ minutes. 

\section*{Acknowledgments}
The authors would like to thank the referees for the useful remarks that helped improve the paper. The first author also acknowledges the support of Institut Mines-Telecom through the Futur \& Ruptures PhD scholarship. 

\begin{figure}
    \centering
    \includegraphics[width = 12cm]{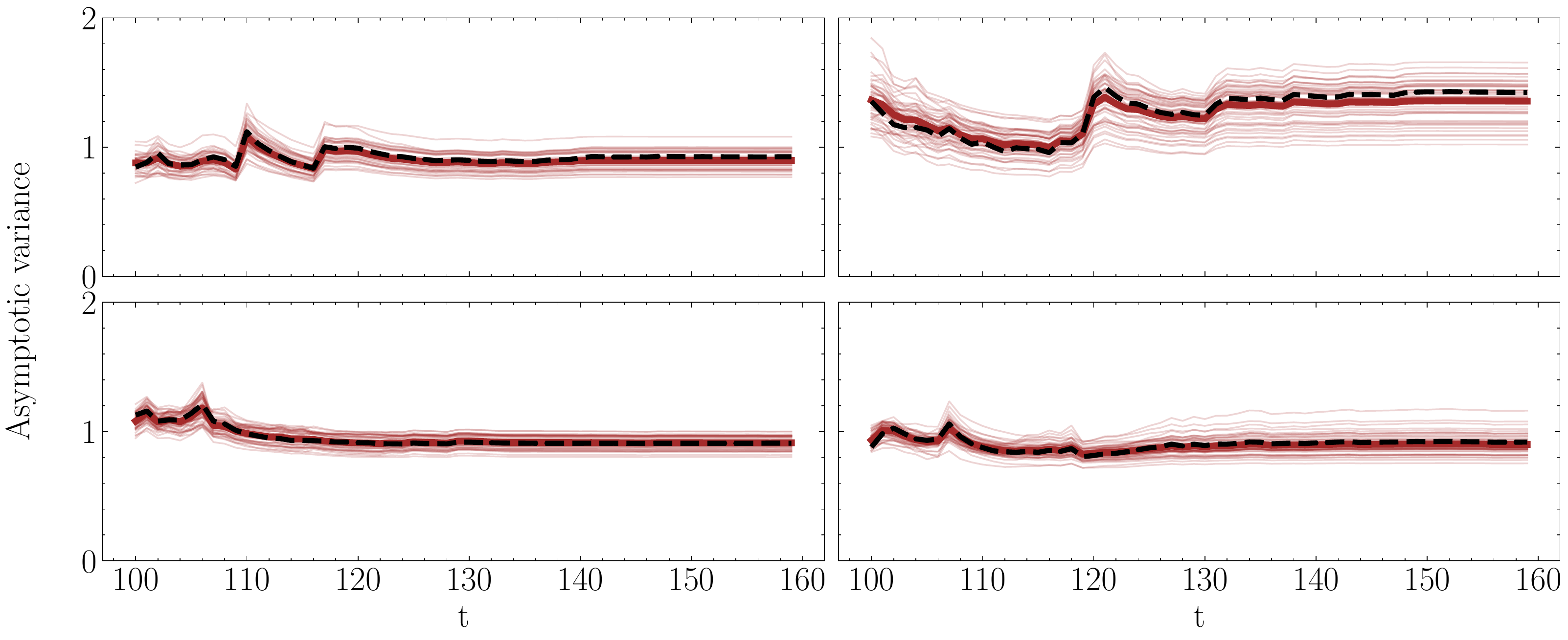}
    \caption{Asymptotic variance estimates for four different observation records of the marginal mean $\smooth{100}{t}^N(\mathrm{Id})$ where $t \in [100,160]$.  The blurred brown lines on the left plot represent 50 runs and the solid brown line their average. The black dashed line is the crude variance estimator. The number of particles $N$ is set to $5000$.}
    \label{fig:ffbs}
\end{figure}
\bibliographystyle{apalike}
\bibliography{bibliography}
\renewcommand{\thesection}{S\arabic{section}}
\renewcommand{\theequation}{S\arabic{section}.\arabic{equation}}

\appendix
\section{Proofs}
In this section we provide the proofs of Propositions~\ref*{prop:mu_expression}, \ref*{prop:identityQes} and Theorems~\ref*{thm:conv}, \ref*{corr:paris}, \ref*{thm:parisvar} and \ref*{thm:convFFBS}. The various Propositions and Lemmata used are stated and proved in Section~\ref{apdx:proofconv}-\ref{sec:supportparis}. Other intermediary technical results are provided in Section~\ref{sec:techres}. Equations, lemmata, propositions and theorems referred to without the prefix S are given in the main text.
\subsection{Preliminaries}

In order to simplify the notations, in what follows we write 
 \[\muest{\BS}{b,t}(h) =
 \sum_{k^{1:2} _{0:t} \in [N]^{2(t+1)}} \bigprod{b, t} (k^1 _{0:t}, k^2 _{0:t}) h(\particletraj{k^1}{t}, \particletraj{k^2}{t}) 
 \eqsp,\] where 
    \begin{equation}
       \label{def:bigprod}
        \bigprod{b, t} (k^1 _{0:t}, k^2 _{0:t}) \eqdef 
        \prod_{s = 0}^{t} N^{b_s} \bigg( \frac{N}{N-1}\bigg)^{1- b_s}
        \joint{t}^N(\boldone)^2 \intersect{b,t}{k^1 _{0:t}, k^2 _{0:t}} \Lambda^\BS _{1,t}(k^1 _{0:t}) \Lambda^\BS _{2,t}(k^1_{0:t}; k^2 _{0:t}) \eqsp.
    \end{equation}
    By \eqref{def:biglambda},
    \begin{align*}
       \Lambda^\BS _{1,t}(k^1 _{0:t}) & = \beta^\BS _t(k^1 _t, k^1 _{t-1}) \Lambda^\BS_{1,t-1}(k^1 _{0:t-1}) \eqsp, \\
     \Lambda^\BS _{2,t}(k^1 _{0:t}; k^2 _{0:t}) & = \big\{ \beta^\BS _t(k^2 _t, k^2 _{t-1}) \1_{k^2 _t \neq k^1 _t} + \normweight{k^2 _{t-1}}{t-1} \1_{k^2 _t = k^1 _t}\big\} \Lambda^\BS _{2,t-1}(k^1 _{0:t-1}; k^2 _{0:t-1}) \eqsp.
    \end{align*}
    and by \eqref{eq:jointdef} we have that $\joint{t}^N(\boldone) = \joint{t-1}^N(\boldone) N^{-1} \Omega _{t-1}$, hence, using \eqref{eq:mu_estimator} $\muest{\BS}{b,t}(h)$ becomes 
    \begin{align}
        \label{eq:alternativeQ}
        \muest{\BS}{b,t}(h) & = \sum_{k^{1:2} _{0:t} \in [N]^{2(t+1)}} \bigprod{b,t}(k^1 _{0:t}, k^2 _{0:t}) h(\particletraj{k^1}{t}, \particletraj{k^2}{t})\\
        & = \sum_{k^{1:2} _{0:t-1} \in [N]^{2t}} \bigprod{b,t-1}(k^1 _{0:t-1}, k^2 _{0:t-1})\frac{\Omega^2 _{t-1}}{N^2} \sum_{k^{1:2} _t \in [N]^2} \beta^\BS _t(k^1 _t, k^1 _{t-1})  \nonumber \\
        & \times \bigg[ \frac{N}{N-1} \beta^\BS _t(k^2 _t, k^2 _{t-1})\1_{k^1 _t \neq k^2 _t, b_t = 0} + N \normweight{k^2 _{t-1}}{t-1} \1_{k^1 _t = k^2 _t, b_t = 1}\bigg] h(\particletraj{k^1}{t}, \particletraj{k^2}{t}) \eqsp. \nonumber
    \end{align}
\subsection{Proof of Proposition~\ref*{prop:mu_expression}}
\label{proof:unbiasedest}
\begin{proof}[Proof of Lemma~\ref*{lem:BSGTidentity}]
    \label{proof:BSGTidentity}
    By $\As{assp:A}{}$, for any $(x,y) \in \Xset^2$, 
    \begin{align*}
        \beta^N _t(x,y) \filter{t-1}^N \transition{t}(\rmd x) & = \beta^N _t(x,y) \sum_{i = 1}^N \normweight{i}{t-1} \transition{t}(\particle{i}{t-1}, \rmd x) \\
        & = \frac{g_{t-1}(y) \transitiondens{t}(y,x)}{\sum_{i = 1}^N \weight{i}{t-1} \transitiondens{t}(\particle{i}{t-1}, x)} \sum_{i = 1}^N \normweight{i}{t-1} \transitiondens{t}(\particle{i}{t-1}, x) \nu(\rmd x) \\
        & = \frac{g_{t-1}(y)}{\Omega _{t-1}}  \transition{t}(y, \rmd x) \eqsp.
    \end{align*}
    Consequently, for any $(k^1 _{t-1}, k^1 _t) \in \Xset^2$ and $h \in \measurable{}$,
    \begin{align*}  
        \pE \big[ \beta^N _t( \particle{k _t}{t},\particle{k _{t-1}}{t-1}) h(\particle{k _t}{t}) \big| \F{t-1} \big] & = \int \beta^N _t(x, \particle{k _{t-1}}{t-1})
        h(x) \filter{t-1}^N \transition{t}(\rmd x) \\
        & = \int \frac{1}{\Omega _{t-1}} \weight{k^1 _{t-1}}{t-1} h(x) \transition{t}(\particle{k^1 _{t-1}}{t-1}, \rmd x) \\
        & = \normweight{k_{t-1}}{t-1} \transition{t}[h](\particle{k_{t-1}}{t-1}) \eqsp.
    \end{align*}
    On the other hand,
    \begin{align*}
        \pE \big[ \1_{k _{t-1} = A^{k _t} _{t-1}} h(\particle{k _t}{t}) \big| \F{t-1} \big] =  \int \sum_{i = 1}^N \1_{k _{t-1} = i} \normweight{i}{t-1} \transition{t}(\particle{i}{t-1}, \rmd x) h(x) = \normweight{k _{t-1}}{t-1} \transition{t}[h](\particle{k _{t-1}}{t-1}) \eqsp.
    \end{align*}
\end{proof}
\begin{proof}[Proof of Proposition~\ref*{prop:mu_expression}]
For the proof of i), note that conditionally on $\F{t-1}$, $\particle{1:N}{t}$ are i.i.d. with distribution $\filter{t-1}^N \transition{t}$, hence,  
\begin{multline*}
    \pE \left[ \sum_{k^{1:2} _t \in [N]^2} \1_{k^1 _t \neq k^2 _t} \beta^\BS _t(k^1 _t, k^1 _{t-1}) \beta^\BS _t(k^2 _t, k^2 _{t-1}) h(\particletraj{k^1}{t}, \particletraj{k^2}{t}) \bigg| \F{t-1}\right] \\
    =  \sum_{k^{1:2} _t \in [N]^2} \int \1_{k^1 _t \neq k^2 _t} \beta^\BS _t(k^1 _t, k^1 _{t-1}) \beta^\BS _t(k^2 _t, k^2 _{t-1})h(\particletraj{k^1}{t}, \particletraj{k^2}{t}) \filter{t-1}^N \transition{t}(\rmd \particle{k^1 _t}{t}) \filter{t-1}^N \transition{t}(\rmd \particle{k^2 _{t}}{t}) \eqsp,
\end{multline*}
and by \eqref{lem:identity} in Lemma~\ref*{lem:BSGTidentity},
\begin{align}
    \label{eq:condexpectb0}
    & \pE \left[ \sum_{k^{1:2} _t \in [N]^2} \1_{k^1 _t \neq k^2 _t} \beta^\BS _t(k^1 _t, k^1 _{t-1}) \beta^\BS _t(k^2 _t, k^2 _{t-1}) h(\particletraj{k^1}{t}, \particletraj{k^2}{t}) \bigg| \F{t-1}\right] \nonumber \\
    & \hspace{2cm} = \frac{N(N-1)}{\Omega _{t-1} ^2} \big(g^{\otimes 2} _{t-1}\bitransition{0}{t}[h]\big) (\particletraj{k^1}{t-1}, \particletraj{k^2}{t-1}) \eqsp,
\end{align}
where \[g_{t-1} ^{\otimes 2} \bitransition{0}{t}[h] : (x_{0:t-1}, x' _{0:t-1}) \mapsto g^{\otimes 2} _{t-1}(x_{t-1}, x' _{t-1}) \int h(x_{0:t}, x' _{0:t}) \transition{t}(x_{t-1}, \rmd x_t)\transition{t}(x'_{t-1}, \rmd x' _t) \eqsp.\]
On the other hand, by \eqref{lem:identity} in Lemma~\ref*{lem:BSGTidentity},
\begin{align}
    \label{eq:condexpectb1}
    &  \pE \left[ \sum_{k^{1:2} _t \in [N]^2} \1_{k^1 _t = k^2 _t} \beta^\BS _t(k^1 _t, k^1 _{t-1}) \normweight{k^2 _{t-1}}{t-1} h(\particletraj{k^1}{t}, \particletraj{k^2}{t}) \bigg| \F{t-1}\right] \nonumber \\
    & \hspace{1cm} =  \sum_{k^{1:2} _t \in [N]^2} \int \1_{k^1 _t = k^2 _t} \beta^\BS _t(k^1 _t, k^1 _{t-1}) \normweight{k^2 _{t-1}}{t-1} h(\particletraj{k^1}{t}, \particletraj{k^2}{t}) \filter{t-1}^N \transition{t}(\rmd \particle{k^1 _t}{t}) \delta_{\particle{k^1 _{t}}{t}} (\rmd \particle{k^2 _{t}}{t})  \nonumber \\
    & \hspace{1cm} = \frac{1}{\Omega _{t-1} ^2} \sum_{k^{1:2} _t \in [N]^2} \1_{k^1 _t = k^2 _t}\weight{k^1 _{t-1}}{t-1} \weight{k^2 _{t-1}}{t-1} \int h(\particletraj{k^1}{t}, \particletraj{k^2}{t}) \transition{t}(\particle{k^1 _{t-1}}{t-1},\rmd \particle{k^1 _t}{t}) \delta_{\particle{k^1 _t}{t}} (\rmd \particle{k^2 _{t}}{t}) \nonumber \\
    & \hspace{1cm} = \frac{N}{\Omega^2 _{t-1}} \big( g^{\otimes 2} _{t-1} \bitransition{1}{t}[h] \big)(\particletraj{k^1}{t}, \particletraj{k^2}{t}) \eqsp,
\end{align}
where 
\[ 
g^{\otimes 2} _{t-1} \bitransition{1}{t}[h] : (x_{0:t-1}, x' _{0:t-1}) \mapsto g^{\otimes 2} _{t-1}(x_{t-1}, x' _{t-1}) \int h(x_{0:t}, x' _{0:t}) \transition{t}(x_{t-1}, \rmd x_t) \delta _{x _t}(\rmd x' _t) \eqsp.
\]
 If $b_t = 0$, since  $\Omega _{t-1}$ and $\bigprod{b,t-1}(k^1 _{0:t-1}, k^2 _{0:t-1})$ are $\F{t-1}$ measurable for any $(k^1 _{0:t-1}, k^2 _{0:t-1}) \in [N]^{2t}$, by \eqref{eq:alternativeQ} and \eqref{eq:condexpectb0},
\begin{align*}
     \pE \left[ \muest{\BS}{b,t}(h) \big| \F{t-1} \right] & =  \sum_{k^{1:2} _{0:t-1} \in [N]^2} \bigprod{b,t-1}(k^1 _{0:t-1}, k^2 _{0:t-1}) \frac{\Omega^2 _{t-1}}{N(N-1)} \\
    & \hspace{1cm}\times \pE \bigg[ \sum_{k^{1:2} _t \in [N]^2} \1_{k^1 _t \neq k^2 _t} \beta^\BS _t(k^1 _t, k^1 _{t-1}) \beta^\BS _t(k^2 _t, k^2 _{t-1}) h(\particletraj{k^1}{t}, \particletraj{k^2}{t}) \bigg| \F{t-1}\bigg]  \\
    & = \sum_{k^{1:2} _{0:t-1} \in [N]^2} \bigprod{b,t-1}(k^1 _{0:t-1}, k^2 _{0:t-1}) \big(g^{\otimes 2} _{t-1}\bitransition{0}{t}[h]\big) (\particletraj{k^1}{t-1}, \particletraj{k^2}{t-1}) \\
    & = \muest{\BS}{b,t-1}(g^{\otimes 2} _{t-1} \bitransition{0}{t}[h]) \eqsp.
\end{align*}
If $b_t = 1$, again by \eqref{eq:alternativeQ} and \eqref{eq:condexpectb1},
\begin{align*}
    \pE \left[ \muest{\BS}{b,t}(h) \big|  \F{t-1}\right] & = \sum_{k^{1:2} _{0:t-1} \in [N]^2} \bigprod{b,t-1}(k^1 _{0:t-1}, k^2 _{0:t-1}) \frac{\Omega^2 _{t-1}}{N} \\
    & \hspace{2cm}\times \pE \bigg[ \sum_{k^{1:2} _t \in [N]^2} \1_{k^1 _t = k^2 _t} \beta^\BS _t(k^1 _t, k^1 _{t-1}) \normweight{k^2 _{t-1}}{t-1} h(\particletraj{k^1}{t}, \particletraj{k^2}{t}) \bigg| \F{t-1}\bigg] \\
    & = \sum_{k^{1:2} _{0:t-1} \in [N]^2} \bigprod{b,t-1}(k^1 _{0:t-1}, k^2 _{0:t-1}) \big(g^{\otimes 2} _{t-1}\bitransition{1}{t}[h]\big) (\particletraj{k^1}{t-1}, \particletraj{k^2}{t-1})  \\
    & = \muest{\BS}{b,t-1}\big( g^{\otimes 2} _{t-1}\bitransition{1}{t}[h] \big) \eqsp.
\end{align*}
For the proof of ii), we proceed by induction. Let $t = 0$ and $h \in \measurable{2}$. If $b_0 = 0$, since $\particle{1:N}{0} \iid \transition{0}$,
\begin{align*}
    \pE \left[ \muest{\BS}{0,0}(h) \right] & = \pE \left[ \frac{1}{N(N-1)} \sum_{i,j \in [N]^2} \1_{i \neq j} h(\particle{i}{0}, \particle{j}{0}) \right] \\
    & = \frac{1}{N(N-1)} \sum_{i, j \in [N]^2} \1_{i \neq j} \bitransition{0}{0}[h] = \bitransition{0}{0}[h] = \mumeasure{b,0}(h) \eqsp.
\end{align*}
If $b_0 = 1$,
\begin{align*}
    \pE \left[ \muest{\BS}{1,0}(h) \right] = \pE \left[ \frac{1}{N} \sum_{i = 1}^N h(\particle{i}{0}, \particle{i}{0}) \right] = \bitransition{1}{0}[h] = \mumeasure{b,0}(h) \eqsp.
\end{align*}
Let $t \in \N^{*}$ and $h \in \measurable{2(t+1)}$. Assume that $\pE \big[ \muest{\BS}{b,t-1}(f)\big] = \mumeasure{b,t-1}(f)$ for any $b \in \Bset_t$ and $f \in \measurable{2t}$. 
By \ref*{item:condexpect} in Proposition~\ref*{prop:mu_expression} and the tower property
\[ 
    \pE \left[ \muest{\BS}{b,t}(h) \right] = \pE \left[ \pE \big[ \muest{\BS}{b,t}(h) \big| \F{t-1} \big]\right] = \pE \left[ \muest{\BS}{b,t-1}\big(g^{\otimes 2} _{t-1} \bitransition{b_t}{t}[h]\big) \right]
    \]
 and $g^{\otimes 2} _{t-1} \bitransition{b_t}{t}[h] \in \measurable{2t}$. Thus, by the induction hypothesis and the definition of $\mumeasure{b,t-1}$ we get
\[ 
    \pE \left[ \muest{\BS}{b,t}(h) \right] = \pE \left[ \muest{\BS}{b,t-1}(g^{\otimes 2} _{t-1} \bitransition{b_t}{t}[h]) \right] = \mumeasure{b,t-1}\big( g^{\otimes 2} _{t-1} \bitransition{b_t}{t}[h] \big) = \mumeasure{b,t}(h) \eqsp,
\]
which completes the proof. The proof of iii) is a direct consequence of \eqref{eq:scdestimator}, \eqref{eq:general_termbyterm} and ii). 
\end{proof}

\subsection{Proof of Proposition~\ref*{prop:identityQes}}
\label{proof:identityQes}
    Let $h_{0:t}$ be an additive functional \eqref{eq:additive}. By definition, for $s \in [t]$,
    \begin{align*}
        \GG{s,t}(x_s, h_{0:t}) & =  \int \big\{ \tilde{h}_{0:s}(x_{0:s}) + \tilde{h}_{s:t}(x_{s:t}) \big\} \bwpath{s}(x_s, \rmd x_{0:s-1}) \Q{s+1:t}(x_s, \rmd x_{s+1: t})\\
        & = \int \big\{ \bwpath{s}[\tilde{h}_{0:s}](x_s) + \tilde{h}_{s:t}(x_{s:t}) \big\} \Q{s+1:t}(x_s, \rmd x_{s+1:t}) \eqsp,
    \end{align*}
and then, setting $\mathsf{H}_{s:t} : x_{s:t} \mapsto\bwpath{s}[\tilde{h}_{0:s}](x_s) + \tilde{h}_{s:t}(x_{s:t})$ we get
\begin{align*}
    & \joint{s}(\boldone)\joint{s}\big(\GG{s,t}[h_{0:t}]^2) \\
    &  = \joint{s}(\boldone)\joint{s}\big( \Q{s+1:t}\big[ \bwpath{s}[\tilde{h}_{0:s}] + \tilde{h}_{s:t} \big]^2\big) \\
    & = \int \joint{0:s-1}(\rmd x' _{0:s-1}) g_{s-1}(x'_{s-1}) \int  \Q{s+1:t}[\mathsf{H}_{s:t}](x_s) \Q{s+1:t}[\mathsf{H}_{s:t}](x_s) \joint{0:s}(\rmd x_{0:s}) \\
    & = \int \joint{0:s-1}(\rmd x' _{0:s-1}) g_{s-1}(x' _{s-1}) \joint{0:s}(\rmd x_{0:s}) \delta_{x_s}(\rmd x' _s) \Q{s+1:t}[\mathsf{H}_{s:t}](x_s) \Q{s+1:t}[\mathsf{H}_{s:t}](x'_s) \eqsp,
\end{align*}
which establishes the result since by definition
\begin{multline*}
    \mumeasure{e_s,t}\big( \rmd x_{0:t}, \rmd x'_{0:t} \big) 
    = \joint{0:s}(\rmd x_{0:s}) \joint{0:s-1}(\rmd x' _{0:s-1}) g_{s-1}(x' _{s-1}) \delta_{x_s}(\rmd x'_s)
    \\\Q{s+1:t} \big(x_s,\rmd x_{s+1:t}\big) \Q{s+1:t}
    \big(x'_s,\rmd x'_{s+1:t}\big) \eqsp.
\end{multline*}
If $s = 0$, then $\GG{0,t}[h_{0:t}](x_0) = \int h_{0:t}(x_{0:t}) \Q{1:t}(x_0, \rmd x_{1:t})$ and 
\begin{align*}
    \joint{0}(\boldone) \joint{0}\big( \GG{0:t}[h_{0:t}]^2 \big) = \joint{0}(\GG{0:t}[h_{0:t}]^2) & = \int \transition{0}(\rmd x_0) \Q{1:t}[h_{0:t}](x_0) \Q{1:t}[h_{0:t}](x_0) \\
    & = \int \transition{0}(\rmd x_0) \delta_{x_0}(\rmd x'_0) \Q{1:t}[h_{0:t}](x_0) \Q{1:t}[h_{0:t}](x' _0) \\
    & = \mumeasure{e_0, t}(h_{0:t} ^{\otimes 2}) \eqsp.
\end{align*}
\subsection{Proof of Theorem~\ref*{thm:conv}}
Let $m \in \N^{*}$ and $N \geq 2$. Define 
 \begin{align}
     \label{def:Iset}
     \mathcal{I}^m _0 & \eqdef \{k^{1:2m} \in [N ]^{2m}: k^{2i - 1} \neq k^{2i} , i \in [1:m]\} \eqsp, \\
     \mathcal{I}^m _1 & \eqdef \{k^{1:2m} \in [N ]^{2m}: k^{2i - 1} = k^{2i}, i \in [1:m]\} \eqsp.
 \end{align}
 Define also for any $p \in [2m]$, \[\sett^p _m \eqdef \{ k^{1:2m} \in [N ]^{2m}: \card{\{k^1, k^2, \cdots, k^{2m-1}, k^{2m}\}} = p \} \eqsp.\]
 Then $[N]^{2m} = \bigsqcup_{p = 1}^{2m} \sett^p _m$ and 
 \begin{align}
     \label{def:Sset}
     \mathcal{I}^m _0 = \bigsqcup_{p = 1}^{2m} \mathcal{I}^m _0 \cap \sett^p _m = \bigsqcup_{p = 2}^{2m} \mathcal{I}^m _0 \cap \sett^p _m \eqsp, \quad
   \mathcal{I}^m _1 =  \bigsqcup_{p = 1}^{2m} \mathcal{I}^m _1 \cap \sett^p _m = \bigsqcup_{p = 1}^{m} \mathcal{I}^m _1 \cap \sett^p _m\eqsp,
 \end{align}
 where $\bigsqcup$ means disjoint union.
 The first equality holds because the tuples in $\mathcal{I}^m _0$ must contain at least two different values and the second because tuples in $\mathcal{I}^m _1$ contain at most $m$ different values. The proof of Theorem~\ref*{thm:conv} is concerned with $m =2$ and that of Proposition~\ref{prop:Qbound} with $m \geq 2$. Example~\ref{ex:explicitset} provides the intersections for the case $m = 2$. 
\begin{example}
     \label{ex:explicitset}
     Choose $m = 2$ and $N \geq 4$. Then,
     \begin{align*}
         \mathcal{I}^2 _0 \cap \sett^2 _2 & = \{k^{1:4} \in [N]^4 : k^1 \neq k^2, k^3 \neq k^4,\{k^3, k^4\} = \{k^1, k^2\} \}\eqsp,\\
         \mathcal{I}^2 _0 \cap \sett^3 _2 & = \{k^{1:4} \in [N]^4 : k^1 \neq k^2, k^3 \neq k^4, k^3 \in \{k^1, k^2\}, k^4 \notin \{k^1, k^2\} \} \eqsp, \\ 
         & \hspace{1cm} \sqcup \{k^{1:4} \in [N]^4 : k^1 \neq k^2, k^3 \neq k^4, k^3 \notin \{k^1, k^2\}, k^4 \in \{k^1, k^2\} \} \eqsp,\\
        \mathcal{I}^2 _0 \cap \sett^4 _2 & = \{k^{1:4} \in [N]^4 : k^1 \neq k^2 \neq k^3 \neq k^4 \} \eqsp,
     \end{align*}
    with $\card{\mathcal{I}^2 _0 \cap \sett^2 _2} = 2N(N-1)$, $\card{\mathcal{I}^2 _0 \cap \sett^3 _2} = 4N(N-1)(N-2)$ and $\card{\mathcal{I}^2 _0 \cap \sett^4 _2} = N(N-1)(N-2)(N-3)$. As for $\mathcal{I}^2 _1$,
     \begin{align*}
         \mathcal{I}^2 _1 \cap \sett^1 _2 & = \{k^{1:4} \in [N]^4 : k^1 = k^2 = k^3 = k^4 \}\eqsp, \\
         \mathcal{I}^2 _1 \cap \sett^2 _2 & = \{k^{1:4} \in [N]^4 : 
         k^1 = k^2, k^3 = k^4, k^1 \neq k^3 \} \eqsp.
     \end{align*}
     with $\card{\mathcal{I}^2 _1 \cap \sett^1 _2} = N$ and $\card{\mathcal{I}^2 _1 \cap \sett^2 _2} = N(N-1).$
 \end{example}
\begin{proof}[Proof of Theorem~\ref*{thm:conv}]
\label{proof:conv}
We proceed by induction. Throughout the proof we assume that $N \geq 4$ for the sake of simplicity. For $t = 0$ and $b = \zero$, using \eqref{lem:bigtauexpr} and \eqref{eq:Qexpr},   \[ \muest{\BS}{\zero, 0}(h) = N^{-1}(N-1)^{-1} \sum_{i, j \in [N]^2} \1_{i \neq j} h(\particle{i}{0}, \particle{j}{0})\eqsp, \] and $\mumeasure{\zero, 0}(h) = \transition{0}^{\otimes 2}(h)$. 
\begin{align*}
    & \big\| \muest{\BS}{\zero,0}(h) - \mumeasure{\zero,0}(h) \big\|^2 _2  \\
    &\hspace{.5cm} = \pE \bigg[ \frac{1}{N^2(N-1)^2} \sum_{i,j,i',j' \in [N]^4} \1_{i \neq j, i' \neq j'} \big\{ 
         h(\particle{i}{0}, \particle{j}{0}) - \mumeasure{\zero,0}(h) \big\} \big\{
            h(\particle{i'}{0}, \particle{j'}{0}) - \mumeasure{\zero,0}(h) \big\} \bigg] \\
    &\hspace{.5cm} = \frac{\tau_0 + \overline{\tau}_0}{N^2(N-1)^2}  \eqsp,
\end{align*}
where
\begin{align*}
    \tau_0 &= \pE \left[ \sum_{i,j,i',j' \in \mathcal{I}^2 _0 \cap \sett^4 _2} \{ 
          h(\particle{i}{0}, \particle{j}{0}) - \mumeasure{\zero,0}(h) \} \{
           h(\particle{i'}{0}, \particle{j'}{0}) - \mumeasure{\zero,0}(h) \}\right]\eqsp,\\
    \overline{\tau}_0 & = \pE \left[ \sum_{i,j,i',j' \in \mathcal{I}^2 _0 \cap \overline{\sett^4 _2}} \{ 
          h(\particle{i}{0}, \particle{j}{0}) - \mumeasure{\zero,0}(h) \} \{
           h(\particle{i'}{0}, \particle{j'}{0}) - \mumeasure{\zero,0}(h) \}\right]\eqsp,
\end{align*}
where $\mathcal{I}^2 _0 \cap \sett^4 _2$ is defined in \eqref{def:Iset}, \eqref{def:Sset} and explicited in Example~\ref{ex:explicitset}, and $\mathcal{I}^2 _0 \cap \overline{\sett^4 _2} = (\mathcal{I}^2 _0 \cap \sett^2 _2) \sqcup (\mathcal{I}^2 _0 \cap \sett^3 _2)$.
If $(i,j,i',j') \in \mathcal{I}^2 _0 \cap \sett^4 _2$, then $\particle{i}{0}, \particle{j}{0}, \particle{i'}{0}$ and $\particle{j'}{0}$ are i.i.d. Therefore,
$$
\pE \left[ \big\{ 
    h(\particle{i}{0}, \particle{j}{0}) - \mumeasure{\zero,0}(h) \big\} \big\{
       h(\particle{i'}{0}, \particle{j'}{0}) - \mumeasure{\zero,0}(h) \big\} \right]  = \pE \left[ \big\{ 
        h(\particle{i}{0}, \particle{j}{0}) - \mumeasure{\zero,0}(h) \big\} \right]^2 = 0 \eqsp,
$$
and $\tau_0 = 0$. Hence, using the fact that because $h$ is bounded, $\| h - \mumeasure{\zero,0}(h) \|_\infty \leq 2 \|h \|_\infty$ and that $\card{\mathcal{I}^2 _0 \cap \overline{\sett^4 _2}} = 4N(N-1)(N-2) + 2N(N-1)$ by Example~\ref{ex:explicitset}, we get 
\begin{align*}
    \big\| \muest{\BS}{\zero,0}(h) - \mumeasure{\zero,0}(h) \big\|^2 _2 &\leq \frac{\overline{\tau}_0 }{N^2 (N-1)^2}\\
    & \leq \frac{4\|h\|^2 _\infty}{N^2 (N-1)^2} \sum_{i,j,i',j' \in \mathcal{I}^2 _0 \cap \overline{\sett^4 _2}} 1\eqsp, \\
    & = \frac{4 \| h \|^2 _\infty \big[ 4N(N-1)(N-2) + 2N(N-1)\big]}{N^2(N-1)^2} = \bigo(N^{-1})\eqsp,
\end{align*}
which completes the proof for $b=0$. For $b = 1$, 
\begin{align*}
    &  \muest{\BS}{1,0}(h) - \mumeasure{1,0}(h) 
    = N^{-1} \sum_{i = 1}^N h(\particle{i}{0}, \particle{i}{0}) - \int h(x,x) \transition{0}(\rmd x) \eqsp,
\end{align*}
and since $h$ is bounded, 
$$ \big\|  \muest{\BS}{1,0}(h) - \mumeasure{1,0}(h) \big\|^2 _2 
        = N^{-1} \pV_{\transition{0}} \big[ h(\xi, \xi) \big] = \bigo(N^{-1})\eqsp,$$ 
where $\pV _{\transition{0}}$ is the variance under $\transition{0}$. This completes the proof for $t = 0$. Let $t>0$, and assume now that \eqref{eq:convlim} holds at time $t - 1$. Consider the following decomposition \begin{align*}
    \label{eq:decomp}
    \muest{\BS}{b,t}(h) - \mumeasure{b,t}(h) 
    = \muest{\BS}{b,t}(h) - \pE \big[ \muest{\BS}{b,t}(h) \big| \F{t-1} \big]
     + \pE \big[ \muest{\BS}{b,t}(h) \big| \F{t-1} \big] - \mumeasure{b,t}(h) \eqsp,
    \end{align*}
    which, by Proposition~\ref*{prop:mu_expression}, becomes
    \begin{multline}
        \muest{\BS}{b,t}(h) - \mumeasure{b,t}(h)  = \muest{\BS}{b,t}(h) - \muest{\BS}{b,t-1}(g_{t-1} ^{\otimes 2} \bitransition{b_{t}}{t}[h]) \\ + \muest{\BS}{b,t-1}(g_{t-1} ^{\otimes 2} \bitransition{b_{t}}{t}[h]) - \mumeasure{b,t-1}(g^{\otimes 2} _{t-1} \bitransition{b_{t}}{t}[h]) \eqsp. 
    \end{multline}
    By the induction hypothesis, since $h$ is bounded and also $g_{t-1}$ by $\As{assp:B}{}$, we have 
    \[
    \simplelim \left\| \muest{\BS}{b,t-1}(g_{t-1} ^{\otimes 2} \bitransition{b_{t}}{t}[h]) - \mumeasure{b,t-1}(g^{\otimes 2} _{t-1} \bitransition{b_{t}}{t}[h]) \right\| _2 = 0 \eqsp,
    \] 
    hence, by Minkowski's inequality it remains to prove that 
    \begin{equation}
        \simplelim \left\|  \muest{\BS}{b,t}(h) - \muest{\BS}{b,t-1}(g_{t-1} ^{\otimes 2} \bitransition{b_{t}}{t}[h]) \right\| _2 = 0 \eqsp.
    \end{equation}
    By Proposition~\ref*{prop:mu_expression}, $\pE \big[ \muest{\BS}{b,t}(h) \big| \F{t-1} \big] = \muest{\BS}{b,t-1}( g^{\otimes 2} _{t-1} \bitransition{b_t}{t}[h])$ and
    \begin{equation*}
        \pE \big[ \muest{\BS}{b,t}(h) \muest{\BS}{b,t-1}(g^{\otimes 2} _{t-1} \bitransition{b_t}{t}[h]) \big] = \pE \big[ \muest{\BS}{b,t-1}(g^{\otimes 2} _{t-1} \bitransition{b_t}{t}[h])^2 \big] \eqsp,
    \end{equation*}
    hence,
    \begin{align*}
         \big\| \muest{\BS}{b,t}(h) - \muest{\BS}{b,t}( g^{\otimes 2} _{t-1} \bitransition{b_t}{t}[h]) \big\|^2 _2 = \big\| \muest{\BS}{b,t}(h) \big\|_2 ^2 - \big\| \muest{\BS}{b,t-1}(g^{\otimes 2} _{t-1} \bitransition{b_t}{t}[h]) \big\|_2 ^2 \eqsp.
    \end{align*}
    Consequently, by Proposition~\ref{prop:norm2QBS}, if $b_t = 0$,
        \begin{align}
          \label{eq:finalnorm2diffb0}
            &\left\| \muest{\BS}{b,t}(h) - \muest{\BS}{b,t-1}(g^{\otimes 2} _{t-1} \bitransition{0}{t}[h])\right\|^2 _2 \\
            & \leq \left( \frac{(N-2)(N-3)}{N(N-1)} - 1 \right) \left\|\muest{\BS}{b,t-1}(g^{\otimes 2} _{t-1}\bitransition{0}{t}[h]) \right\|_2 ^2 \nonumber\\
            & + \frac{N-2}{N-1} \boundg^3 | h |^2 _\infty \int \nu(\rmd x)\pE \bigg[ \frac{\Omega_{t-1}}{N} \bigg\{ \muest{\BS}{b,t-1}(\transitiondens{t}(.,x) \otimes \boldone) \muest{\BS}{b,t-1}(\beta^N _{t}(x,.) \otimes \boldone) \nonumber\\
            & + \muest{\BS}{b,t-1}(\boldone \otimes \transitiondens{t}(.,x)) \muest{\BS}{b,t-1}(\beta^N _{t}(x,.) \otimes \boldone) + \muest{\BS}{b,t-1}(\transitiondens{t}(.,x) \otimes \boldone) \nonumber \\
             & \hspace{1cm}\times \muest{\BS}{b,t-1}(\boldone \otimes \beta^N _{t}(x,.)) 
             + \muest{\BS}{b,t-1}(\boldone \otimes \transitiondens{t}(.,x)) \muest{\BS}{b,t-1}(\boldone \otimes \beta^N _{t}(x,.)) \bigg] \nonumber \\
             & + \boundg^2 |h |^2 _\infty \int \nu^{\otimes 2}(\rmd y,\rmd x)\pE \bigg[ \frac{ \Omega_{t-1} ^2}{N(N-1)} \nonumber \\
            & \hspace{1cm}\times \bigg\{
                 \muest{\BS}{b,t-1}(\transitiondens{t}(.,x) \otimes \transitiondens{t}(.,y)) \muest{\BS}{b,t-1}\left( \beta^N _{t}(x,.) \otimes \beta^N _{t}(y,.) \right) \nonumber \\
             & \hspace{2cm} + \muest{\BS}{b,t-1}(\transitiondens{t}(.,x) \otimes \transitiondens{t}(.,y))  \muest{\BS}{b,t-1}\left( \beta^N _{t}(y,.) \otimes \beta^N _{t}(x,.) \right) \bigg\}\bigg] \eqsp, \nonumber
        \end{align}
        and if $b_t = 1$, 
        \begin{multline}   
             \label{eq:finalnorm2diffb1}
            \left\| \muest{\BS}{b,t}(h) - \muest{\BS}{b,t-1}(g^{\otimes 2} _{t-1} \bitransition{1}{t}[h])\right\|^2 _2 
             \leq \left( \frac{N-1}{N} - 1\right) \| \muest{\BS}{b,t-1}(g^{\otimes 2} _{t-1}\bitransition{1}{t}[h]) \|^2 _2  \\
             + \boundg^3 |h|_\infty ^2 \int \pE \left[ \frac{\Omega _{t-1}}{N} \muest{\BS}{b,t-1}(\transitiondens{t}(.,x) \otimes \boldone) \muest{\BS}{b,t-1}(\beta^N _{t}(x,.) \otimes \boldone) \right] \nu(\rmd x) \eqsp. 
        \end{multline}
        By Proposition~\ref{prop:Qbound}, the first term in the r.h.s. of \eqref{eq:finalnorm2diffb0} and \eqref{eq:finalnorm2diffb1} is $\bigo(N^{-1})$ in both cases because 
$\big| g^{\otimes 2} _{t-1}\bitransition{b_{t}}{t}[h] \big|_\infty \leq \boundg^2 |h|_\infty < \infty$. We now show that the remaining terms go to zero when $N$ goes to infinity. Define for any $x \in \Xset$ and $N \in \N^{*}$, 
\begin{equation}
\begin{aligned}
    \label{eq:defBN}
    B_N (x) & \eqdef  \frac{\Omega_{t-1}}{N} \muest{\BS}{b,t-1}\left(\transitiondens{t}(.,x) \otimes \boldone\right) \muest{\BS}{b,t-1}(\boldone \otimes \beta^N _{t}(x,.)) \eqsp, \\
    \widetilde{B}_N(x) & \eqdef  \muest{\BS}{b,t-1}\left(\transitiondens{t}(.,x) \otimes \boldone\right) \muest{\BS}{b,t-1}\left( \boldone \otimes \boldone\right) \eqsp, \\
    \widetilde{B}(x) & \eqdef \mumeasure{b,t-1}\left(\transitiondens{t}(.,x) \otimes \boldone\right) \mumeasure{b,t-1}(\boldone \otimes \boldone) \eqsp.
\end{aligned}
\end{equation}
We apply Theorem~\ref{thm:GDCT} with $f_N = \pE \big[ B_N \big]$, $g_N = \pE \big[ \widetilde{B}_N \big]$, $g = \widetilde{B}$ and $f = 0$. To establish  i), 
note that $\pE \big[ B_N(x) \big] \leq \boundg \pE \big[ \widetilde{B}_N(x) \big]$ for all $N \in \N^{*}$ and $x \in \Xset$, since for all $(x,i) \in \Xset \times [N]$, $\beta^N _{t}(x,\particle{i}{t-1}) \leq 1$ and $N^{-1} \Omega _{t-1} \leq \boundg$. 
   Then, to prove ii),  for all $(h, f) \in \bounded{t}^2$, by the Cauchy-Schwarz inequality,
    \begin{equation*}
        \label{eq:doubleQconv}
    \begin{alignedat}{3}
        & \left| \pE \left[ \muest{\BS}{b,t-1}(h)\muest{\BS}{b,t-1}(f) - \mumeasure{b,t-1}(h) \mumeasure{b,t-1}(f) \right] \right| \\
        & \hspace{.5cm}\leq \pE \left[ \left| \left( \muest{\BS}{b,t-1}(h) - \mumeasure{b,t-1}(h)\right) \muest{\BS}{b,t-1}(f) \right| \right] + \pE \left[ \left| \left( \muest{\BS}{b,t-1}(f) - \mumeasure{b,t-1}(f)\right) \mumeasure{b,t-1}(h) \right| \right] \\
        & \hspace{.5cm}\leq \big\| \muest{\BS}{b,t-1}(h) - \mumeasure{b,t-1}(h) \big\|_2 \big\| \muest{\BS}{b,t-1}(f) \big\|_2 + \big\| \muest{\BS}{b,t-1}(f) - \mumeasure{b,t-1}(f) \big\|_2 \big| \mumeasure{b,t-1}(h)\big| \eqsp,
    \end{alignedat}
\end{equation*}
    which goes to zero by the induction hypothesis, the fact that $\sup_{N \in \N} \big\| \muest{\BS}{b,t-1}(f) \big\|_2 < \infty$ by Proposition~\ref{prop:Qbound} and $| \mumeasure{b,t-1}(h) | < \infty$.  Hence, for all $x \in \Xset$, 
    \[ \simplelim g_N(x) = g(x) \quad \text{and} \quad \simplelim \pE \left[ \muest{\BS}{b,t-1}(\boldone \otimes \boldone)^2 \right] = \mumeasure{b,t-1}(\boldone \otimes \boldone)^2 \eqsp.\] 
    Added to the fact that $\int \widetilde{B}_N (x) \nu(\rmd x) = \muest{\BS}{b,t-1}(\boldone \otimes \boldone)^2$ and $\int \widetilde{B}(x) \nu(\rmd x) = \mumeasure{b,t-1}(\boldone \otimes \boldone)^2$, we get by applying Fubini's theorem
    \begin{align*}
        \simplelim \int \pE \left[ \widetilde{B}_N(x) \right] \nu(\rmd x) = \simplelim \pE \left[ \muest{\BS}{b,t-1}(\boldone \otimes \boldone)^2 \right] = \mumeasure{b,t-1}(\boldone \otimes \boldone)^2 = \int \widetilde{B}(x) \nu(\rmd x) \eqsp.
    \end{align*}
  Then, for iii), first we have that $\pE \big[ B_N (x)^{3/2} \big] \leq \boundg^{3/2} \pE \big[ \widetilde{B}_N(x) ^{3/2} ]$ and 
    \[
        \sup_{N \in \N} \pE \big[ \widetilde{B}_N(x) ^{3/2} ] \leq \sigma_{+} ^{3/2} \sup_{N \in \N} \pE \big[ \muest{\BS}{b,t-1}(\boldone \otimes \boldone)^3 \big] \eqsp,
    \]
    where the r.h.s. is finite by choosing $m = 3$ in Proposition~\ref{prop:Qbound}. The family of non negative random variables $\{ B_N(x) \}_{N \in \N}$ is then uniformly integrable for any $x \in \Xset$. Indeed, for any $x \in \Xset$, $\alpha \in \R^{*} _{+}$ and $N \in \N$,
    \begin{align*}
        \pE \left[ B_N(x) \1_{B_N(x) \geq \alpha} \right] & \leq  \pE \big[ B_N(x)^{3/2} \big] / \sqrt{\alpha} \\
         & \leq \sigma_{+}^{3/2} \boundg^{3/2} \sup_{N \in \N} \pE \big[ \muest{\BS}{b,t-1}(\boldone \otimes \boldone)^3 \big] / \sqrt{\alpha}  \eqsp,
    \end{align*}
    hence $\underset{\alpha \to \infty}{\lim} \sup_{N \in \N}\pE \left[ B_N(x) \1_{B_N(x) \geq \alpha} \right] = 0 \eqsp.$
    On the other hand,
    \[ 
        B_N(x) =  \frac{\muest{\BS}{b,t-1}\left( \transitiondens{t}(.,x) \otimes \boldone \right) \muest{\BS}{b,t-1}\left(\boldone \otimes g_{t-1}\transitiondens{t}(.,x) \right) }{N \filter{t-1}^N\left(\transitiondens{t}(.,x)\right)} \eqsp,\] 
    and the induction hypothesis coupled with the fact that $\filter{t-1}^N(\transitiondens{t}(.,x)) \plim \filter{t-1}(\transitiondens{t}(.,x))$ with $\filter{t-1}(\transitiondens{t}(.,x))> 0$ by $\As{assp:B}{assp:positive}$ gives 
    \begin{multline*}
         \frac{\muest{\BS}{b,t-1}\left( \transitiondens{t}(.,x) \otimes \boldone \right) \muest{\BS}{b,t-1}\left(\boldone \otimes g_{t-1}\transitiondens{t}(.,x) \right)}{\filter{t-1}^N(\transitiondens{t}(.,x))} \\
    \plim \frac{\mumeasure{b,t-1}\left( \transitiondens{t}(.,x) \otimes \boldone \right) \mumeasure{b,t-1}\left(\boldone \otimes g_{t-1}\transitiondens{t}(.,x) \right)}{\filter{t-1}(\transitiondens{t}(.,x))}  \eqsp.
    \end{multline*}
         Hence, $B_N(x) \plim 0$, and by uniform integrability, for any $x \in \Xset$
     \[ 
        \simplelim f_N(x) = \simplelim \pE[B_N(x)] = 0 \eqsp.
     \]
    Finally, by Theorem~\ref{thm:GDCT} we deduce that
\begin{equation}
    \label{eq:limBN}
    \underset{N \to \infty}{\lim} \int \pE \big[ B_N(x) \big] \nu(\rmd x) = \int  \underset{N \to \infty}{\lim} \pE \big[ B_N(x) \big] \nu(\rmd x) = 0 \eqsp.
\end{equation}
The other similar  terms are treated in the same way by adapting the definitions in \eqref{eq:defBN}. As for the second integral, define for any $(x,y) \in \Xset^2$, 
\begin{align*}
    R_N(x,y) & \eqdef \frac{\Omega_{t-1} ^2}{N(N-1)}
         \muest{\BS}{b,t-1}(\transitiondens{t}(.,x) \otimes \transitiondens{t}(.,y))  \muest{\BS}{b,t-1}( \beta^N _{t}(x,.) \otimes \beta^N _{t}(y,.) ) \eqsp.
\end{align*}
Then, using that
 \begin{multline*} \int \muest{\BS}{b,t-1}\big( \transitiondens{t}(.,x) \otimes \transitiondens{t}(.,y) \big) \muest{\BS}{b,t-1}\left( \beta^N _{t}(x,.) \otimes \beta^N _{t}(y,.) \right) \nu(\rmd y) \\ \leq \muest{\BS}{b,t-1}\big( \transitiondens{t}(.,x) \otimes \boldone \big) \muest{\BS}{b,t-1}\left( \beta^N _{t}(x,.) \otimes \boldone \right) \eqsp, \end{multline*}
together with Fubini's theorem we obtain, using $N\geq 4$, 
\[ 
    0 \leq \int \pE \big[ R_N(x,y) \big] \nu^{\otimes 2}(\rmd x, \rmd y) \leq \frac{4\boundg}{3} \int \pE \big[ B_N(x) \big] \nu(\rmd x) \eqsp,
    \]
and by \eqref{eq:limBN} we get that
\[ 
\simplelim \int \pE \big[ R_N(x,y) \big] \nu^{\otimes 2}(\rmd x, \rmd y) = 0 \eqsp.
\]
The remaining term goes to zero by a similar reasoning. This completes the proof of \eqref{eq:convlim}.




For the convergence rate, by the strong mixing assumption we have that
 \begin{equation}
    \label{eq:betaNbound}
    \beta^N _t(x,y) \leq \frac{\boundg \sigma_{+}}{\sigma_{-} \Omega _{t-1}} \quad \forall (x,y) \in \Xset^2 \eqsp,
\end{equation}
and in the case $b_t = 0$, we have for example that 
\begin{align*}
    \int \muest{\BS}{b,t-1}(\transitiondens{t}(.,x) \otimes \boldone) & \muest{\BS}{b,t-1}(\beta^N _t(x,.) \otimes \boldone) \nu(\rmd x) \\
    & \leq \frac{\boundg \sigma_{+}}{\sigma_{-} \Omega _{t-1}} \int \muest{\BS}{b,t-1}(\transitiondens{t}(.,x) \otimes \boldone) \muest{\BS}{b,t-1}(\boldone \otimes \boldone) \nu(\rmd x) \\
    & \leq \frac{\boundg \sigma_{+}}{\sigma_{-} \Omega _{t-1}} \muest{\BS}{b,t-1}(\boldone \otimes \boldone)^2 \eqsp.
\end{align*}
and
\begin{align*}
    \int \muest{\BS}{b,t-1}\big(\transitiondens{t}(.,x) \otimes \transitiondens{t}(.,y)\big) & \muest{\BS}{b,t-1}\big(\beta^N _t(x,.) \otimes \beta^N _t(y,.) \big) \nu^{\otimes 2}(\rmd x, \rmd y) \\
    & \leq \frac{\boundg^2 \sigma^2 _{+}}{\sigma^2 _{-} \Omega^2 _{t-1}} \muest{\BS}{b,t-1}(\boldone \otimes \boldone)^2 \eqsp.
\end{align*}
Thus, replacing in \ref{eq:finalnorm2b0}, we get
\begin{align*}
    &\big\| \muest{\BS}{b,t}(h) - \muest{\BS}{b,t-1}(g^{\otimes 2} _{t-1} \bitransition{0}{t}[h])\big\|^2 _2 \\
    & \hspace{1cm} \leq \left[ \frac{(N-2)(N-3)}{N(N-1)} - 1 \right] \big\|\muest{\BS}{b,t-1}(g^{\otimes 2} _{t-1}\bitransition{b_{t}}{t}[h]) \big\|_2 ^2 \\
    & \hspace{3cm} + \frac{2\sigma_{+} \boundg^4 |h |^2 _\infty}{\sigma_{-}} \left[ \frac{2(N-2)}{N(N-1)} + \frac{\sigma _{+}}{\sigma _{-} N(N-1)} \right]\big\| \muest{\BS}{b,t-1}(\boldone) \big\|^2 _2 \eqsp.
\end{align*}
The case $b_t = 1$ is handled similarly using \ref*{eq:finalnorm2b1} which yields
\begin{align*}
    &\big\| \muest{\BS}{b,t}(h) - \muest{\BS}{b,t-1}(g^{\otimes 2} _{t-1} \bitransition{1}{t}[h]) \big\|^2 _2 \\
    & \hspace{2cm} \leq \left[ \frac{N-1}{N} - 1\right] \| \muest{\BS}{b,t-1}(g^{\otimes 2} _{t-1}\bitransition{1}{t}[h]) \|^2 _2 + \frac{\sigma_{+} \boundg^4 |h|_\infty ^2}{N} \big\| \muest{\BS}{b,t-1}(\boldone) \big\|^2 _2 \eqsp.
\end{align*}
Both upper bounds are $\bigo(N^{-1})$ by Proposition~\ref{prop:Qbound}. This concludes the proof.
\end{proof}
\subsection{Proof of Theorem~\ref*{thm:consistencyVBS}}
\label{proof:consistVBS}
The proof is a straightforward adaptation of the proof in \cite[Theorem 1]{leewhiteley}. Let $h \in \bounded{}$. By \eqref{prop:scdmoment}, 
\begin{align*}
    \asymptvarestim{\gamma, t}{\BS}(h) & = \sum_{s = 0}^t \left( \frac{N-1}{N} \right)^t \muest{\BS}{e_s, t}(h^{\otimes 2} _t) + N \left[ \left( \frac{N-1}{N} \right)^{t+1} - 1 \right] \muest{\BS}{\zero, t}(h^{\otimes 2} _t) \\
    & \hspace{1.5cm} +  \sum_{b \in \Bset_t \setminus \{\zero, e_{0:t} \}} N\bigg\{ \prod_{s = 0} ^{t}  \frac{1}{N^{b_s}}\bigg( \frac{N-1}{N}\bigg)^{1 - b_s} \bigg\} \muest{\BS}{b,t}(h^{\otimes 2} _t) \\
    & \plim \sum_{s = 0}^t \left\{ \mumeasure{b,t}(h^{\otimes 2} _t) - \mumeasure{\zero, t}(h^{\otimes 2} _t) \right\} = \asymptvar{\gamma,t}(h) \eqsp,
\end{align*}
where we have used that for any $b\in \Bset_t$, $\muest{\BS}{b,t}(h^{\otimes 2} _t) \plim \mumeasure{b,t}(h^{\otimes 2} _t)$ by Theorem~\ref*{thm:conv}.
\subsection{Proof of Theorem~\ref*{corr:paris}}
\label{proof:thmparis}
Define for any $t \in [N]$ and $b \in \Bset_t$,
\begin{equation}
    \label{eq:Gfilt}
    \Gfilt{b}{t} \eqdef \sigma\big( \Gfilt{b}{t-1} \cup \sigma(\{ J^i _{k, t-1}\} _{(i,k) \in [N]^2}) \cup \sigma(\{A^i _{t-1}, \particle{i}{t}\}_{i = 1}^N)\big) \eqsp,
\end{equation}
with $\Gfilt{b}{0} = \F{0}$. In the following, we write \begin{equation}
\label{eqdef:const}
    \Const{t} \eqdef \left\{\prod_{s = 0}^t N^{b_s} \left( \frac{N}{N-1}\right)^{1 - b_s}\right\} \joint{t}^N (\boldone)^2 / N^2\eqsp.
\end{equation}
The intermediary results used in the next proof are given in Section~\ref{sec:supportparis}.
\begin{proof}[Proof of Theorem~\ref*{corr:paris}]

Let $h \in \bounded{2}$. We proceed again by induction. The case $t = 0$ is a consequence of Theorem~\ref*{thm:conv} since $\parismuest{\BS}{b,0}(h) = \muest{\BS}{b,0}(h)$ for any $b \in \Bset_0$.  Let $t > 0$. Similarly to Theorem~\ref*{thm:conv} we make use of the following decomposition: 
\begin{multline*}
    \parismuest{\BS}{b,t}(h) - \mumeasure{b,t}(h) = \parismuest{\BS}{b,t}(h) - \parismuest{\BS}{b,t-1}( g^{\otimes 2} _{t-1} \bitransition{b_t}{t}[h]) \\ + \parismuest{\BS}{b,t-1}( g^{\otimes 2} _{t-1} \bitransition{b_t}{t}[h]) - \mumeasure{b,t-1}(g^{\otimes 2} _{t-1} \bitransition{b_t}{t}[h]) \eqsp.
\end{multline*}
By Minkowski's inequality and the induction hypothesis, it remains to prove that 
\begin{equation}
    \label{eq:conv_L2paris}
    \simplelim \big\| \parismuest{\BS}{b,t}(h) - \parismuest{\BS}{b,t-1}( g^{\otimes 2} _{t-1} \bitransition{b_t}{t}[h]) \big\|_2 = 0\eqsp.
\end{equation}
By Lemma~\ref{lem:condexpectparis}, $\pE \big[ \parismuest{\BS}{b,t}(h) \big| \Gfilt{b}{t-1} \big] = \parismuest{\BS}{b,t}( g^{\otimes 2} _{t-1} \bitransition{b_t}{t}[h])$ and 
\begin{equation*}
    \pE \big[ \parismuest{\BS}{b,t}(h) \parismuest{\BS}{b,t-1}(g^{\otimes 2} _{t-1} \bitransition{b_t}{t}[h]) \big] = \pE \big[ \parismuest{\BS}{b,t-1}(g^{\otimes 2} _{t-1} \bitransition{b_t}{t}[h])^2 \big] \eqsp,
\end{equation*}
hence,
\begin{align*}
     \big\| \parismuest{\BS}{b,t}(h) - \parismuest{\BS}{b,t}( g^{\otimes 2} _{t-1} \bitransition{b_t}{t}[h]) \big\|^2 _2 = \big\| \parismuest{\BS}{b,t}(h) \big\|_2 ^2 - \big\| \parismuest{\BS}{b,t-1}(g^{\otimes 2} _{t-1} \bitransition{b_t}{t}[h]) \big\|_2 ^2 \eqsp.
\end{align*}
By Proposition~\ref{prop:paris_conv_partitions}, if $b_t = 0$,
\begin{align*}
    \big\| \parismuest{\BS}{b,t}(h) \big\|_2 ^2 &= \sum_{p = 2}^4 \pE \bigg[ \Const{t}^2  \sum_{k^{1:4} _t \in \mathcal{I}^2 _0 \cap \sett^p _2} \Pbacksum^b _t(k^1 _t, k^2 _t) \Pbacksum^b _t(k^3 _t, k^4 _t) \bigg] \\
    & \leq \frac{(N-2)(N-3)}{N(N-1)} \big\| \parismuest{\BS}{b,t-1}(g^{\otimes 2 }_{t-1} \bitransition{0}{t}[h]) \big\|_2 ^2 \\
    & \hspace{1cm} + |h|_\infty ^2 \sum_{p = 2}^3 \pE \bigg[ \Const{t}^2  \sum_{k^{1:4} _t \in \mathcal{I}^2 _0 \cap \sett^p _2} \Pbacksum^b _t(k^1 _t, k^2 _t) \Pbacksum^b _t(k^3 _t, k^4 _t) \bigg] \eqsp,
\end{align*}
and 
\begin{align*}
       &  \big\| \parismuest{\BS}{b,t}(h) - \parismuest{\BS}{b,t}( g^{\otimes 2} _{t-1} \bitransition{b_t}{t}[h]) \big\|^2 _2 \\
       & \hspace{1cm} \leq \left( \frac{(N-2)(N-3)}{N(N-1)} - 1\right)  \big\| \parismuest{\BS}{b,t-1}(g^{\otimes 2} _{t-1} \bitransition{b_t}{t}[h]) \big\|_2 ^2 \\
        & \hspace{3cm} + |h|_\infty ^2 \sum_{p = 2}^3 \pE \bigg[ \Const{t}^2  \sum_{k^{1:4} _t \in \mathcal{I}^2 _0 \cap \sett^p _2} \Pbacksum^b _t(k^1 _t, k^2 _t) \Pbacksum^b _t(k^3 _t, k^4 _t) \bigg] \eqsp.
\end{align*}
By Proposition~\ref{prop:parisQbound}, $\As{assp:B}{}$ and the fact that $h$ is bounded, $\sup_{N \in \N} \big\| \parismuest{\BS}{b,t-1}(g^{\otimes 2} _{t-1} \bitransition{b_t}{t}[h]) \big\|_2 ^2 < \infty$ and
\[ 
    \simplelim \left( \frac{(N-2)(N-3)}{N(N-1)} - 1\right)  \big\| \parismuest{\BS}{b,t-1}(g^{\otimes 2} _{t-1} \bitransition{b_t}{t}[h]) \big\|_2 ^2  = 0 \eqsp,
\]
and by \ref*{item:paris_convb0} in Proposition~\ref{prop:conv_partitions}, the second term in the r.h.s. also goes to zero, which shows \eqref{eq:conv_L2paris} when $b_t = 0$. If $b_t = 1$, 
    \begin{align*}
        \big\| \muest{\BS}{b,t}(h) \big\|_2 ^2 &= \sum_{p = 1}^2 \pE \bigg[ \Const{t}^2  \sum_{k^{1:4} _t \in \mathcal{I}^2 _1 \cap \sett^p _2} \Pbacksum^b _t(k^1 _t, k^2 _t) \Pbacksum^b _t(k^3 _t, k^4 _t) \bigg] \\
        & \leq \frac{N-1}{N} \big\| \parismuest{\BS}{b,t-1}(g^{\otimes 2 }_{t-1} \bitransition{1}{t}[h]) \big\|_2 ^2 + |h|_\infty ^2 \pE \bigg[ \Const{t}^2  \sum_{k^{1:4} _t \in \mathcal{I}^2 _1 \cap \sett^1 _2} \Pbacksum^b _t(k^1 _t, k^2 _t) \Pbacksum^b _t(k^3 _t, k^4 _t) \bigg] \eqsp,
    \end{align*}
and $\big\| \parismuest{\BS}{b,t}(h) - \parismuest{\BS}{b,t}( g^{\otimes 2} _{t-1} \bitransition{b_t}{t}[h]) \big\|^2 _2$ goes to zero similarly to the case $b_t = 0$ and by application of Proposition~\ref{prop:conv_partitions}. 

The convergence rate follows straightforwardly by Proposition~\ref{prop:conv_partitions} since for $p \in \{2,3\}$
\begin{equation*}
\pE \bigg[ \Const{t}^2 \sum_{k^{1:4} _t \in \mathcal{I}^2 _0 \cap \sett^p _2} \Pbacksum^b _{t}(k^1 _t, k^2 _t) \Pbacksum^b _t(k^3 _t, k^4 _t) \bigg] = \bigo(N^{-1}) \eqsp,
\end{equation*}
and
\begin{equation*}
\pE \bigg[ \Const{t}^2 \sum_{k^{1:4} _t \in \mathcal{I}^2 _1 \cap \sett^1 _2} \Pbacksum^b _{t}(k^1 _t, k^2 _t) \Pbacksum^b _t(k^3 _t, k^4 _t) \bigg] = \bigo(N^{-1}) \eqsp.
\end{equation*}
\end{proof}

\subsection{Proof of Theorem~\ref*{thm:parisvar}}
\label{proof:scdmoment_paris}
The proof boils down to showing a $\paris$ version  of the identity \eqref{prop:scdmoment}. Let us first prove that for all $t \in \N$ and $(k^1 _t, k^2 _t) \in [N]^2$,
\begin{equation}
    \label{eq:sumb_paris}
    \sum_{b \in \Bset_t} \Pbacksum^b _t(k^1 _t, k^2 _t) = 1 \eqsp.
\end{equation}
We proceed by induction. If $t = 0$, 
\begin{equation*}
    \sum_{b \in \Bset_0} \Pbacksum^b _t(k^1 _0, k^2 _0) = \1_{k^1 _0 \neq k^2 _0} + \1_{k^1 _0 = k^2 _0} = 1 \eqsp.
\end{equation*}
Let $t > 0$ and assume that \eqref{eq:sumb_paris} holds at $t-1$ for all $(k^1 _{t-1}, k^2 _{t-1}) \in [N]^2$. By the induction hypothesis, 
\begin{align*}
     \sum_{b \in \Bset_t} & \Pbacksum^b _t(k^1 _t, k^2 _t) \\
    & = \sum_{b \in \Bset_{t-1}} M^{-1}\bigg\{ \1_{k^1 _t \neq k^2 _t} \sum_{i = 1}^M \Pbacksum^b _{t-1}(J^i _{k^1 _t, t-1}, J^i _{k^2 _t, t-1}) + \1_{k^1 _t = k^2 _t} \sum_{i = 1}^M \sum_{n = 1}^N \normweight{n}{t-1} \Pbacksum^b _{t-1}(J^i _{k^1 _t, t-1}, n)  \bigg\} \\
    & = M^{-1} \sum_{i = 1}^M \big\{ \1_{k^1 _t \neq k^2 _t} \sum_{b \in \Bset_{t-1}} \Pbacksum^b _{t-1}(J^i _{k^1 _t, t-1}, J^i _{k^2 _t, t-1}) + \1_{k^1 _t = k^2 _t} \sum_{n = 1}^N \normweight{n}{t-1} \sum_{b \in \Bset_{t-1}}  \Pbacksum^b _{t-1}(J^i _{k^1 _t, t-1}, n)  \bigg\} \\
    & = \1_{k^1 _t \neq k^2 _t} M^{-1} \sum_{i = 1}^M 1 + \1_{k^1 _t = k^2 _t} M^{-1} \sum_{n = 1}^N \sum_{i = 1}^M \normweight{n}{t-1} \\
    & = \1_{k^1 _t \neq k^2 _t} + \1_{k^1 _t = k^2 _t} = 1 \eqsp.
\end{align*}
which proves \eqref{eq:sumb_paris} at time $t$. Consequently, we have that for all $h \in \bounded{}$
\begin{align*}
    \sum_{b \in \Bset_t} \prod_{s = 0}^t N^{- b_s} \left( \frac{N-1}{N}\right)^{1 - b_s} \parismuest{\BS}{b,t}(h^{\otimes 2}) & = \sum_{b \in \Bset_t} \frac{\joint{t}^N (\boldone)^2}{N^2} \sum_{k^{1:2} _t \in [N]^2} \Pbacksum^b _t(k^1 _t, k^2 _t) h(\particle{k^1 _t}{t}) h(\particle{k^2 _t}{t}) \\
    & = \frac{\joint{t}^N(\boldone)^2}{N^2} \sum_{k^{1:2} _t \in [N]^2} h(\particle{k^1 _t}{t}) h(\particle{k^2 _t}{t}) \sum_{b \in \Bset_t} \Pbacksum^b _t(k^1 _t, k^2 _t) \\
    & = \joint{t}^N(\boldone)^2 \pred{t}^N(h)^2 = \joint{t}^N(h)^2 \eqsp.
\end{align*}
The convergence in probability is then obtained by mimicking the proof of Theorem~\ref*{thm:consistencyVBS} and using Theorem~\ref*{corr:paris}.

\subsection{Proof of Theorem~\ref*{thm:convFFBS}}
\label{proof:FFBS}
The proof of Theorem~\ref*{thm:convFFBS} requires the convergence of $\bwpath{s}^N[h_{0:s}](x)$ to $\bwpath{s}[h_{0:s}](x)$ $\pP$-a.s. for any $x \in \Xset$. 
\begin{proposition}
    \label{prop:Tasconv}
    For any $s \in \N^{*}$, any $x \in \Xset$ and additive functional $h_{0:s}$ \eqref{eq:additive},
    \begin{equation}
        \label{eq:TNas}
        \bwpath{s}^N [h_{0:s}](x) \aslim \bwpath{s}[h_{0:s}](x) \eqsp.
    \end{equation}
\end{proposition}
\begin{proof}
Let $x \in \Xset$. Define \begin{align*}
    \begin{cases}
    a_N \eqdef \pred{s-1}^N \left(g_{s-1} \big\{ \bwpath{s-1}^N[\tilde{h}_{0:s-1}]f^x _{s-1} +  \tilde{f}^x _{s-1}\big\}\right)\eqsp, \\
    b_N \eqdef \pred{s-1}^N \left( g_{s-1} m_{s}(., x) \right)\eqsp, \\
    b \eqdef \pred{s-1}\big(g_{s-1} m_{s}(.,x) \big) \eqsp.
    \end{cases}
\end{align*}
where $f^x _{s-1} : y \mapsto \transitiondens{s}(y,x)$ and 
\[ 
\tilde{f}^x _{s-1} : y \mapsto \transitiondens{s}(y, x) \left\{ \tilde{h}_{s-1}(y,x) - \bwpath{s}[h_{0:s}](x) \right\} \eqsp.
\]
Then, we have that $ a_N / b_N = \bwpath{s}^N [h_{0:s}](x) - \bwpath{s} [h_{0:s}](x)$. By $\As{assp:boundup}{}$, 
 $(f^x _{s-1}, \tilde{f}^x _{s-1}) \in \bounded{}^2$ for any $x \in \Xset$  and 
\begin{align*}
    \pred{s-1}\left(g_{s-1}\{ \bwpath{s-1}[\tilde{h}_{0:s-1}] f^x _{s-1} + \tilde{f}^x_{s-1} \} \right) & = 0 \eqsp.
\end{align*}
Hence, choosing $f_{s-1} = f^x _{s-1}$ and $\tilde{f}_{s-1} = \tilde{f}^x _{s-1}$ in  Theorem~\ref{thm:parishoeff} \eqref{eq:hoeffding1}, there exists $(d, \tilde{d}) \in (\R^{*} _{+})^2$ such that
\begin{equation}
    \pP \left( \left| a_N \right| \geq \epsilon \right) \leq \tilde{d} \exp(- dN\epsilon^2) \eqsp.
\end{equation}
On the other hand, by choosing $f_{s-1} = f^x _{s-1}$ and $\tilde{f}_{s-1} = 0$, there exists $(d', \tilde{d}') \in (\R^{*} _{+})^2$ such that 
\begin{align*}
    \pP \left( \left| b_N - b \right| \geq \epsilon\right)  \leq \tilde{d}' \exp(- d' N\epsilon^2) \eqsp.
\end{align*}
Finally, since $\left| a_N / b_N \right| \leq \left| \bwpath{s}^N [h_{0:s}](x) \right| + \left| \bwpath{s} [h_{0:s}](x) \right| \leq 2 | h_{0:s} | _\infty$ $\pP$-a.s. and $b > 0$ by $\As{assp:B}{assp:positive}$, there exist $(c_s, \tilde{c}_s) \in (\R^{*} _{+})^2 $ by Lemma~\ref{lem:genhoeff} such that 
\[
\pP \left( \left| a_N / b_N \right| \geq \epsilon\right) = \pP \left( \left| \bwpath{s}^N [h_{0:s}](x) - \bwpath{s} [h_{0:s}](x) \right|\geq \epsilon \right) \leq \tilde{c}_s \exp(- c_s N \epsilon^2) \eqsp,
\]
from which \eqref{eq:TNas} follows by applying the Borel-Cantelli Lemma. 
\end{proof}
\begin{proposition}
    \label{prop:filterTcomposition}
    For any $s \in \N^{*}$ and additive functional $h_{0:s}$ \eqref{eq:additive}
    \begin{equation}
        \simplelim \pE \left[ \filter{s-1}^N \transition{s}\left[ \big( \bwpath{s}^N [h_{0:s}] - \bwpath{s}[h_{0:s}] \big)^4\right]\right] = 0\eqsp.
    \end{equation}
\end{proposition}
\begin{proof}
    The proof is a straightforward application of \cite[Lemma 17]{paris} (which dates back to \cite{doucmoulines}). We recall it with its proof for the sake of completeness. Define for any $x \in \Xset$
    \begin{align*}
        \begin{cases}
            A_N(x) & \eqdef \big| \bwpath{s}^N [h_{0:s}](x) - \bwpath{s}[h_{0:s}](x) \big|^4 \filter{s-1}^N\big(\transitiondens{s}(.,x)\big) \eqsp, \\
            \widetilde{A} _N(x) & \eqdef \filter{s-1}^N\big(\transitiondens{s}(.,x)\big)  \eqsp, \\
            \widetilde{A} (x) & \eqdef \filter{s-1}\big(\transitiondens{s}(.,x)\big)  \eqsp.
        \end{cases}
    \end{align*}
    We apply Theorem~\ref{thm:GDCT} with $f_N = \pE \big[A _N \big]$, $g_N = \pE \big[ \widetilde{A} _N \big]$, $f = 0$ and  $g = \widetilde{A}$.
    \begin{enumerate}[label=(\roman*)]
    \item For any $x_s \in \Xset$, $\big| \bwpath{s}^N [h_{0:s}] \big|(x_s) \leq \int | h_{0:s} |(x_{0:s}) \bwpath{s}^N (x_s, \rmd x_{0:s-1}) \leq |h_{0:s}|_\infty$, hence \[ \pE \big[ A_N(x) \big] \leq 16|h_{0:s} |^4 _\infty \pE \big[ \widetilde{A}_N(x) \big]. \]
    
    \item We have that $\widetilde{A}_N (x) \aslim \widetilde{A}(x)$ for any $x \in \Xset$ by $\As{assp:boundup}{}$, \eqref{eq:asconv} and  by the dominated convergence theorem $\simplelim \pE \big[ \widetilde{A}_N(x) \big] = \widetilde{A}(x)$. On the other hand, $\int \pE \big[ \widetilde{A}_N(x) \big] \nu(\rmd x) = \pE \big[ \filter{s}^N(\boldone) \big]$, $\int \widetilde{A}(x) \nu(\rmd x) = \filter{s}(\boldone)$ and $\simplelim \pE \big[ \filter{s}^N(\boldone) \big] = \filter{s}(\boldone)$ again by the dominated convergence theorem. Hence
    \begin{equation}
        \simplelim \int \pE \big[ \widetilde{A}_N(x) \big] \nu(\rmd x) = \simplelim \pE \big[ \filter{s}^N(\boldone) \big] = \filter{s}(\boldone) = \int \widetilde{A}(x) \nu(\rmd x) \eqsp.
    \end{equation}
    \item  By Proposition~\ref{prop:Tasconv}, $\As{assp:boundup}{}$ and \eqref{eq:asconv}
    \[
        A_N(x) \aslim 0 \eqsp,
    \] 
    and since $A_N(x) \leq 16 |h_{0:s} |^4 _\infty \sigma_{+}$ $\pP$-a.s.,  by the dominated convergence theorem we get $\simplelim \pE \big[ A _N(x) ] = 0 \eqsp.$
    \end{enumerate}
    Finally, by Theorem~\ref{thm:GDCT} 
    \begin{align}
        \simplelim \pE \bigg[ \int A_N(x) \nu(\rmd x) \bigg] & = \simplelim \pE \big[ \filter{s-1}^N \transition{s}\big[ \big( \bwpath{s}^N [h_{0:s}] - \bwpath{s}[h_{0:s}] \big)^4\big]\big] \\
        & = \int \simplelim \pE \big[ A_N(x) \big] \nu(\rmd x) = 0 \eqsp. \nonumber
    \end{align}
\end{proof}
\begin{proof}[Proof of Theorem \ref*{thm:convFFBS}]
We write \begin{align*}
    \bigH^N _{s,t} & \eqdef \bwpath{s}^N [h_{0:s}]c_t + \tilde{h}_{s:t}, \quad \bigH _{s,t} \eqdef \bwpath{s} [h_{0:s}]d_t + \tilde{h}_{s:t} \\
    \bigF^N _{s,t} & \eqdef \bwpath{s}^N [f_{0:s}]c_t + \tilde{f}_{s:t}, \quad \bigF _{s,t} \eqdef \bwpath{s} [f_{0:s}]d_t + \tilde{f}_{s:t}
    \end{align*}
We proceed by induction on $t \geq s$ with $s$ fixed.  By Theorem~\ref*{thm:conv},
\begin{equation*}
    \simplelim \left\| \muest{\BS}{e_s, s}\big(\bigH_{s,s} \otimes \bigF_{s,s}\big) - \mumeasure{e_s, s}\big(\bigH_{s,s} \otimes \bigF_{s,s}\big) \right\|_2 ^2 = 0\eqsp. 
\end{equation*}
Hence, by the triangle inequality it suffices to show that the difference with the "idealized"
estimator goes to 0, i.e. 
\begin{equation}
    \label{eq:inducFFBS}
    \simplelim \left\| \muest{\BS}{e_s, s}\big(\bigH^N _{s,s} \otimes \bigF^N _{s,s}\big) - \muest{\BS}{e_s, s}\big(\bigH_{s,s} \otimes \bigF_{s,s}\big) \right\|_2 ^2 = 0\eqsp.
\end{equation}
For any $(h_{0:s}, f_{0:s})$ \eqref{eq:additive} and $(h_s, f_s) \in \bounded{}^2$, by \eqref{eq:Qexpr}
\begin{align*}
    & \left\| \muest{\BS}{e_s, s}\big(\bigH^N _{s,s} \otimes \bigF^N _{s,s}\big) - \muest{\BS}{e_s, s}\big(\bigH_{s,s} \otimes \bigF_{s,s}\big) \right\|_2  \\
     & \leq N^{-2} \bigg\| \sum_{i,j \in [N]^2} \frac{N^{s+1} \joint{s}^N(\boldone)^2}{(N-1)^s} \backsum^{e_s} _s(i,j) \bigg[ \bigH^N _{s,s}(\particle{i}{s}) \bigF^N _{s,s}(\particle{j}{s}) - \bigH_{s,s}(\particle{i}{s}) \bigF _{s,s}(\particle{j}{s}) \bigg] \bigg\|_2 \\
     & \leq N^{-1} \sum_{i,j \in [N]^2} \left\|  \frac{N^s \joint{s}^N(\boldone)^2}{(N-1)^s} \backsum^{e_s} _s(i,j) \right\|_4 \bigg\| \bigH^N _{s,s}(\particle{i}{s}) \bigF^N _{s,s}(\particle{j}{s}) - \bigH_{s,s}(\particle{i}{s}) \bigF _{s,s}(\particle{j}{s}) \bigg\|_4
\end{align*}
by Cauchy-Schwarz inequality. We now show $\sup_{N \in \N}\left\|  N^s (N-1)^{-s} \joint{s}^N(\boldone)^2\backsum^{e_s} _s(i,j) \right\|_4$ is bounded for all $(i,j) \in [N]^2$. We first show by induction that for any $n \in \N$, $\backsum^\zero _n(i,j) \leq \1_{i \neq j}$. For all $(i,j) \in [N]^2$ , $\backsum^\zero _0 (i,j) = \1_{i \neq j}$, and for any $n > 0$, by \eqref{lem:bigtauexpr}
\begin{align*}
    \backsum^\zero _n(i,j) & = \1_{i \neq j} \sum_{k,\ell \in [N]^2} \beta^\BS _n(i,k) \beta^\BS _n(j,\ell) \backsum^\zero _{n-1}(k,\ell) \\
    &  \leq \1_{i \neq j} \sum_{k, \ell \in [N]^2} \beta^\BS _n(i,k) \beta^\BS _n(j,\ell) \1_{k \neq \ell} \\
    & \leq \1_{i \neq j} \sum_{k, \ell \in [N]^2} \beta^\BS _n(i,k) \beta^\BS _n(j, \ell) \leq \1_{i \neq j} \eqsp,
\end{align*}
where we have used the induction hypothesis in the second line. This shows the result. Next, we have that 
\begin{align*}
    \backsum^{e_s} _s(i,j) & = \1_{i = j} \sum_{k, \ell \in [N]^2} \beta^\BS _s(i,k) \normweight{\ell}{s-1} \backsum^\zero _{s-1}(k, \ell) \\
    & \leq \1_{i = j} \sum_{k, \ell \in [N]^2} \beta^\BS _s(i,k) \normweight{\ell}{s-1} \1_{k \neq \ell} \\
    & \leq \1_{i = j} \sum_{k, \ell \in [N]^2} \beta^\BS _s(i,k) \normweight{\ell}{s-1}  = \1_{i = j} \eqsp.
\end{align*}
Hence, $\sup_{N \in \N}\left\|  N^s \joint{s}^N(\boldone)^2/(N-1)^s \backsum^{e_s} _s(i,j) \right\|_4 \leq (2\boundg^2)^s \1_{i = j}$. Consequently, 
\begin{multline}
    \label{eq:boundQes}
 \left\| \muest{\BS}{e_s, s}\big(\bigH^N _{s,s} \otimes \bigF^N _{s,s}\big) - \muest{\BS}{e_s, t}\big(\bigH_{s,s} \otimes \bigF_{s,s}\big) \right\|_2   \\
     \leq \frac{(2 \boundg^2)^s}{N}\sum_{i = 1}^N \left\| \bigH^N _{s,s}(\particle{i}{s}) \bigF^N _{s,s}(\particle{i}{s}) - \bigH_{s,s}(\particle{i}{s}) \bigF _{s,s}(\particle{i}{s}) \right\|_4 \eqsp, 
\end{multline}
and
\begin{align*}
    & \left\| \bigH^N _{s,s}(\particle{i}{s}) \bigF^N _{s,s}(\particle{j}{s}) - \bigH_{s,s}(\particle{i}{s}) \bigF _{s,s}(\particle{j}{s}) \right\|_4 \\
    & \hspace{1cm} \leq \big\| \big( \bigH^N _{s,s}(\particle{i}{s}) - \bigH _{s,s}(\particle{i}{s}) \big) \bigF^N _{s,s}(\particle{i}{s}) \big\|_4 + \big\| \big( \bigF^N _{s,s}(\particle{i}{s}) - \bigF _{s,s}(\particle{i}{s}) \big) \bigH^N _{s,s}(\particle{i}{s}) \big\|_4 \\
    & \hspace{1cm} \leq C_{f,c} \big\|  \bwpath{s}^N[h_{0:s}](\particle{i}{s}) - \bwpath{s}[h_{0:s}](\particle{i}{s}) \big\|_4 + C_{h,d} \big\| \bwpath{s}^N[f_{0:s}](\particle{i}{s}) - \bwpath{s}[f_{0:s}](\particle{i}{s})  \big\|_4 \eqsp.
\end{align*}
where $C_{f,d} \eqdef |d_s|^4 _\infty \big( | f_{0:s} |_\infty + |f_s| _\infty \big)^4 $ and $C_{h,c} \eqdef |c_s|^4 _\infty \big( | h_{0:s} |_\infty + |h_s| _\infty \big)^4 $ 
which are finite because $\tilde{h}_{s}, \tilde{f}_{s}) \in \bounded{2}^4$, $(c_s, d_s, h_s, f_s) \in \bounded{}^4$. We have used that
\begin{equation*}
\big| \bigH^N _{s,s}(\particle{i}{s}) \big| \leq \int |c_s|_\infty | h_{0:s} |_\infty \bwpath{s}^N(\particle{i}{s}, \rmd x_{0:s-1}) + | h_s |_\infty = |c_s |_\infty |h_{0:s}|_\infty + | h_s |_\infty = C_{h,c} 
\eqsp.
\end{equation*}
For any $h_{0:s} \in \bounded{s+1}$,
\begin{align*}
    & \big\| \bwpath{s}^N [h_{0:s}](\particle{i}{s}) - \bwpath{s}[h_{0:s}](\particle{i}{s}) \big\|_4 ^4 \\
    & \hspace{2cm} =  \pE \bigg[ \pE \big[ \bwpath{s}^N [h_{0:s}](\particle{i}{s}) - \bwpath{s}[h_{0:s}](\particle{i}{s}) \big)^4 \big| \F{t-1} \big] \bigg] \\
    & \hspace{2cm}  = \pE \bigg[ \sum_{j = 1}^N \normweight{j}{s-1} \int \big( \bwpath{s}^N [h_{0:s}](x) - \bwpath{s}[h_{0:s}](x) \big)^4 \transition{s}(\particle{j}{s-1}, \rmd x) \bigg]\\
    & \hspace{2cm}  = \pE \left[ \filter{s-1}^N \transition{s}\left[ \big( \bwpath{s}^N [h_{0:s}] - \bwpath{s}[h_{0:s}] \big)^4\right] \right]\eqsp,
\end{align*}
and replacing in \eqref{eq:boundQes} we get
\begin{align*}
& \left\| \muest{\BS}{e_s, t}\big(\bigH^N _{s,s} \otimes \bigF^N _{s,t}\big) - \muest{\BS}{e_s, s}\big(\bigH_{s,t} \otimes \bigF_{s,t}\big) \right\|_2 ^2  \\
     & \hspace{2cm} \leq (2\boundg^2)^s \bigg\{ C_{f,d} \pE \left[ \filter{s-1}^N \transition{s}\left[ \big( \bwpath{s}^N [h_{0:s}] - \bwpath{s}[h_{0:s}] \big)^4\right] \right] \\
     &\hspace{5cm}  + C_{h,c} \pE \left[ \filter{s-1}^N \transition{s}\left[ \big( \bwpath{s}^N [f_{0:s}] - \bwpath{s}[f_{0:s}] \big)^4\right] \right] \bigg\}
\end{align*}
The upperbound goes to zero by Proposition~\ref{prop:filterTcomposition} and this finishes the proof of the initialization.

Let $t > s$ and $h_{0:s} \in \additive{s+1}$ an additive functional. Assume that \eqref{eq:inducFFBS} holds at $t-1$. By the induction hypothesis 
\begin{align*}
    \simplelim \big\| \muest{\BS}{e_s, t-1}\big( \Q{t} \bigH^N _{s,t} \otimes 
    \Q{t} \bigF^N _{s,t} \big) - \mumeasure{e_s, t-1}\big( \Q{t} \bigH _{s,t} \otimes 
    \Q{t} \bigF _{s,t} \big)\big\|_2 = 0 \eqsp,
\end{align*}
where $\Q{t}$ is defined in \eqref{def:Qdef} and for example
\[ 
  \Q{t} \bigH^N _{s,t} (x_{s:t-1}) = \bwpath{s}^N[h_{0:s}](x_s) \Q{t}[c_t](x_{t-1}) + \Q{t}[\tilde{h}_{s:t}](x_{s:t-1}) \eqsp,
\]
where $\Q{t}[c_t]$ and $\Q{t}[\tilde{h}_{s:t}]$ are bounded by $\As{assp:B}{}$,
and by defintion of $\mumeasure{e_s, t}$ \eqref{def:Qdef}
\begin{equation*}
    \mumeasure{e_s, t-1}\big( \Q{t} \bigH _{s,t} \otimes 
    \Q{t} \bigF _{s,t} \big) = \mumeasure{e_s, t}\big( \bigH_{s,t} \otimes \bigF _{s,t}\big) \eqsp.
\end{equation*} 
Hence, to prove \eqref{eq:hypthmFFBS} it is enough to show
\begin{equation}
\begin{alignedat}{2}
    \label{eq:FFBSdelta}
     \simplelim \big\| \muest{\BS}{e_s, t}\big( \bigH^N _{s,t}\otimes 
    \bigF^N _{s,t} \big) - \muest{\BS}{e_s, t-1}\big( \Q{t} \bigH^N _{s,t} \otimes 
    \Q{t} \bigF^N _{s,t} \big)\big\|_2 = 0 \eqsp.
\end{alignedat}
\end{equation} 
Because $\bwpath{s}^N [h_{0:s}]$ and $\bwpath{s}^N [f_{0:s}]$ are $\F{t-1}$-measurable, by Proposition~\ref*{prop:mu_expression} 
\[ 
\pE \big[ \muest{\BS}{e_s, t}\big( \bigH^N _{s,t} \otimes \bigF^N _{s,t} \big) \big| \F{t-1} \big] =  \muest{\BS}{e_s, t-1}\big( \Q{t} \bigH^N _{s,t} \otimes \Q{t} \bigF^N _{s,t} ) \eqsp,
\]
and thus
\begin{multline}
    \big\| \muest{\BS}{e_s, t}\big( \bigH^N _{s,t}\otimes 
    \bigF^N _{s,t} \big) - \muest{\BS}{e_s, t-1}\big( \Q{t} \bigH^N _{s,t} \otimes 
    \Q{t} \bigF^N _{s,t} \big)\big\|_2 \\ = \big\| \muest{\BS}{e_s, t}\big( \bigH^N _{s,t}\otimes 
    \bigF^N _{s,t} \big) \big\|_2 - \big\| \muest{\BS}{e_s, t-1}\big( \Q{t} \bigH^N _{s,t} \otimes 
    \Q{t} \bigF^N _{s,t} \big)\big\|_2 \eqsp.
\end{multline}
Now note that Proposition~\ref{prop:norm2QBS} is still applicable with $h = \bigH^N _{s,t} \otimes \bigF^N _{s,t}$ although there is a slight abuse because this specific $h$ depends on the particles up to $s-1$ through $\bwpath{s}^N[h_{0:s}]$ and $\bwpath{s}^N[f_{0:s}]$. However, as they are $\F{t-1}$-measurable, Proposition~\ref{cor:condexpect} is still valid and hence Proposition~\ref{prop:norm2QBS}. Additionally, this specific $h$ is bounded almost surely since for any $(x_{s:t}, x' _{s:t}) \in \big( \Xset^{t-s+1} \big)^2$
\[ 
\big| \bigH^N _{s,t}(x_{s:t}) \bigF^N _{s,t}(x' _{s:t}) \big| \leq C \eqdef \big( |h_{0:s}|_\infty |c_t| _\infty + | \tilde{h}_{s:t} |_\infty \big)\big( |f_{0:s}|_\infty |d_t| _\infty + | \tilde{f}_{s:t} |_\infty \big)
\]
and hence
\begin{align*}
    &\big\| \muest{\BS}{e_s, t}\big( \bigH^N _{s,t}\otimes 
    \bigF^N _{s,t} \big) - \muest{\BS}{e_s, t-1}\big( \Q{t} \bigH^N _{s,t} \otimes 
    \Q{t} \bigF^N _{s,t} \big)\big\|_2\\
    & \hspace{.3cm} \leq \left[ \frac{(N-2)(N-3)}{N(N-1)} - 1 \right] \big\|\muest{\BS}{b,t-1}\big( \Q{t} \bigH^N _{s,t} \otimes 
    \Q{t} \bigF^N _{s,t} \big) \big\|_2 ^2 \\
    & \hspace{1cm} + \frac{N-2}{N-1} \boundg^3 C ^2 \int \nu(\rmd x)\pE \bigg[ \frac{\Omega_{t-1}}{N} \bigg\{ \muest{\BS}{b,t-1}(\transitiondens{t}(.,x) \otimes \boldone) \muest{\BS}{b,t-1}(\beta^N _{t}(x,.) \otimes \boldone) \\
    & \hspace{1cm} + \muest{\BS}{b,t-1}(\boldone \otimes \transitiondens{t}(.,x)) \muest{\BS}{b,t-1}(\beta^N _{t}(x,.) \otimes \boldone) + \muest{\BS}{b,t-1}(\transitiondens{t}(.,x) \otimes \boldone) \\
     & \hspace{1cm}\times \muest{\BS}{b,t-1}(\boldone \otimes \beta^N _{t}(x,.)) 
     + \muest{\BS}{b,t-1}(\boldone \otimes \transitiondens{t}(.,x)) \muest{\BS}{b,t-1}(\boldone \otimes \beta^N _{t}(x,.)) \bigg] \\
     & + \int \nu^{\otimes 2}(\rmd y,\rmd x)\pE \bigg[ \frac{\boundg^2 C^2 \Omega_{t-1} ^2}{N(N-1)} \\
    & \hspace{1cm}\times \bigg\{
         \muest{\BS}{b,t-1}(\transitiondens{t}(.,x) \otimes \transitiondens{t}(.,y)) \muest{\BS}{b,t-1}\left( \beta^N _{t}(x,.) \otimes \beta^N _{t}(y,.) \right) \\
     & \hspace{2cm} + \muest{\BS}{b,t-1}(\transitiondens{t}(.,x) \otimes \transitiondens{t}(.,y))  \muest{\BS}{b,t-1}\left( \beta^N _{t}(y,.) \otimes \beta^N _{t}(x,.) \right) \bigg\}\bigg] \eqsp.
\end{align*}
The first term in the r.h.s. goes to zero by Proposition~\ref{prop:Qbound} and the fact that $\Q{t} \bigH^N _{s,t} \otimes \Q{t} \bigF^N _{s,t}$ are bounded. The remaining terms are similar to those that appear in the proof of Theorem~\ref*{thm:conv} up to some constants and thus go to zero. 
\vspace{.2cm}\\
For the second part, by Theorem~\ref{thm:parishoeff} and Borel-Cantelli Lemma, $\smooth{0:t}{t}^N(h) \aslim \smooth{0:t}{t}(h)$. Then, by multiple applications of Theorem~\ref*{thm:convFFBS} and using the bilinearity of $\muest{\BS}{b,t}$ and $\mumeasure{b,t}$, for any $s \in [0:t]$ and bounded additive functional $h_t$
\begin{align*}
    \muest{\BS}{e_s,t}\big( \big[& g_t \big\{  \bwpath{s}^N [h_{0:s}] + \tilde{h}_{s:t} - \smooth{0:t}{t}^N (h_t) \big\} \big]^{\otimes 2}\big) \\
    & = \muest{\BS}{e_s,t} \big( \big[g_t \big\{ \bwpath{s}^N [h_{0:s}] + \tilde{h}_{s:t} \big\}]^{\otimes 2} \big) - \smooth{0:t}{t}^N(h_{t}) \bigg( \muest{\BS}{e_s,t} \big( \big[g_t \big\{ \bwpath{s}^N [h_{0:s}] + \tilde{h}_{s:t} \big\}] \otimes \boldone \big) + \\
    & \hspace{1cm} +  \muest{\BS}{e_s,t} \big( \boldone \otimes \big[g_t \big\{ \bwpath{s}^N [h_{0:s}] + \tilde{h}_{s:t} \big\}] \big) \bigg) + \smooth{0:t}{t}^N(h_t)^2 \muest{\BS}{e_s,t}\big( \boldone \otimes \boldone \big) \\
    & \plim \mumeasure{e_s,t} \big( \big[g_t \big\{ \bwpath{s} [h_{0:s}] + \tilde{h}_{s:t} \big\}]^{\otimes 2} \big) - \smooth{0:t}{t}(h_{t}) \bigg( \mumeasure{e_s,t} \big( \big[g_t \big\{ \bwpath{s} [h_{0:s}] + \tilde{h}_{s:t} \big\}] \otimes \boldone \big) + \\
    & \hspace{1cm} +  \mumeasure{e_s,t} \big( \boldone \otimes \big[g_t \big\{ \bwpath{s} [h_{0:s}] + \tilde{h}_{s:t} \big\}] \big) \bigg) + \smooth{0:t}{t}(h_t)^2 \mumeasure{e_s,t}\big( \boldone \otimes \boldone \big) \\
    & = \mumeasure{e_s,t}\big( \big[ g_t \big\{  \bwpath{s} [h_{0:s}] + \tilde{h}_{s:t} - \smooth{0:t}{t} (h_t) \big\} \big]^{\otimes 2}\big) \eqsp,
\end{align*}
from which the weak consistency of $\asymptvarestim{0:t|t}{\BS}(h)$ follows.
\end{proof}
\subsection{Supporting results for Theorem~\ref*{thm:conv}}
\label{apdx:proofconv}
In this section we prove Proposition~\ref{prop:Qbound} and the upperbound of $\big\| \muest{\BS}{b,t}(h) \big\|^2 _2$ used in the proof of Theorem~\ref*{thm:conv}.
\begin{proposition}
    \label{prop:Qbound}
    Assume that $\As{assp:B}{}$ holds. For any $t \in \N \eqsp,$ $b \in \Bset_t$ and $m \in \N$,
    \begin{equation}
        \label{eq:normbound}
        \sup_{N \in \N} \pE \big\| \muest{\BS}{b,t}(\boldone) \big\| _m < \infty \eqsp.
    \end{equation}
\end{proposition}
    We preface the proof with supporting lemmata.
    \begin{lemma}
        \label{lem:cardinals}
        For any $p \geq 2$, and $N \geq 2m$
        \begin{align*}
            \card{\mathcal{I}^m _0 \cap \sett^p _m} = \bigo(N^p) \quad \textrm{and}\quad \card{\mathcal{I}^m _1 \cap \sett^p _m} = \bigo(N^p) \eqsp.
        \end{align*}
    \end{lemma}
    \begin{proof}
        The tuples in $\sett^p _m$ contain $p$ distinct elements. These $p$ distinct elements can be selected in ${N \choose p}$ ways. For each of these tuples of size $p$, there are $p^{2m}$ tuples of size $2m$ with each element taking one of the $p$ values. These tuples of size $2m$ contain at most $p$ distinct elements. Hence,
        \[ 
            \card{\mathcal{I}^m _0 \cap \sett^p _m} \leq \card{\sett^p _m} \leq {N \choose p} p^{2m} \leq \frac{N^p}{p !} p^{2m} \eqsp,
        \]
    and similarly, 
    
    $$
    \card{\mathcal{I}^m _1 \cap \sett^p _m} \leq \frac{N^p}{p !} p^{2m} \eqsp.
    $$
    \end{proof}
    \begin{proposition}
        \label{lem:prodbeta}
        Let $t \in \N^{*}$, $m \in \N^{*}$ and $(k^1 _{t-1}, \cdots, k^{2m} _{t-1}) \in [N]^{2m}$.
        \begin{enumerate}[label=(\roman*)]
            \item \label{eq:b0}If $p \in [2:2m]$ and $(k^1 _t, \cdots, k^{2m} _t) \in \mathcal{I}^m _0 \cap \sett^p _m$,
        \begin{align*}
            \pE \left[ \prod_{j = 1}^{2m} \beta^\BS _t(k^j _t, k^j _{t-1}) \bigg| \F{t-1} \right] \leq \frac{\boundg^p}{\Omega^p _{t-1}} \eqsp.
        \end{align*}
        \item \label{eq:b1} If $p \in [1:m]$ and $(k^1 _t, \cdots, k^{2m} _t) \in \mathcal{I}^m _1 \cap \sett^p _m$,
        \begin{align*}
            \pE \left[ \prod_{j = 1}^{m} \beta^\BS _t(k^{2j-1} _t, k^{2j-1} _{t-1}) \normweight{k^{2j} _{t-1}}{t-1} \bigg| \F{t-1} \right] \leq \frac{\boundg^{p+m}}{\Omega^{p+m} _{t-1}} \eqsp.  
        \end{align*}
    \end{enumerate}
    \end{proposition}
    \begin{proof}
        Let $p \in [2:2m]$. By definition there are $p$ distinct elements in each $\bm{k} \eqdef (k^1 _t, \cdots, k^{2m} _t)\in \mathcal{I}^m _0 \cap \sett^p _m$. Let $\bm{k}_p \eqdef \{ a_1, \cdots, a_p \} = \{k^1 _t, \cdots, k^{2m} _t\}$ the set of cardinal $p$ containing the $p$ distinct elements in a tuple $\bm{k}\in \sett^p _m$. Define for any $a_i \in \bm{k}_p$, $V_{a_i} \eqdef \{j \in [2m]: k^j _t = a_i\}$. Each $V_{a_i}$ is non-empty so that it is possible to pick $j_i\in V_{a_i}$, and by \eqref{lem:identity} in Lemma~\ref*{lem:BSGTidentity} and the fact that different particles are i.i.d. conditionally to $\F{t-1}$,
        \begin{align*}
            & \pE \left[ \prod_{j = 1}^{2m} \beta^\BS _t(k^j _t, k^j _{t-1}) \bigg| \F{t-1} \right] \\
            & \hspace{1.5cm} = \prod_{i = 1}^p \pE \bigg[ \prod_{j \in V_{a_i}} \beta^\BS _t(k^j _t, k^j _{t-1}) \bigg| \F{t-1} \bigg] \\
            & \hspace{1.5cm} =  \prod_{i = 1}^p \int \left\{\prod_{j \in V_{a_i}} \beta^N _t(\particle{k^j _{t}}{t}, \particle{k^j _{t-1}}{t-1}) \filter{t-1}^N \transition{t}( \rmd \particle{a_i}{t}) \right\}\eqsp, \\
            & \hspace{1.5cm} = \prod_{i = 1}^p \int \left\{ \prod_{j \in V_{a_i} \setminus \{j_i\}}  \beta^N  _t(\particle{a_i}{t}, \particle{k^j _{t-1}}{t-1}) \beta^N _t(\particle{a_i}{t}, \particle{k^{j_i} _{t-1}}{t-1}) \filter{t-1}^N \transition{t}( \rmd \particle{a_i}{t})\right\}\eqsp,  \\
            & \hspace{1.5cm} = \prod_{i = 1}^p \int \left\{ \prod_{j \in V_{a_i} \setminus \{j_i\}}  \beta^N  _t(\particle{a_i}{t}, \particle{k^j _{t-1}}{t-1}) \normweight{k^{j_i} _{t-1}}{t-1} \transition{t}(\particle{k^{j_i} _{t-1}}{t-1}, \rmd \particle{a_i}{t})\right\} \eqsp,
        \end{align*}
        with the convention $\prod_{\emptyset} = 1$.
     Then, since for any $(k, \ell) \in [N]^2$, $\beta^\BS _t(\particle{k}{t},\particle{\ell}{t-1}) \leq 1$, we get 
        \begin{align*}
            \pE \left[ \prod_{j = 1}^{2m} \beta^\BS _t(k^m _t, k^m _{t-1}) \bigg| \F{t-1} \right] \leq \prod_{i = 1}^p \frac{\weight{k^{j_i} _{t-1}}{t-1}}{\Omega_{t-1}} \int \transition{t}(\particle{k^{j_i} _{t-1}}{t-1}, \rmd \particle{a_i}{t}) \leq \frac{\boundg ^p}{\Omega _{t-1}^p} \eqsp.
        \end{align*}
     Now let $p \in [1:m]$ and $\bm{k} = (k^1 _t, \cdots, k^{2m} _t) \in \mathcal{I}^m _1 \cap \sett^p _m$. Define $\widetilde{V}_{a_i} \eqdef V_{a_i} \cap \{1, 3, \ldots, 2m-1\}$ for any $a_i \in \bm{k}_p$. These sets are non-empty since each $V_{a_i}$ is non-empty and has as many even indices as odd indices, by definition of $\mathcal{I}^m _1 \cap \sett^p _m$. It is then possible to pick $j_i\in \widetilde V_{a_i}$ and, since for any $i \in [N]$ $\normweight{i}{t-1}$ is $\F{t-1}$-measurable,
    \begin{align*}
        & \pE \left[ \prod_{j = 1}^{m} \beta^\BS _t(k^{2j-1} _t, k^{2j-1} _{t-1}) \normweight{k^{2j} _{t-1}}{t-1} \bigg| \F{t-1} \right] \\ 
        & \hspace{.5cm} = \prod_{j = 1}^m \normweight{k^{2j} _{t-1}}{t-1}  \pE \left[ \prod_{j = 1}^{m} \beta^\BS _t(k^{2j-1} _t, k^{2j-1} _{t-1}) \bigg| \F{t-1} \right]\eqsp, \\
        & \hspace{.5cm} =  \prod_{j = 1}^m \normweight{k^{2j} _{t-1}}{t-1} \prod_{i = 1}^p \int \left\{\prod_{j \in \widetilde{V}_{a_i}} \beta^\BS _t(k^j _{t}, k^j _{t-1}) \filter{t-1}^N \transition{t}( \rmd \particle{a_i}{t}) \right\}\eqsp, \\
        & \hspace{.5cm} = \prod_{j = 1}^m \normweight{k^{2j} _{t-1}}{t-1} \prod_{i = 1}^p \int \left\{ \prod_{j \in \widetilde{V}_{a_i} \setminus \{j_i\}}  \beta^N  _t(\particle{a_i}{t}, \particle{k^j _{t-1}}{t-1}) \beta^N  _t(\particle{a_i}{t}, \particle{k^{j_i} _{t-1}}{t-1}) \filter{t-1}^N \transition{t}( \rmd \particle{a_i}{t})\right\}\eqsp,  \\
        & \hspace{.5cm} = \prod_{j = 1}^m \normweight{k^{2j} _{t-1}}{t-1} \prod_{i = 1}^p \int \left\{ \prod_{j \in \widetilde{V}_{a^i} \setminus \{j_i\}}  \beta^N  _t(\particle{a_i}{t}, \particle{k^j _{t-1}}{t-1}) \normweight{k^{j_i} _{t-1}}{t-1} \transition{t}(\particle{k^{j_i} _{t-1}}{t-1}, \rmd \particle{a_i}{t})\right\} \eqsp,
    \end{align*}
    with the convention $\prod_\emptyset = 1$. Hence, 
    \begin{equation*}
        \pE \left[ \prod_{j = 1}^{m} \beta^\BS _t(k^{2j-1} _t, k^{2j-1} _{t-1}) \normweight{k^{2j} _{t-1}}{t-1} \bigg| \F{t-1} \right] \leq \frac{\boundg^{p+m}}{\Omega^{p+m} _{t-1}} \eqsp.
    \end{equation*}
    \end{proof}
    \begin{proof}[Proof of proposition~\ref*{prop:Qbound}]
        \label{proof:Qbound}
    Let $m \in \N$ and assume for now that $N \geq 2m$.
    We proceed by induction. For $t = 0$, and $b_0 = 0$ we have that $\muest{\BS}{\zero, 0}(\boldone) = N^{-1} (N-1)^{-1} \sum_{i,j \in [N]^2} \1_{i \neq j} = 1$ which completes the proof. If $b_0 = 1$, $\muest{\BS}{1,0}(\boldone) = N^{-1} \sum_{i,j \in [N]^2} \1_{i = j} = 1$ and the result follows. 
    \vspace{.3cm}\\
    \indent Let $t \in \N^{*}$ and $b \in \Bset_t$. Assume \eqref{eq:normbound} holds at time $t-1$. Again we treat the cases $b_t = 0$ and $b_t = 1$ separately.  
    In the case $b_t = 0$, by \eqref{eq:alternativeQ},
    \begin{align*}
        \pE \left[ \muest{\BS}{b,t}(\boldone)^m \right] = \pE \left[ \sum_{k^{1:2m} _{0:t} \in [N]^{2m(t+1)}} \prod_{j = 1}^m \bigprod{b,t}(k^{2j - 1} _{0:t}, k^{2j} _{0:t}) \right]
    \end{align*}
    and 
    \begin{align*}
        & \sum_{k^{1:2m} _{0:t} \in [N]^{2m(t+1)}} \prod_{j = 1}^m \bigprod{b,t}(k^{2j - 1} _{0:t}, k^{2j} _{0:t}) \\
        & \hspace{1cm} = \sum_{k^{1:2m} _{0:t-1} \in [N]^{2mt}} \prod_{j = 1}^m \bigprod{b,t-1}(k^{2j -1} _{0:t-1}, k^{2j} _{0:t-1}) \\
        & \hspace{3cm} \times \sum_{k^{1:2m} _{t} \in [N]^{2m}} \prod_{j = 1}^m \frac{\Omega^{2} _{t-1} \1_{k^{2j - 1} _t \neq k^{2j} _t}}{N (N-1)} \beta^\BS _t(k^{2j-1} _t, k^{2j -1} _{t-1}) \beta^\BS _t(k^{2j} _t, k^{2j} _{t-1}) \eqsp.
    \end{align*}
    By Proposition~\ref{lem:prodbeta}\ref{eq:b0} and Lemma~\ref{lem:cardinals}, for any $k^{1:2m} _{t-1} \in [N]^{2m}$, 
    \begin{align*}
       & \sum_{k^{1:2m} _t \in [N]^{2m}} \pE \left[ \prod_{j = 1}^{m}  \beta^\BS _t(k^{2j -1} _t, k^{2j-1} _{t-1}) \beta^\BS(k^{2j} _t, k^{2j} _{t-1}) \1_{k^{2j - 1} _t \neq k^{2j} _t} \bigg| \F{t-1} \right] \\
        & \hspace{2cm} = \sum_{p = 2}^{2m} \sum_{k^{1:2m} _t \in \mathcal{I}^m _0 \cap \sett^p _m} \pE \left[ \prod_{j = 1}^{2m}  \beta^\BS _t(k^j _t, k^j _{t-1}) \bigg| \F{t-1} \right] \\
        & \hspace{2cm} \leq  \sum_{p = 2}^{2m} \sum_{k^{1:2m} _t \in \mathcal{I}^m _0 \cap \sett^p _m} \frac{\boundg^p}{\Omega^p _{t-1}} \leq \sum_{p = 2}^{2m}  \frac{\boundg ^p \card{ \mathcal{I}^m _0 \cap \sett^p _m}}{\Omega^p _{t-1}} \leq C\sum_{p = 2}^{2m} \frac{N^p }{\Omega^p _{t-1}} \eqsp,
    \end{align*}
    where $C$ is a constant independent of $N$. 
   Consequently, using the fact that \[\sum_{k^{1:2m} _{0:t-1} \in [N]^{2mt}} \prod_{j = 1}^m \bigprod{b,t-1}(k^{2j -1} _{0:t-1}, k^{2j} _{0:t-1})= \muest{\BS}{b,t-1}(\boldone)^m\] which is $\F{t-1}$-measurable and that $\Omega _{t-1} \leq N \boundg$ $\pP$-a.s. by $\As{assp:B}{}$, we get 
    \begin{align*}
        \pE \left[ \muest{\BS}{b,t}(\boldone)^m \right] 
        & \leq C\pE \left[ \muest{\BS}{b,t-1}(\boldone)^m  \frac{N^{2m - p} N^p}{N^m (N-1)^m} \right] \leq C\pE\left[ \muest{\BS}{b,t-1}(\boldone)^m \right] \eqsp,
    \end{align*}
    which completes the proof. 
    In the case $b_t = 1$, again by \eqref{eq:alternativeQ},
    \begin{align*}
        & \sum_{k^{1:2m} _{0:t} \in [N]^{2m(t+1)}} \prod_{j = 1}^m \bigprod{b,t}(k^{2j - 1} _{0:t}, k^{2j} _{0:t}) \\
        & \hspace{1cm} = \sum_{k^{1:2m} _{0:t-1} \in [N]^{2mt}} \prod_{j = 1}^m \bigprod{b,t-1}(k^{2j -1} _{0:t-1}, k^{2j} _{0:t-1}) \\
        & \hspace{3cm} \times \sum_{k^{1:2m} _{t} \in [N]^{2m}} \prod_{j = 1}^m \frac{\Omega^{2} _{t-1} \1_{k^{2j - 1} _t = k^{2j} _t}}{N} \beta^\BS _t(k^{2j-1} _t, k^{2j -1} _{t-1}) \normweight{k^{2j} _{t-1}}{t-1} \eqsp.
    \end{align*}
    Then, similarly to the case $b_t = 0$, by Proposition~\ref{lem:prodbeta}-\ref{eq:b1} and Lemma~\ref{lem:cardinals}, for any $k^{1:2m} _{t-1} \in [N]^{2m}$,
    \begin{align*}
        \sum_{k_t^{1:2m} \in [N]^{2m}} \pE & \left[ \prod_{j = 1}^{m}  \beta^\BS _t(k^{2j - 1} _t, k^{2j -1} _{t-1}) \normweight{k^{2j}_{t-1}}{t-1}\1_{k^{2j - 1} _t = k^{2j} _t}\bigg| \F{t-1}\right] \\
        & = \sum_{p = 1}^{m} \sum_{k_t^{1:2m} \in \mathcal{I}^1 _m \cap \sett^p _m} \pE \left[ \prod_{j = 1}^{m} \beta^\BS _t(k^{2j-1} _t, k^{2j-1} _{t-1}) \normweight{k^{2j} _{t-1}}{t-1} \bigg| \F{t-1} \right]  \\
        & = \sum_{p = 1}^{m} \sum_{k_t^{1:2m} \in \mathcal{I}^1 _m \cap \sett^p _m}  \frac{\boundg^{m+p}}{\Omega^{m+p} _{t-1}} \leq C\sum_{p = 1}^m \frac{N^p}{\Omega^{m+p} _{t-1}} \eqsp,
    \end{align*}
    where $C$ is a constant independent of $N$, and
    \begin{align*}
        \pE \left[ \muest{\BS}{b,t}(\boldone)^m \right] & \leq C\pE \left[ \muest{\BS}{b,t-1}(\boldone)^m \sum_{p = 1}^m  \frac{\Omega^{2m} _{t-1} N^p}{N^{m} \Omega^{m + p}} \right] \leq C\pE \left[ \muest{\BS}{b,t-1}(\boldone)^m \right]  \eqsp,
    \end{align*}
    which completes the proof.

    If $N < 2m$ (resp. $N < m$), then a tuple in $\mathcal{I}^m _0$ (resp. $\mathcal{I}^m _1$) contains at most $N$ different elements and the proof proceeds similarly by truncating the sums over $p$. 
\end{proof}

We now give more explicit computations for the case $m = 2$. The sets $\mathcal{I}^2 _0 \cap \sett^p _2$ and $\mathcal{I}^2 _1 \cap \sett^p _2$ are detailed in Example~\ref{ex:explicitset}.
\begin{proposition}
    \label{cor:condexpect}
    For any $h \in \bounded{2(t+1)}$ and $(k^1 _{0:t-1}, \cdots, k^4 _{0:t-1}) \in [N]^{4t}$,
    \begin{equation*}
    \begin{alignedat}{6}
        & \sum_{k^{1:4} _t \in [N]^4} \pE \left[h^{\otimes 2}(\particle{k^1 _{0:t}}{0:t}, \cdots, \particle{k^4 _{0:t}}{0:t}) \prod_{j = 1}^4 \beta^\BS _t(k^j _t, k^j _{t-1}) \1_{k^1 _t \neq k^2 _t, k^3 _t \neq k^4 _t} \bigg| \F{t-1} \right] \\
        &  \hspace{2cm}\leq \frac{N(N-1)(N-2)(N-3)}{\Omega^4 _{t-1}} (g_{t-1} ^{\otimes 2} \bitransition{0}{t}[h])^{\otimes 2}(\particle{k^1 _{0:t-1}}{0:t-1}, \cdots, \particle{k^4 _{0:t-1}}{0:t-1}) \\ 
        & \hspace{2.4cm}+ \frac{N(N-1)(N-2)\boundg^3 |h|^2 _\infty}{\Omega^3 _{t-1}} \vartheta^N_t(\particle{k^{1:4}_{t-1}}{t-1})  + \frac{N(N-1)\boundg^2 |h|^2 _\infty}{\Omega^2 _{t-1}} \upsilon^N_t(\particle{k^{1:4}_{t-1}}{t-1})  \eqsp,
    \end{alignedat}
\end{equation*}
and 
\begin{multline*}
     \sum_{k^{1:4} _t \in [N]^4} \pE \left[h^{\otimes 2}(\particle{k^1 _{0:t}}{0:t}, \cdots, \particle{k^4 _{0:t}}{0:t}) \prod_{j = 1}^2 \beta^\BS _t(k^{2j-1} _t, k^{2j-1} _{t-1}) \normweight{k^{2j} _{t-1}}{t-1} \1_{k^1 _t = k^2 _t, k^3 _t = k^4 _t} \bigg| \F{t-1} \right] \\
     \leq \frac{N(N-1)}{\Omega^{4} _{t-1}} (g_{t-1} ^{\otimes 2}\bitransition{1}{t}[h])^{\otimes 2}(\particletraj{k^1}{t}, \cdots, \particletraj{k^4}{t}) + \frac{N\boundg^3}{\Omega^3 _{t-1}} \transition{t}\left[ \beta^N _t(., \particle{k^3 _{t-1}}{t-1})\right](\particle{k^1 _{t-1}}{t-1})  \eqsp.
\end{multline*}
where 
\begin{align*}
    \vartheta^N_t : x^{1:4} &\mapsto \transition{t}[ \beta^N _t(.,x^{4})](x^{1}) + \transition{t}[ \beta^N _t(.,x^{4})](x^{2}) + \transition{t}[ \beta^N _t(.,x^{3})](x^{1}) + \transition{t}[ \beta^N _t(.,x^{3})](x^{2})\eqsp,\\
    \upsilon^N_t : x^{1:4} &\mapsto \transition{t}[ \beta^N _t(.,x^3)](x^1) \transition{t}[\beta^N _t(.,x^4)](x^2) + \transition{t}[ \beta^N _t(.,x^4)](x^1) \transition{t}[ \beta^N _t(.,x^3)](x^2)\eqsp.
\end{align*}
\end{proposition}
\begin{proof}
Let $(k^1 _{0:t-1}, \cdots, k^4 _{0:t-1}) \in [N]^{4t}$. First note that 
    \begin{multline*}
        \sum_{k^{1:4} _t \in [N]^4} \pE \left[h^{\otimes 2}(\particle{k^1 _{0:t}}{0:t}, \cdots, \particle{k^4 _{0:t}}{0:t}) \prod_{j = 1}^4 \beta^\BS _t(k^j _t, k^j _{t-1}) \1_{k^1 _t \neq k^2 _t, k^3 _t \neq k^4 _t} \bigg| \F{t-1} \right] \\
        = \sum_{p = 2}^4 \sum_{k^{1:4} _t \in \mathcal{I}^2 _0 \cap \sett^p _2} \pE \left[h^{\otimes 2}(\particle{k^1 _{0:t}}{0:t}, \cdots, \particle{k^4 _{0:t}}{0:t}) \prod_{j = 1}^4 \beta^\BS _t(k^j _t, k^j _{t-1}) \bigg| \F{t-1} \right] \eqsp.
    \end{multline*}
We compute each term  herebelow. For each $p\in [2:4]$ and $\bm{k} \eqdef (k^1 _t, \cdots, k^{4} _t)\in \mathcal{I}^2 _0 \cap \sett^p _2$, let $\bm{k}_p \eqdef \{ a_1, \cdots, a_p \} = \{k^1 _t, \cdots, k^{4} _t\}$ the set of cardinal $p$ containing the $p$ distinct elements in a tuple $\bm{k}\in \sett^p _2$. Define for any $a_i \in \bm{k}_p$, $V_{a_i} \eqdef \{k^j _t: j \in [1:4],\  k^j _t = a_i\}$.
\vspace{.3cm}\\
$-$   Let $(k^1 _t, \cdots, k^4 _t) \in \mathcal{I}^2 _0 \cap \sett^2 _2$. Then, $\bm{k}_2 = \{k^1 _t, k^2 _t\}$ and we either have $V_{k^1 _t} = \{ k^1 _t, k^3 _t\}$ and $V_{k^2} = \{k^2 _t, k^4 _t\}$ or $V_{k^1 _t} = \{k^1, k^4\}$ and $V_{k^2 _t} = \{k^2 _t, k^3 _t\}$. Assume that $V_{k^1 _t} = \{ k^1 _t, k^3 _t\}$ and $V_{k^2 _t} = \{k^2 _t, k^4 _t\}$. Then, by \eqref{lem:identity} in Lemma~\ref*{lem:BSGTidentity}
\begin{align*}
    \pE \bigg[\prod_{j = 1}^4 & \beta^\BS _t(k^j _t, k^j _{t-1}) \bigg| \F{t-1} \bigg]\\
     & = \pE \bigg[ \beta^\BS _t(k^1 _t, k^1 _{t-1}) \beta^\BS _t(k^1 _t, k^3 _{t-1})\big| \F{t-1} \bigg] \pE \bigg[ \beta^\BS _t(k^2 _t, k^2 _{t-1}) \beta^\BS _t(k^2 _t, k^4 _{t-1})\big| \F{t-1} \bigg]\\ 
     & =  \int \beta^\BS _t(k^1 _t, k^3 _{t-1}) \normweight{k^1 _{t-1}}{t-1} \transition{t}(\particle{k^1 _{t-1}}{t-1}, \rmd\particle{k^1 _t}{t}) \int \beta^\BS _t(k^2 _t, k^4 _{t-1}) \normweight{k^2 _{t-1}}{t-1} \transition{t}(\particle{k^2 _{t-1}}{t-1}, \rmd\particle{k^2 _t}{t})  \\
     & \leq \frac{\boundg^2}{\Omega^2 _{t-1}} \transition{t}\left[ \beta^N _t(.,\particle{k^3 _{t-1}}{t-1})\right](\particle{k^1 _{t-1}}{t-1}) \transition{t}\left[ \beta^N _t(.,\particle{k^4 _{t-1}}{t-1}) \right](\particle{k^2 _{t-1}}{t-1}) \eqsp.
\end{align*}
If $V_{k^1 _t} = \{k^1 _t, k^4 _t\}$ and $V_{k^2 _t} = \{k^2 _t, k^3 _t\}$,
$$
    \pE \left[\prod_{j = 1}^4 \beta^\BS _t(k^j _t, k^j _{t-1}) \bigg| \F{t-1} \right] 
    \leq \frac{\boundg^2}{\Omega^2 _{t-1}} \transition{t}\left[ \beta^N _t(.,\particle{k^4 _{t-1}}{t-1})\right](\particle{k^1 _{t-1}}{t-1}) \transition{t}\left[ \beta^N _t(.,\particle{k^3 _{t-1}}{t-1})\right](\particle{k^2 _{t-1}}{t-1}) \eqsp,
$$
and 
\begin{align*}
    &  \sum_{k^{1:4} _t \in \mathcal{I}^2 _0 \cap \sett^2 _2} \pE \left[ \prod_{j = 1}^4 \beta^\BS _t(k^j _t, k^j _{t-1}) \bigg| \F{t-1} \right] \\
    & \hspace{1.5cm} \leq \frac{N(N-1)\boundg^2}{\Omega^2 _{t-1}} \bigg\{ \transition{t}\left[ \beta^N _t(.,\particle{k^3 _{t-1}}{t-1})\right](\particle{k^1 _{t-1}}{t-1}) \transition{t}\left[ \beta^N _t(.,\particle{k^4 _{t-1}}{t-1})\right](\particle{k^2 _{t-1}}{t-1}) \\
     & \hspace{4cm} + \transition{t}\left[ \beta^N _t(.,\particle{k^4 _{t-1}}{t-1})\right](\particle{k^1 _{t-1}}{t-1}) \transition{t}\left[ \beta^N _t(.,\particle{k^3 _{t-1}}{t-1})\right](\particle{k^2 _{t-1}}{t-1})\bigg\} \eqsp.
\end{align*}
\vspace{.3cm}\\
$-$ Let $(k^1 _t, \cdots, k^4 _t) \in \mathcal{I}^2 _0 \cap \sett^3 _2$. Then, either $\bm{k}_3 = \{k^1 _t, k^2 _t, k^3 _t\}$ or $\bm{k}_3 = \{k^1 _t, k^2 _t, k^4 _t\}$. Assume that $\bm{k}_3 = \{k^1 _t, k^2 _t, k^3 _t\}$ and $V_{k^1 _t} = \{k^1 _t, k^4 _t\}$. Then, 
\begin{align*}
    \pE \left[\prod_{j = 1}^4 \beta^\BS _t(k^j _t, k^j _{t-1}) \bigg| \F{t-1} \right] & = \normweight{k^1 _{t-1}}{t-1}\normweight{k^2 _{t-1}}{t-1}\normweight{k^3 _{t-1}}{t-1} \transition{t}\left[ \beta^N _t(.,\particle{k^4 _{t-1}}{t-1})\right](\particle{k^1 _{t-1}}{t-1}) \\
    & \leq \frac{\boundg^3}{\Omega^3 _{t-1}} \transition{t}\left[ \beta^N _t(.,\particle{k^4 _{t-1}}{t-1})\right](\particle{k^1 _{t-1}}{t-1}) \eqsp.
\end{align*}    
Applying the same reasoning to all the combinations within $\mathcal{I}^2 _0 \cap \sett^3 _2$ we get 
\begin{align*}
    &  \sum_{k^{1:4} _t \in \mathcal{I}^2 _0 \cap \sett^3 _2} \pE \left[ \prod_{j = 1}^4 \beta^\BS _t(k^j _t, k^j _{t-1}) \bigg| \F{t-1} \right] \\ 
    & \hspace{1cm} \leq \frac{N(N-1)(N-2)\boundg^3}{\Omega^3 _{t-1}} \bigg\{ \transition{t}\left[ \beta^N _t(.,\particle{k^4 _{t-1}}{t-1})\right](\particle{k^1 _{t-1}}{t-1}) + \transition{t}\left[ \beta^N _t(.,\particle{k^4 _{t-1}}{t-1})\right](\particle{k^2 _{t-1}}{t-1}) \\
    & \hspace{3cm} + \transition{t}\left[ \beta^N _t(.,\particle{k^3_{t-1}}{t-1})\right](\particle{k^1 _{t-1}}{t-1}) + \transition{t}\left[ \beta^N _t(.,\particle{k^3_{t-1}}{t-1})\right](\particle{k^2 _{t-1}}{t-1}) \bigg\} \eqsp.
\end{align*}
\vspace{.3cm}\\
$-$ Let $(k^1 _t, \cdots, k^4 _t) \in \mathcal{I}^4 _0 \cap \sett^4 _2$. Then, $\bm{k}_4 = \{k^1 _t, k^2 _t, k^3 _t, k^4 _t\}$ and 
\begin{align*}
    & \pE \left[h^{\otimes 2}(\particletraj{k^1}{t}, \cdots, \particletraj{k^4}{t}) \prod_{j = 1}^4 \beta^\BS _t(k^j _t, k^j _{t-1}) \bigg| \F{t-1} \right] \\
    & \hspace{2cm}= \int h^{\otimes 2}(\particletraj{k^1}{t}, \cdots, \particletraj{k^4}{t}) \prod_{j = 1}^4 \beta^\BS _t(k^j _t, k^j _{t-1})\filter{t-1}^N \transition{t}( \rmd \particle{k^j _t}{t}) \\
    & \hspace{2cm}= \int h^{\otimes 2}(\particletraj{k^1}{t}, \cdots, \particletraj{k^4}{t}) \prod_{j = 1}^4 \normweight{k^j _{t-1}}{t-1} \transition{t}(\particle{k^j _{t-1}}{t-1}, \rmd \particle{k^j _t}{t}) \eqsp. \\
    & \hspace{2cm} = (g_{t-1} ^{\otimes 2} \bitransition{0}{t}[h])^{\otimes 2}(\particletraj{k^1}{t-1}, \cdots, \particletraj{k^4}{t-1}) / \Omega^4 _{t-1} \eqsp,
\end{align*}
which completes the proof of the first inequality. For the second inequality, write
\begin{multline*}
    \sum_{k^{1:4} _t \in [N]^4} \pE \left[h^{\otimes 2}(\particle{k^1 _{0:t}}{0:t}, \cdots, \particle{k^4 _{0:t}}{0:t}) \prod_{j = 1}^2 \beta^\BS _t(k^{2j-1} _t, k^{2j-1} _{t-1}) \normweight{k^{2j} _{t-1}}{t-1} \1_{k^1 _t = k^2 _t, k^3 _t = k^4 _t} \bigg| \F{t-1} \right] \\
    = \sum_{p = 1}^2 \sum_{k^{1:4} _t \in \mathcal{I}^2 _1 \cap \sett^p _2} \pE \left[h^{\otimes 2}(\particle{k^1 _{0:t}}{0:t}, \cdots, \particle{k^4 _{0:t}}{0:t}) \prod_{j = 1}^2 \beta^\BS _t(k^{2j-1} _t, k^{2j-1} _{t-1}) \normweight{k^{2j} _{t-1}}{t-1} \bigg| \F{t-1} \right] \eqsp,
\end{multline*}
\vspace{.3cm}\\
$-$ Let $(k^1 _t, \cdots, k^4 _t) \in \mathcal{I}^2 _1 \cap \sett^1 _2$. Then $\bm{k}_1 = \{k^1 _t\}$ and 
\begin{align*}
    \pE \left[\prod_{j = 1}^2 \beta^\BS _t(k^{2j - 1} _t, k^{2j -1} _{t-1}) \normweight{k^{2j} _{t-1}}{t-1} \bigg| \F{t-1} \right]  & \leq \normweight{k^1 _{t-1}}{t-1}\normweight{k^2 _{t-1}}{t-1}\normweight{k^4 _{t-1}}{t-1} \transition{t}[\beta^N _t(., \particle{k^3 _{t-1}}{t-1})](\particle{k^1 _{t-1}}{t-1}) \\
    & \leq \frac{\boundg^3}{\Omega^3 _{t-1}} \transition{t}\left[ \beta^N _t(., \particle{k^3 _{t-1}}{t-1})\right](\particle{k^1 _{t-1}}{t-1}) \eqsp,
\end{align*}
and 
\begin{equation*}
\sum_{k^{1:4} _t \in \mathcal{I}^2 _1 \cap \sett^1 _2} \pE \left[\prod_{j = 1}^2 \beta^\BS _t(k^{2j - 1} _t, k^{2j -1} _{t-1}) \normweight{k^{2j} _{t-1}}{t-1} \bigg| \F{t-1} \right]
 \leq \frac{N\boundg^3}{\Omega^3 _{t-1}} \transition{t}\left[ \beta^N _t(., \particle{k^3 _{t-1}}{t-1})\right](\particle{k^1 _{t-1}}{t-1}) \eqsp.
\end{equation*}
\vspace{.3cm}\\
$-$ Let $(k^1 _t, \cdots, k^4 _t) \in \mathcal{I}^2 _1 \cap \sett^2 _2$. Then, $\bm{k}_2 = \{k^1 _t, k^3 _t\}$, $V_{k^1 _t} = \{k^1 _t, k^2 _t\}$ and $V_{k^3 _t} = \{k^3 _t, k^4 _t\}$. Hence,
\begin{multline*}
     \pE \left[h^{\otimes 2}(\particle{k^1 _{0:t}}{0:t}, \cdots, \particle{k^4 _{0:t}}{0:t}) \prod_{j = 1}^2 \beta^\BS _t(k^{2j-1} _t, k^{2j-1} _{t-1}) \normweight{k^{2j} _{t-1}}{t-1}\bigg| \F{t-1} \right] \\  
     = (g_{t-1} ^{\otimes 2}\bitransition{1}{t}[h])^{\otimes 2}(\particle{k^1 _{0:t-1}}{0:t-1}, \cdots, \particle{k^4 _{0:t-1}}{0:t-1}) / \Omega^4 _{t-1} \eqsp,
\end{multline*}
and 
\begin{multline*}
     \sum_{k^{1:4} _{t} \in \mathcal{I}^2 _1 \cap \sett^2 _2}\pE \left[h^{\otimes 2}(\particle{k^1 _{0:t}}{0:t}, \cdots, \particle{k^4 _{0:t}}{0:t}) \prod_{j = 1}^2 \beta^\BS _t(k^{2j-1} _t, k^{2j-1} _{t-1}) \normweight{k^{2j} _{t-1}}{t-1}\bigg| \F{t-1} \right] \\  
     = \frac{N(N-1)}{\Omega^{4} _{t-1}} (g_{t-1} ^{\otimes 2}\bitransition{1}{t}[h])^{\otimes 2}(\particle{k^1 _{0:t-1}}{0:t-1}, \cdots, \particle{k^4 _{0:t-1}}{0:t-1})  \eqsp,
\end{multline*}
which completes the proof.
\end{proof} 
\begin{proposition} For any $t \in \N$, $h \in \bounded{2(t+1)}$ and $b \in \Bset_t$,
    \label{prop:norm2QBS}
    \begin{enumerate}[label=(\roman*)]
        \item \label{eq:finalnorm2b0}If $b_t = 0$,
        \begin{align*}
              &\big\| \muest{\BS}{b,t}(h) \big\|^2 _2 \leq  \frac{(N-2)(N-3)}{N(N-1)}  \big\|\muest{\BS}{b,t-1}(g^{\otimes 2} _{t-1}\bitransition{0}{t}[h]) \big\|_2 ^2  + \frac{N-2}{N-1} \boundg^3 | h |^2 _\infty \pE \bigg[ \frac{\Omega_{t-1}}{N} \nu(\Theta_{b,t}^N)\bigg] \nonumber \\
               & \hspace{6.5cm} + \pE \bigg[ \frac{\boundg^2 |h |^2 _\infty \Omega_{t-1} ^2}{N(N-1)}\nu^{\otimes 2}(\Upsilon_{b,t}^N)\bigg] \eqsp, \nonumber
          \end{align*}
          where 
          \begin{align*}
            \Theta^N _{b,t} & : x \mapsto \big[ \muest{\BS}{b,t-1}\big(\transitiondens{t}(.,x) \otimes \boldone \big)  + \muest{\BS}{b,t-1} \big(\boldone \otimes \transitiondens{t}(.,x) \big)\big] \\ &\hspace{3cm} \times \big[ \muest{\BS}{b,t-1}\big(\beta^N _t(x,.) \otimes \boldone) + \muest{\BS}{b,t-1}(\boldone \otimes \beta^N _t(x,.)\big)\big] \eqsp,\\
       \Upsilon^N _{b,t}&: (x,y) \mapsto \muest{\BS}{b,t-1}\big(\transitiondens{t}(.,x) \otimes \transitiondens{t}(.,y) \big) \muest{\BS}{b,t-1}\big(\beta^N _{t}(x,.) \otimes \beta^N _{t}(y,.)\big) \\
       &\hspace{3cm}+ \muest{\BS}{b,t-1}\big(\transitiondens{t}(.,x) \otimes \transitiondens{t}(.,y)\big)  \muest{\BS}{b,t-1}\big(\beta^N _{t}(y,.) \otimes \beta^N _{t}(x,.)\big) \eqsp. 
    \end{align*}
          \item \label{eq:finalnorm2b1}
          If $b_t = 1$,
          \begin{multline}    
              \big\| \muest{\BS}{b,t}(h) \big\|^2 _2 
               \leq  \frac{N-1}{N}  \big\| \muest{\BS}{b,t-1}(g^{\otimes 2} _{t-1}\bitransition{1}{t}[h]) \big\|^2 _2  \\
               + \boundg^3 |h|_\infty ^2 \int \pE \left[ \frac{\Omega _{t-1}}{N} \muest{\BS}{b,t-1}(\transitiondens{t}(.,x) \otimes \boldone) \muest{\BS}{b,t-1}(\beta^N _{t}(x,.) \otimes \boldone) \right] \nu(\rmd x) \eqsp. 
          \end{multline}
\end{enumerate}
\end{proposition}
\begin{proof}
    To prove \ref{eq:finalnorm2b0}, if $b_t = 0$, by \eqref{eq:alternativeQ},
    \begin{align*}
        \muest{\BS}{b,t}(h)^2 & = \sum_{k^{1:4} _{0:t} \in [N]^{4(t+1)}} \bigprod{b,t}(k^1 _{0:t}, k^2 _{0:t}) \bigprod{b,t}(k^3 _{0:t}, k^4 _{0:t})h^{\otimes 2}(\particletraj{k^1}{t}, \cdots, \particletraj{k^4}{t} ) \\
        & = \sum_{k^{1:4} _{0:t-1} \in [N]^{4t}} \bigprod{b,t}(k^1 _{0:t-1}, k^2 _{0:t-1}) \bigprod{b,t}(k^3 _{0:t-1}, k^4 _{0:t-1}) \\
        & \hspace{1cm} \times \bigg\{ \sum_{p = 2}^4 \sum_{k^{1:4} _t \in \mathcal{I}^2 _0 \cap \sett^p _2} \frac{\Omega^4 _{t-1}}{N^2(N-1)^2} \prod_{j = 1}^4 \beta^\BS _t(k^j _t, k^j _{t-1}) h^{\otimes 2}(\particletraj{k^1}{t}, \cdots, \particletraj{k^4}{t} )\bigg\} \eqsp,
        \end{align*}
        and by Proposition~\ref{cor:condexpect}, 
        \begin{align*}
        \big\| \muest{\BS}{b,t}(h) \big\|_2 ^2 & \leq \frac{(N-2)(N-3)}{N(N-1)} \left\| \muest{\BS}{b,t-1}(g^{\otimes 2} _{t-1} \bitransition{0}{t}[h]) \right\|_2 ^2 \\
        & \hspace{.2cm} + \frac{\boundg^3 |h|^2 _\infty (N-2)}{N(N-1)}\mathbb{E}\left[ \Omega_{t-1} \sum_{k^{1:4} _{0:t-1} \in [N]^{4t}} \bigprod{b,t-1}(k^1 _{0:t-1}, k^2 _{0:t-1}) \bigprod{b,t-1}(k^3 _{0:t-1}, k^4 _{0:t-1})\right. \\
        &\hspace{1.5cm}  \times \bigg\{ \transition{t}\left[ \beta^N _t(.,\particle{k^4 _{t-1}}{t-1})\right](\particle{k^1 _{t-1}}{t-1}) + \transition{t}\left[ \beta^N _t(.,\particle{k^4 _{t-1}}{t-1})\right](\particle{k^2 _{t-1}}{t-1}) \\
        & \hspace{3cm} \left.+ \transition{t}\left[ \beta^N _t(.,\particle{k^3_{t-1}}{t-1})\right](\particle{k^1 _{t-1}}{t-1}) + \transition{t}\left[ \beta^N _t(.,\particle{k^3_{t-1}}{t-1})\right](\particle{k^2 _{t-1}}{t-1}) \bigg\}\right] \\
        & \hspace{.5cm} + \frac{\boundg^2 |h|_\infty ^2 }{N(N-1)} \mathbb{E}\left[\Omega^2 _{t-1}\sum_{k^{1:4} _{0:t-1} \in [N]^{4t}} \bigprod{b,t-1}(k^1 _{0:t-1}, k^2 _{0:t-1}) \bigprod{b,t-1}(k^3 _{0:t-1}, k^4 _{0:t-1})\right. \\
        &\hspace{1.5cm} \times \bigg\{ \transition{t}\left[ \beta^N _t(.,\particle{k^3 _{t-1}}{t-1})\right](\particle{k^1 _{t-1}}{t-1}) \transition{t}\left[ \beta^N _t(.,\particle{k^4 _{t-1}}{t-1})\right](\particle{k^2 _{t-1}}{t-1}) \\
        & \hspace{3cm} \left.+ \transition{t}\left[ \beta^N _t(.,\particle{k^4 _{t-1}}{t-1})\right](\particle{k^1 _{t-1}}{t-1}) \transition{t}\left[ \beta^N _t(.,\particle{k^3 _{t-1}}{t-1})\right](\particle{k^2 _{t-1}}{t-1}) \bigg\}\right] \eqsp.
    \end{align*}    
    Then using $\As{assp:A}{}$,
    \begin{align*}
            & \sum_{k^{1:4} _{0:t-1}\in [N]^{4t}} \bigprod{b,t-1}(k^1 _{0:t-1}, k^2 _{0:t-1}) \bigprod{b,t-1}(k^3 _{0:t-1}, k^4 _{0:t-1}) \transition{t}[\beta^N _t(., \particle{k^3 _{t-1}}{t-1})](\particle{k^1 _{t-1}}{t-1}) \\
            & = \int \sum_{k^{1:4} _{0:t-1}\in [N]^{4t}} \bigprod{b,t-1}(k^1 _{0:t-1}, k^2 _{0:t-1}) \bigprod{b,t-1}(k^3 _{0:t-1}, k^4 _{0:t-1})  \beta^N _{t}(x, \particle{k^3 _{t-1}}{t-1}) \transitiondens{t}(\particle{k^1 _{t-1}}{t-1}, x) \nu(\rmd x) \nonumber\\
            & = \int  \muest{\BS}{b,t-1}\left( \transitiondens{t}(. ,x) \otimes \boldone\right) \muest{\BS}{b,t-1}\left( \beta^N _{t}(x,.) \otimes \boldone  \right) \nu(\rmd x) \eqsp, \nonumber
        \end{align*}
        and
        \begin{align*}
         & \sum_{k^{1:4} _{0:t-1}\in [N]^{4t}} \bigprod{b,t-1}(k^1 _{0:t-1}, k^2 _{0:t-1}) \bigprod{b,t-1}(k^3 _{0:t-1}, k^4 _{0:t-1}) \\
         & \hspace{4cm} \times \transition{t}[\beta^N _t(., \particle{k^3 _{t-1}}{t-1})](\particle{k^1 _{t-1}}{t-1}) \transition{t}[\beta^N _{t}(., \particle{k^4 _{t-1}}{t-1})](\particle{k^2 _{t-1}}{t-1}) \nonumber \\
            & \hspace{1cm} = \int \muest{\BS}{b,t-1}(\transitiondens{t}(.,x) \otimes \transitiondens{t}(.,y))  \muest{\BS}{b,t-1}\left( \beta^N _{t}(x,.) \otimes \beta^N _{t}(y,.) \right) \nu^{\otimes 2}(\rmd x, \rmd y)  \nonumber \eqsp.
        \end{align*}  
        which completes the proof of  \ref{eq:finalnorm2b0} by applying the same reasoning to the remaining terms. Item \ref{eq:finalnorm2b1} is obtained in the same way.
\end{proof}
\subsection{Supporting results for Theorem~\ref*{corr:paris}}
\label{sec:supportparis}

In this section, we prove the analogues of Propositions~\ref*{prop:mu_expression}-\ref*{prop:Qbound}-\ref*{lem:prodbeta} and ~\ref*{cor:condexpect} for $\parismuest{\BS}{b,t}$. We remind the reader that the number of sampled indices $M$ in the $\paris$ estimator is \emph{fixed} and that $\Gfilt{b}{t}$ is defined in \eqref{eq:Gfilt}. Let $\Gfilt{b}{t-1} \vee \particle{1:N}{t}$ be the following $\sigma$-algebra:
 \begin{equation}
    \label{eq:Gfiltparticle}
    \Gfilt{b}{t-1} \vee \particle{1:N}{t} \eqdef \sigma\big( \Gfilt{b}{t-1} \cup \sigma(\particle{1:N}{t}) \big) \eqsp.
\end{equation}
\begin{lemma}
    \label{lem:condexpectparis}
For any $t \in \N^{*}$, $h \in \measurable{2}$ and $b \in \Bset_t$, 
\begin{enumerate}[label=(\roman*)]
    \item $\pE \big[ \parismuest{\BS}{b,t}(h) \big| \Gfilt{b}{t-1} \big]= \parismuest{\BS}{b,t-1}(g^{\otimes 2} _{t-1} \bitransition{b_t}{t}[h])$,
    \item $\parismuest{\BS}{b,t}(h)$ is an unbiased estimator of $\mumeasure{b,t}(h)$,
\end{enumerate}
where $\parismuest{\BS}{b,t}(h)$ is defined in \eqref{eqdef:qtilde}.
\end{lemma}
\begin{proof}
We start with the case $b_t = 0$. For any $(k, \ell) \in [N]^2$, $J^i _{k, t-1}$ and $J^i _{\ell, t-1}$ are independent conditionally on $\Gfilt{b}{t-1} \vee \particle{1:N}{t}$ for any $i \in [M]$ if $k \neq \ell$ and $\Pbacksum^b _{t-1}$ is $\Gfilt{b}{t-1}$-measurable, hence, using \eqref{eqdef:tautilde},
    \begin{align}
        \label{eq:pariscondexpect0}
        & \pE \left[ \Pbacksum^b _t(k,\ell) h(\particle{k}{t}, \particle{\ell}{t}) \middle| \Gfilt{b}{t-1} \right] \\
        & \hspace{2cm} = \pE \big[ \1_{k \neq \ell} h(\particle{k}{t}, \particle{\ell}{t}) M^{-1} \sum_{i = 1}^M \Pbacksum^b _{t-1}(J^i _{k,t-1}, J^i _{\ell, t-1}) \big| \Gfilt{b}{t-1} \big] \nonumber\\
        & \hspace{2cm} = \pE \bigg[  \1_{k \neq \ell} h(\particle{k}{t}, \particle{\ell}{t}) M^{-1} \sum_{i = 1}^M \sum_{n, m \in [N]^2} \beta^\BS _t(k,n) \beta^\BS _t(\ell,m) \Pbacksum^b _{t-1}(n,m) \big| \Gfilt{b}{t-1} \bigg] \nonumber \\
        & \hspace{2cm} =  \sum_{n,m \in [N]^2} \Pbacksum^b _{t-1}(n,m) \pE \bigg[ \1_{k \neq \ell} \beta^\BS _t(k,n) \beta^\BS _t(\ell,m)  h(\particle{k}{t}, \particle{\ell}{t}) \big| \F{t-1} \bigg] \nonumber\\
        & \hspace{2cm} = \frac{\1_{k \neq \ell}}{\Omega^2 _{t-1}} \sum_{n, m \in [N]^2} \Pbacksum^b _{t-1}(n,m) \big( g^{\otimes 2} _{t-1} \bitransition{0}{t}[h] \big)(\particle{n}{t-1}, \particle{m}{t-1}) \nonumber \eqsp.
    \end{align}    
    If $b_t = 1$, 
    \begin{align}
        \label{eq:pariscondexpect1}
        & \hspace{1cm}\pE \left[ \Pbacksum^b _t(k,\ell) h(\particle{k}{t}, \particle{\ell}{t}) \middle| \Gfilt{b}{t-1} \right] \\
        & \hspace{2cm} = \pE \big[ \1_{k = \ell} h(\particle{k}{t}, \particle{\ell}{t}) M^{-1} \sum_{i = 1}^M \sum_{m = 1}^N \normweight{m}{t-1} \Pbacksum^b _{t-1}(J^i _{k,t-1}, m) \big| \Gfilt{b}{t-1} \big] \nonumber \\
        & \hspace{2cm} = \pE \bigg[  \1_{k = \ell} h(\particle{k}{t}, \particle{\ell}{t}) M^{-1} \sum_{i = 1}^M \sum_{n, m \in [N]^2} \beta^\BS _t(k,n) \normweight{m}{t-1} \Pbacksum^b _{t-1}(n,m) \big| \Gfilt{b}{t-1} \bigg] \nonumber \\
        & \hspace{2cm} =  \sum_{n,m \in [N]^2} \Pbacksum^b _{t-1}(n,m) \pE \bigg[ \1_{k = \ell} \beta^\BS _t(k,n) \normweight{m}{t-1}  h(\particle{k}{t}, \particle{\ell}{t}) \big| \F{t-1} \bigg] \nonumber \\
        & \hspace{2cm} = \frac{\1_{k = \ell}}{\Omega^2 _{t-1}} \sum_{n, m \in [N]^2} \Pbacksum^b _{t-1}(n,m) \big( g^{\otimes 2} _{t-1} \bitransition{1}{t}[h] \big)(\particle{n}{t-1}, \particle{m}{t-1}) \eqsp. \nonumber
    \end{align}   
    Consequently, if $b_t = 0$, by \eqref{eq:pariscondexpect0}, 
    \begin{align*}
        & \pE \big[ \parismuest{\BS}{b,t}(h) \big| \Gfilt{b}{t-1} \big] \\
        & \hspace{.5cm} = \prod_{s = 0}^{t-1} N^{b_s} \left( \frac{N}{N-1}\right)^{1 - b_s} \frac{\joint{t-1}^N(\boldone)^2}{N^2} \sum_{k, \ell \in [N]^2} \frac{\Omega^2 _{t-1}}{N(N-1)}  \pE \bigg[ \Pbacksum^b _{t}(k,\ell) h(\particle{k}{t}, \particle{\ell}{t}) \big| \Gfilt{b}{t-1}\bigg] \\
        & \hspace{.5cm} = \prod_{s = 0}^{t-1} N^{b_s} \left( \frac{N}{N-1}\right)^{1 - b_s} \frac{\joint{t-1}^N(\boldone)^2}{N^2} \sum_{n, m \in [N]^2} \Pbacksum^b _{t-1}(n,m) \big( g^{\otimes 2} _{t-1} \bitransition{0}{t}[h] \big)(\particle{n}{t-1}, \particle{m}{t-1}) \\
        & \hspace{.5cm} = \parismuest{\BS}{b,t-1}(g^{\otimes 2} _{t-1} \bitransition{0}{t}[h] ) \eqsp,
    \end{align*}
    and in a similar way, $ \pE \big[ \parismuest{\BS}{b,t}(h) \big| \Gfilt{b}{t-1} \big] = \parismuest{\BS}{b,t-1}(g^{\otimes 2} _{t-1} \bitransition{1}{t}[h] )$ by \eqref{eq:pariscondexpect1} if $b_t = 1$. The second item follows straightforwardly by induction and the tower property. The induction is initialized by noting that $\parismuest{M}{b,0}(h)$ is equal to $\muest{\BS}{b,0}(h)$ which is an unbiased estimator of $\mumeasure{b,0}(h)$ by  Proposition~\ref*{prop:mu_expression}.
\end{proof}
    \begin{proposition}
        \label{prop:paris_conv_partitions} Let $t > 0$ and $N \geq 4$.
        \begin{enumerate}[label=(\roman*)]
        \item \label{item:paris_convb0} If $b_t = 0$, 
        \begin{multline}
 \label{eq:paris_conv_partitions_card4}
             \pE \bigg[ \Const{t}^2 \sum_{k^{1:4} _{t} \in \mathcal{I}^2 _0 \cap \sett^{4} _2} h(\particle{k^1 _t}{t}, \particle{k^2 _t}{t}) h(\particle{k^3 _t}{t}, \particle{k^4 _t}{t})\Pbacksum^b _t(k^1 _t, k^2 _t) \Pbacksum^b _t(k^3 _t, k^4 _t) \bigg| \Gfilt{b}{t-1} \bigg] \\
              =  \frac{(N-2)(N-3)}{N(N-1)} \parismuest{\BS}{b,t-1}(g^{\otimes 2} _{t-1} \bitransition{0}{t}[h])^2  \eqsp, 
        \end{multline}
        \begin{multline}
            \label{eq:paris_conv_partitions_card2}
            \pE \bigg[ \Const{t} ^2 \sum_{k^{1:4} _{t} \in \mathcal{I}^2 _0 \cap \sett^{2} _2} \Pbacksum^b _t(k^1 _t, k^2 _t) \Pbacksum^b _t(k^3 _t, k^4 _t) \bigg| \Gfilt{b}{t-1} \bigg] 
            \leq  \frac{\boundg^2 (M-1) \Omega^2 _{t-1}}{MN(N-1)} \nu^{\otimes 2}( \widetilde{\Upsilon}_{b,t}^{N,M})  \\ + \frac{\boundg^2 \Omega^2 _{t-1}}{MN(N-1)} \Const{t-1}^2 \sum_{k^{1:2} _{t-1} \in [N]^2} \mathsf{T}^{(1)}_{t-1,b}(k^1_{t-1},k^2_{t-1}) \eqsp,
        \end{multline}
        and
        \begin{multline}
            \label{eq:paris_conv_partitions_card3}
             \pE \bigg[ \Const{t}^2 \sum_{k^{1:4} _{t} \in \mathcal{I}^2 _0 \cap \sett^{3} _2} \Pbacksum^b _t(k^1 _t, k^2 _t) \Pbacksum^b _t(k^3 _t, k^4 _t) \bigg| \Gfilt{b}{t-1} \bigg]  \leq \frac{\Omega _{t-1} \boundg^3 (M-1)(N-2)}{MN(N-1)} \nu(\widetilde{\Theta}_{b,t}^{N,M})  \\
            + \frac{\Omega _{t-1} \boundg^3 (N-2)}{MN(N-1)} \Const{t-1}^2 \sum_{k^{1:3} _{t-1} \in [N]^3}
             \mathsf{T}^{(2)}_{t-1,b}(k^{1:3} _{t-1}) \eqsp, 
        \end{multline}
        where 
          \begin{align*}
            \widetilde{\Theta}^{N,M} _{b,t} & : x \mapsto \big[ \parismuest{\BS}{b,t-1}\big(\transitiondens{t}(.,x) \otimes \boldone \big)  + \parismuest{\BS}{b,t-1} \big(\boldone \otimes \transitiondens{t}(.,x) \big)\big] \\ &\hspace{3cm} \times \big[ \parismuest{\BS}{b,t-1}\big(\beta^N _t(x,.) \otimes \boldone) + \parismuest{\BS}{b,t-1}(\boldone \otimes \beta^N _t(x,.)\big)\big] \eqsp,\\
       \widetilde{\Upsilon}^{N,M} _{b,t}&: (x,y) \mapsto \parismuest{\BS}{b,t-1}\big(\transitiondens{t}(.,x) \otimes \transitiondens{t}(.,y) \big) \parismuest{\BS}{b,t-1}\big(\beta^N _{t}(x,.) \otimes \beta^N _{t}(y,.)\big) \\
       &\hspace{3cm}+ \parismuest{\BS}{b,t-1}\big(\transitiondens{t}(.,x) \otimes \transitiondens{t}(.,y)\big)  \parismuest{\BS}{b,t-1}\big(\beta^N _{t}(y,.) \otimes \beta^N _{t}(x,.)\big) \eqsp,
          \end{align*}
        and
        \begin{align*}
       \mathsf{T}^{(1)}_{b,t}&: (k^1 _{t},k^2 _{t})\mapsto\Pbacksum^b _{t}(k^1 _{t}, k^2 _{t})^2 + \Pbacksum^b _{t}(k^1 _{t}, k^2 _{t}) \Pbacksum^b _{t}(k^2 _{t}, k^1 _{t})\eqsp,\\
       \mathsf{T}^{(2)}_{b,t}&: (k^1 _{t},k^2 _{t},k^3 _{t})\mapsto \Pbacksum^b _{t}(k^1 _{t}, k^2 _{t}) \Pbacksum^b _{t}(k^3 _{t}, k^1 _{t}) + \Pbacksum^b _{t}(k^1 _{t}, k^2 _{t}) \Pbacksum^b _{t}(k^1 _{t}, k^3 _{t}) \\
            & \hspace{3cm} + \Pbacksum^b _{t}(k^1 _{t}, k^2 _{t}) \Pbacksum^b _{t}(k^2 _{t}, k^3 _{t}) + \Pbacksum^b _{t}(k^1 _{t}, k^2 _{t}) \Pbacksum^b _{t}(k^3 _{t}, k^2 _{t})\eqsp.
       \end{align*}
        \item \label{item:paris_convb1} If $b_t = 1$, 
        \begin{multline}
            \label{eq:paris_conv_partitions_b1card2}
            \pE \bigg[ \Const{t}^2 \sum_{k^{1:4} _{t} \in \mathcal{I}^2 _1 \cap \sett^{2} _2} h(\particle{k^1 _t}{t}, \particle{k^2 _t}{t}) h(\particle{k^3 _t}{t}, \particle{k^4 _t}{t}) \Pbacksum^b _t(k^1 _t, k^2 _t) \Pbacksum^b _t(k^3 _t, k^4 _t) \bigg| \Gfilt{b}{t-1} \bigg] \\
             =  \frac{N-1}{N}  \parismuest{\BS}{b,t-1}(g^{\otimes 2} _{t-1} \bitransition{1}{t}[h])^2  \eqsp,
       \end{multline}
       and 
        \begin{align}
            \label{eq:paris_conv_partitions_b1card1}
            & \pE \bigg[ \Const{t}^2 \sum_{k^{1:4} _{t} \in \mathcal{I}^2 _1 \cap \sett^{1} _2} \Pbacksum^b _t(k^1 _t, k^2 _t) \Pbacksum^b _t(k^3 _t, k^4 _t) \bigg| \Gfilt{b}{t-1} \bigg] \\
            & \hspace{1cm} \leq \frac{\Omega_{t-1} \boundg^3}{N} \int \parismuest{\BS}{b,t-1}(\transitiondens{t}(.,x) \otimes \boldone) \parismuest{\BS}{b,t-1}(\beta^N _t(x,.) \otimes 1) \nu(\rmd x)  \nonumber \\
            & \hspace{3cm} + \frac{\boundg^3\Omega _{t-1}}{M N} \Const{t-1}^2 \sum_{k^{1,2,4} _{t-1} \in [N]^3} \Pbacksum^b _{t-1}(k^1 _{t-1}, k^2 _{t-1}) \Pbacksum^b _{t-1}(k^1 _{t-1}, k^4 _{t-1})\eqsp. \nonumber
        \end{align}
    \end{enumerate}
    \end{proposition}
    \begin{proof} We start with the case $b_t = 0$. 
        \vspace{.3cm}\\
        $-$ If $(k^1 _t, \cdots, k^4 _{t}) \in \mathcal{I}^2 _0 \cap \sett^4 _2$, then $k^1 _t \neq k^2 _t \neq k^3 _t \neq k^4 _t$ and conditionally on $\Gfilt{b}{t-1} \vee \particle{1:N}{t}$, $J^i _{k^1 _t,t-1}$, $J^i _{k^2 _t,t-1}$, $J^j _{k^3 _t,t-1}$ and $J^j _{k^4 _t,t-1}$ are independent for any $(i,j) \in [M]^2$ and $J^{1:M} _{k^\ell _t, t-1} \iid \beta^\BS _t(k^\ell _t, .)$ for any $\ell \in [1:4]$, thus, for any $(i,j) \in [M]^2$
    \begin{align*}
        \pE \bigg[ \Pbacksum^b _t(J^i _{k^1 _t, t-1} &, J^i _{k^2 _t, t-1})  \Pbacksum^b _t(J^j _{k^3 _t, t-1},J^j _{k^4 _t, t-1}) \bigg| \Gfilt{b}{t-1} \vee \particle{1:N}{t} \bigg] \\
        & = \pE \bigg[ \Pbacksum^b _t(J^i _{k^1 _t, t-1} , J^i _{k^2 _t, t-1}) \bigg| \Gfilt{b}{t-1} \vee \particle{1:N}{t} \bigg] \pE \bigg[ \Pbacksum^b _t(J^i _{k^3 _t, t-1} , J^i _{k^4 _t, t-1}) \bigg| \Gfilt{b}{t-1} \vee \particle{1:N}{t} \bigg] \\
        & = \sum_{k^{1:4} _{t-1} \in [N]^4} \prod_{n = 1}^4 \beta^\BS _t(k^n _t, k^n _{t-1}) \Pbacksum^b _{t-1}(k^1 _{t-1}, k^2 _{t-1}) \Pbacksum^b _{t-1}(k^3 _{t-1}, k^4 _{t-1}) \eqsp,
    \end{align*}
    and
    \begin{align*}
        & \pE \bigg[ h(\particle{k^1 _t}{t}, \particle{k^2 _t}{t}) h(\particle{k^3 _t}{t}, \particle{k^4 _t}{t}) M^{-2} \sum_{i,j \in [M]^2} \Pbacksum^b _t(J^i _{k^1 _t, t-1} , J^i _{k^2 _t, t-1}) \Pbacksum^b _t(J^j _{k^3 _t, t-1},J^j _{k^4 _t, t-1}) \bigg| \Gfilt{b}{t-1} \bigg] \\
        & = \pE \bigg[ h(\particle{k^1 _t}{t}, \particle{k^2 _t}{t}) h(\particle{k^3 _t}{t}, \particle{k^4 _t}{t}) M^{-2} \sum_{i,j \in [M]^2} \sum_{k^{1:4} _{t-1} \in [N]^4} \prod_{n = 1}^4 \beta^\BS _t(k^n _t, k^n _{t-1}) \\
        & \hspace{6cm} \times \Pbacksum^b _{t-1}(k^1 _{t-1}, k^2 _{t-1}) \Pbacksum^b _{t-1}(k^3 _{t-1}, k^4 _{t-1})\bigg| \Gfilt{b}{t-1} \bigg] \\
        & = \sum_{k^{1:4} _{t-1} \in [N]^4} \Pbacksum^b _{t-1}(k^1 _{t-1}, k^2 _{t-1}) \Pbacksum^b _{t-1}(k^3 _{t-1}, k^4 _{t-1}) \pE \bigg[ \prod_{n = 1}^4 \beta^\BS _t(k^n _t, k^n _{t-1})  h(\particle{k^1 _t}{t}, \particle{k^2 _t}{t}) h(\particle{k^3 _t}{t}, \particle{k^4 _t}{t}) \bigg| \Gfilt{b}{t-1} \bigg] \\
        & = \frac{1}{\Omega^4 _{t-1}} \sum_{k^{1:4} _{t-1} \in [N]^4} \Pbacksum^b _{t-1}(k^1 _{t-1}, k^2 _{t-1}) \Pbacksum^b _{t-1}(k^3 _{t-1}, k^4 _{t-1}) \big( g^{\otimes 2} _{t-1} \bitransition{0}{t}[h] \big)^{\otimes 2} (\particle{k^1 _{t-1}}{t-1}, \particle{k^2 _{t-1}}{t-1}, \particle{k^3 _{t-1}}{t-1}, \particle{k^4 _{t-1}}{t-1}) \eqsp,
    \end{align*}
    where we have used that $\Pbacksum^b _{t-1}(k^1 _{t-1}, k^2 _{t-1})$ and $\Pbacksum^b _{t-1}(k^3 _{t-1}, k^4 _{t-1})$ are $\Gfilt{b}{t-1}$-measurable in the third equality and then proceeded similarly to Proposition~\ref{cor:condexpect}. By Example~\ref{ex:explicitset}, $\card{\mathcal{I}^2 _0 \cap \sett^4 _2} = N(N-1)(N-2)(N-3)$ and since $b_t = 0$ implies that
    \begin{equation}
        \label{eq:ConstNidentity}
    \Const{t}^2 = \Const{t-1}^2 \frac{ \Omega^4 _{t-1}}{N^2(N-1)^2}    \eqsp,
    \end{equation}
    we obtain
    \begin{align*}
        & \pE \bigg[ \Const{t}^2 \sum_{k^{1:4} _{t} \in \mathcal{I}^2 _0 \cap \sett^{4} _2} h(\particle{k^1 _t}{t}, \particle{k^2 _t}{t}) h(\particle{k^3 _t}{t}, \particle{k^4 _t}{t})\Pbacksum^b _t(k^1 _t, k^2 _t) \Pbacksum^b _t(k^3 _t, k^4 _t) \bigg| \Gfilt{b}{t-1} \bigg] \\
        & =  \Const{t}^2 \sum_{k^{1:4} _{t} \in \mathcal{I}^2 _0 \cap \sett^{4} _2} \pE \bigg[ h(\particle{k^1 _t}{t}, \particle{k^2 _t}{t}) h(\particle{k^3 _t}{t}, \particle{k^4 _t}{t}) \\
        & \hspace{3cm} \times M^{-2} \sum_{i,j \in [M]^2} \Pbacksum^b _t(J^i _{k^1 _t, t-1} , J^i _{k^2 _t, t-1}) \Pbacksum^b _t(J^j _{k^3 _t, t-1},J^j _{k^4 _t, t-1}) \bigg| \Gfilt{b}{t-1} \bigg] \\
        & =  \frac{(N-2)(N-3)}{N(N-1)} \parismuest{\BS}{b,t-1}\big( g^{\otimes 2} _{t-1} \bitransition{0}{t}[h] \big)^2\eqsp.
    \end{align*}    
    \vspace{0.3cm}\\
        $-$ If $(k^1 _t, \cdots, k^4 _t) \in \mathcal{I}^2 _0 \cap \sett^2 _2$. Then $\bm{k}_2 = \{k^1 _t, k^2 _t\}$ and we either have $V_{k^1 _t} = \{k^1 _t, k^3 _t\}$ and $V_{k^2 _t} = \{k^2 _t, k^4 _t\}$ or $V_{k^1 _t} = \{k^1 _t, k^4 _t\}$ and $V_{k^2 _t} = \{k^2 _t, k^3 _t\}$. Assume that $V_{k^1 _t} = \{k^1 _t, k^3 _t\}$ and $V_{k^2 _t} = \{k^2 _t, k^4 _t\}$. Taking into account that $J^i _{k^1 _t, t-1} = J^i _{k^3 _t, t-1}$ and $J^i _{k^2 _t,t-1} = J^i _{k^4 _t, t-1}$, we get
        \begin{align*}
              \pE \bigg[  \Pbacksum^b _t(k^1 _t, k^2 _t) \Pbacksum^b _t(k^3 _t, k^4 _t) \bigg| \Gfilt{b}{t-1} \bigg]&\\
             &\hspace{-4cm}= \pE \bigg[ M^{-2} \sum_{i,j \in [M]^2} \Pbacksum^b _{t-1}(J^i _{k^1 _t, t-1} , J^i _{k^2 _t, t-1}) \Pbacksum^b _{t-1}(J^j _{k^3 _t, t-1},J^j _{k^4 _t, t-1}) \big| \Gfilt{b}{t-1} \big] \\
             & \hspace{-4cm} = \pE \bigg[ M^{-2} \sum_{i,j \in [M]^2} \1_{i = j} \sum_{k^{1:2} _{t-1}\in [N]^2} \beta^\BS _t(k^1 _t, k^1 _{t-1}) \beta^\BS _t(k^2 _t, k^2 _{t-1})  \Pbacksum^b _{t-1}(k^1 _{t-1}, k^2 _{t-1})^2  \\
             & \hspace{-3.5cm} +  M^{-2} \sum_{i, j \in [M]^2} \1_{i \neq j} \sum_{k^{1:4} _{t-1} \in [N]^4} \prod_{\ell = 1}^4 \beta^\BS _t(k^\ell _t, k^\ell _{t-1})  \Pbacksum^b _{t-1}(k^1 _{t-1}, k^2 _{t-1}) \Pbacksum^b _{t-1}(k^3 _{t-1}, k^4 _{t-1}) \big| \Gfilt{b}{t-1} \bigg] \eqsp.
        \end{align*}
        Then,
             \begin{align*}
            & \pE \bigg[  \Pbacksum^b _t(k^1 _t, k^2 _t) \Pbacksum^b _t(k^3 _t, k^4 _t) \bigg| \Gfilt{b}{t-1} \bigg] \leq \frac{\boundg^2 }{M \Omega^2 _{t-1}} \sum_{k^{1:2} _{t-1} \in [N]^2} \Pbacksum^b _{t-1}(k^1 _{t-1}, k^2 _{t-1})^2 +  \frac{(M-1)\boundg^2}{M \Omega^2 _{t-1}} \\
             & \times \sum_{k^{1:4} _{t-1} \in [N]^4} \transition{t}[\beta^N _t(., \particle{k^3 _{t-1}}{t-1})](\particle{k^1 _{t-1}}{t-1}) \transition{t}[\beta^N _t(., \particle{k^4 _{t-1}}{t-1})](\particle{k^2 _{t-1}}{t-1}) \Pbacksum^b _{t-1}(k^1 _{t-1}, k^2 _{t-1}) \Pbacksum^b _{t-1}(k^3 _{t-1}, k^4 _{t-1}) \eqsp.
        \end{align*}
    Similarly, if $V_{k^1 _t} = \{k^1 _t, k^4 _t\}$ and $V_{k^2 _t} = \{k^2 _t, k^3 _t\}$, we obtain 
    \begin{align*}
        & \pE \bigg[  \Pbacksum^b _t(k^1 _t, k^2 _t) \Pbacksum^b _t(k^3 _t, k^4 _t) \bigg| \Gfilt{b}{t-1} \bigg] \\
         & \leq \frac{\boundg^2 }{M \Omega^2 _{t-1}} \sum_{k^{1:2} _{t-1} \in [N]^2} \Pbacksum^b _{t-1}(k^1 _{t-1}, k^2 _{t-1}) \Pbacksum^b _{t-1}(k^2 _{t-1}, k^1 _{t-1}) +  \frac{(M-1)\boundg^2}{M \Omega^2 _{t-1}} \\
         & \times \sum_{k^{1:4} _{t-1} \in [N]^4} \transition{t}[\beta^N _t(., \particle{k^4 _{t-1}}{t-1})](\particle{k^1 _{t-1}}{t-1}) \transition{t}[\beta^N _t(., \particle{k^3 _{t-1}}{t-1})](\particle{k^2 _{t-1}}{t-1}) \Pbacksum^b _{t-1}(k^1 _{t-1}, k^2 _{t-1}) \Pbacksum^b _{t-1}(k^3 _{t-1}, k^4 _{t-1}) \eqsp.
    \end{align*}
    Consequently, 
    \begin{align*}
        & \pE \bigg[ \sum_{k^{1:4} _t \in \mathcal{I}^2 _0 \cap \sett^2 _2} \Pbacksum^b _{t} (k^1 _t, k^2 _t) \Pbacksum^b _t (k^3 _t, k^4 _t) \bigg| \Gfilt{b}{t-1}\bigg] \\
        & \leq \frac{\boundg^2 N(N-1)}{M \Omega^2 _{t-1}} \sum_{k^{1:2} _{t-1} \in [N]^2}\left\{ \Pbacksum^b _{t-1}(k^1 _{t-1}, k^2 _{t-1})^2 + \Pbacksum^b _{t-1}(k^1 _{t-1}, k^2 _{t-1}) \Pbacksum^b _{t-1}(k^2 _{t-1}, k^1 _{t-1})\right\} \\
        &  + \frac{\boundg^2 N(N-1)(M-1)}{M \Omega^2 _{t-1}} \int \sum_{k^{1:4} _{t-1} \in [N]^4}  \bigg\{ \transitiondens{t}(\particle{k^1 _{t-1}}{t-1}, x) \beta^N _t(x, \particle{k^3 _{t-1}}{t-1}) \transitiondens{t}(\particle{k^2 _{t-1}}{t-1}, y) \beta^N _t(y, \particle{k^4 _{t-1}}{t-1}) + \\
        &  \transitiondens{t}(\particle{k^1 _{t-1}}{t-1}, x) \beta^N _t(x, \particle{k^4 _{t-1}}{t-1}) \transitiondens{t}(\particle{k^2 _{t-1}}{t-1}, y) \beta^N _t(y, \particle{k^3 _{t-1}}{t-1}) \bigg\} \Pbacksum^b _{t-1}(k^1 _{t-1}, k^2 _{t-1}) \Pbacksum^b _{t-1}(k^3 _{t-1}, k^4 _{t-1})  \nu^{\otimes 2}(\rmd x, \rmd y)\eqsp.
    \end{align*}
    Then, using \eqref{eq:ConstNidentity},
    \begin{align*}
        & \pE \bigg[ \Const{t}^2 \sum_{k^{1:4} _t \in \mathcal{I}^2 _0 \cap \sett^2 _2} \Pbacksum^b _{t} (k^1 _t, k^2 _t) \Pbacksum^b _t (k^3 _t, k^4 _t) \bigg| \Gfilt{b}{t-1}\bigg]  \\
        & \leq \frac{\boundg^2 \Omega^2 _{t-1}}{M N(N-1)} \Const{t-1}^2 \sum_{k^{1:2} _{t-1} \in [N]^2}\left\{ \Pbacksum^b _{t-1}(k^1 _{t-1}, k^2 _{t-1})^2 + \Pbacksum^b _{t-1}(k^1 _{t-1}, k^2 _{t-1}) \Pbacksum^b _{t-1}(k^2 _{t-1}, k^1 _{t-1})\right\} \\
        & +  \frac{\boundg^2 (M-1) \Omega^2 _{t-1}}{MN(N-1)}  \int \bigg\{ \parismuest{\BS}{b,t-1}\big( \transitiondens{t}(.,x) \otimes \transitiondens{t}(.,y) \big) \parismuest{\BS}{b,t-1}\big( \beta^N _t(x,.) \otimes \beta^N _t(y,.) \big) \\
        & \hspace{3cm} + \parismuest{\BS}{b,t-1}\big( \transitiondens{t}(.,x) \otimes \transitiondens{t}(.,y) \big) \parismuest{\BS}{b,t-1}\big( \beta^N _t(y,.) \otimes \beta^N _t(x,.) \big) \bigg\} \nu^{\otimes 2}(\rmd x, \rmd y)\eqsp, 
        \end{align*}
    which yields \eqref{eq:paris_conv_partitions_card2}.
    \vspace{.3cm}\\
    $-$ If $(k^1 _t, \cdots, k^4 _t) \in \mathcal{I}^2 _0 \cap \sett^3 _2$. Then either $\bm{k}_3 = \{k^1 _t, k^2 _t, k^3 _t\}$ or $\bm{k}_3 = \{k^1 _t, k^2 _t, k^4 _t\}$. Assume that $\bm{k}_3 = \{k^1 _t, k^2 _t, k^3 _t\}$. Then $V_{k^1 _t} = \{k^1 _t, k^4 _t\}$, $J^i _{k^1 _t, t-1} = J^i _{k^4 _t, t-1}$ for any $i \in [M]$ and
    \begin{align*}
    &\pE \bigg[  \Pbacksum^b _t(k^1 _t, k^2 _t) \Pbacksum^b _t(k^3 _t, k^4 _t) \bigg| \Gfilt{b}{t-1} \bigg] \\
        & =\pE \bigg[ M^{-2} \sum_{i,j \in [M]^2} \Pbacksum^b _{t-1}(J^i _{k^1 _t, t-1} , J^i _{k^2 _t, t-1}) \Pbacksum^b _{t-1}(J^j _{k^3 _t, t-1},J^j _{k^4 _t, t-1}) \big| \Gfilt{b}{t-1} \bigg] \\
        &  = \pE \bigg[ M^{-2} \sum_{i,j \in [M]^2} \1_{i = j} \sum_{k^{1:3} _{t-1}\in [N]^3} \prod_{i = 1}^3 \beta^\BS _t(k^i _t, k^i _{t-1}) \Pbacksum^b _{t-1}(k^1 _{t-1}, k^2 _{t-1}) \Pbacksum^b _{t-1}(k^3 _{t-1}, k^1 _{t-1}) \\
        & +  M^{-2} \sum_{i, j \in [M]^2} \1_{i \neq j} \sum_{k^{1:4} _{t-1} \in [N]^4} \prod_{\ell = 1}^4 \beta^\BS _t(k^\ell _t, k^\ell _{t-1})  \Pbacksum^b _{t-1}(k^1 _{t-1}, k^2 _{t-1}) \Pbacksum^b _{t-1}(k^3 _{t-1}, k^4 _{t-1}) \big| \Gfilt{b}{t-1} \bigg] \\
        & \leq \frac{\boundg^3}{M\Omega^3 _{t-1}} \sum_{k^{1:3} _{t-1} \in [N]^3}  \Pbacksum^b _{t-1}(k^1 _{t-1}, k^2 _{t-1}) \Pbacksum^b _{t-1}(k^3 _{t-1}, k^1 _{t-1}) \\
        & \hspace{1cm} + \frac{(M-1) \boundg^3}{M \Omega^3 _{t-1}} \sum_{k^{1:4} _{t-1} \in [N]^4} \transition{t}[\beta^N _t(., \particle{k^4 _{t-1}}{t-1})](\particle{k^1 _{t-1}}{t-1}) \Pbacksum^b _t(k^1 _{t-1}, k^2 _{t-1}) \Pbacksum^b _t(k^3 _{t-1}, k^4 _{t-1}) \eqsp.
    \end{align*}
    The remaining combinations are treated in the exact same way  and \eqref{eq:paris_conv_partitions_card3} is obtained by using again \eqref{eq:ConstNidentity}.
    \vspace{.3cm}\\
     Consider now the case $b_t = 1$ and let $(k^1 _t, \cdots, k^4 _t) \in \mathcal{I}^2 _1 \cap \sett^2 _2$. Then $k^1 _t = k^2 _t$, $k^3 _t = k^4 _t$ and $k^1 _t \neq k^3 _t$. Thus,
    \begin{align*}
        & \pE \bigg[ h(\particle{k^1 _t}{t}, \particle{k^2 _t}{t}) h(\particle{k^3 _t}{t}, \particle{k^4 _t}{t}) M^{-2} \sum_{i,j \in [M]^2} \sum_{k^{2,4} _{t-1} \in [N]^2} \normweight{k^2 _{t-1}}{t-1} \normweight{k^4 _{t-1}}{t-1} \\
        & \hspace{5cm} \times \Pbacksum^b _t(J^i _{k^1 _t, t-1} , k^2 _{t-1}) \Pbacksum^b _t(J^j _{k^3 _t, t-1},k^4 _{t-1}) \bigg| \Gfilt{b}{t-1} \bigg] \\
        & = \pE \bigg[ h(\particle{k^1 _t}{t}, \particle{k^2 _t}{t}) h(\particle{k^3 _t}{t}, \particle{k^4 _t}{t})  \sum_{k^{1:4} _{t-1} \in [N]^4} \prod_{\ell = 1}^2 \beta^\BS _t(k^{2\ell-1} _t, k^{2\ell - 1} _{t-1}) \normweight{k^{2\ell} _{t-1}}{t-1} \\
        & \hspace{5cm} \times \Pbacksum^b _{t-1}(k^1 _{t-1}, k^2 _{t-1}) \Pbacksum^b _{t-1}(k^3 _{t-1}, k^4 _{t-1})\bigg| \Gfilt{b}{t-1} \bigg] \\
        & = \frac{1}{\Omega^4 _{t-1}} \sum_{k^{1:4} _{t-1} \in [N]^4} \Pbacksum^b _{t-1}(k^1 _{t-1}, k^2 _{t-1}) \Pbacksum^b _{t-1}(k^3 _{t-1}, k^4 _{t-1}) \big( g^{\otimes 2} _{t-1} \bitransition{1}{t}[h] \big)^{\otimes 2} (\particle{k^1 _{t-1}}{t-1}, \particle{k^2 _{t-1}}{t-1}, \particle{k^3 _{t-1}}{t-1}, \particle{k^4 _{t-1}}{t-1}) \eqsp.
    \end{align*}
    Since $b_t = 1$ implies that
    \[ 
        \Const{t}^2 = \Const{t-1}^2 \frac{ \Omega^4 _{t-1}}{N^2} \eqsp,
    \]
    and using $\card{\mathcal{I}^1 _1 \cap \sett^2 _2} = N(N-1)$, we get
    \begin{align*}
        & \pE \bigg[ \Const{t}^2 \sum_{k^{1:4} _{t} \in \mathcal{I}^2 _1 \cap \sett^{2} _2} h(\particle{k^1 _t}{t}, \particle{k^2 _t}{t}) h(\particle{k^3 _t}{t}, \particle{k^4 _t}{t})\Pbacksum^b _t(k^1 _t, k^2 _t) \Pbacksum^b _t(k^3 _t, k^4 _t) \bigg| \Gfilt{b}{t-1} \bigg] \\
        & =  \Const{t}^2 \sum_{k^{1:4} _{t} \in \mathcal{I}^2 _1 \cap \sett^{2} _2} \pE \bigg[ h(\particle{k^1 _t}{t}, \particle{k^2 _t}{t}) h(\particle{k^3 _t}{t}, \particle{k^4 _t}{t}) M^{-2} \sum_{i,j \in [M]^2} \sum_{k^{2,4} _{t-1} \in [N]^2} \normweight{k^2 _{t-1}}{t-1} \normweight{k^4 _{t-1}}{t-1} \\
        & \hspace{5cm} \times \Pbacksum^b _t(J^i _{k^1 _t, t-1} , k^2 _{t-1}) \Pbacksum^b _t(J^j _{k^3 _t, t-1},k^4 _{t-1}) \bigg| \Gfilt{b}{t-1} \bigg] \\
        & =  \frac{N-1}{N} \parismuest{\BS}{b,t-1}\big( g^{\otimes 2} _{t-1} \bitransition{1}{t}[h] \big)^2 \eqsp,
    \end{align*} 
    which proves \eqref{eq:paris_conv_partitions_b1card2}. 
    \vspace{0.2cm}\\
    Let $(k^1_t, \cdots, k^4 _t) \in \mathcal{I}^2 _1 \cap \sett^1 _2$. Then $\bm{k}_1 = \{k^1 _t\}$, $k^1 _t = k^2 _t = k^3 _t = k^4 _t$ and $J^i _{k^1 _t, t-1} = J^i _{k^3 _t, t-1}$. Hence, 
    \begin{align*}
        & \pE \bigg[ M^{-2} \sum_{i,j \in [M]^2} \sum_{k^{2,4} _{t-1} \in [N]^2} \normweight{k^2 _{t-1}}{t-1} \normweight{k^4 _{t-1}}{t-1} \Pbacksum^b _{t-1}(J^i _{k^1 _t, t-1} , k^2 _{t-1}) \Pbacksum^b _{t-1}(J^j _{k^3 _t, t-1}, k^4 _{t-1}) \big| \Gfilt{b}{t-1} \bigg] \\
        & = \pE \bigg[\frac{1}{M} \sum_{k^{1,2,4} _{t-1}\in [N]^3} \beta^\BS _t(k^1 _t, k^1 _{t-1}) \normweight{k^2 _{t-1}}{t-1} \normweight{k^4 _{t-1}}{t-1} \Pbacksum^b _{t-1}(k^1 _{t-1}, k^2 _{t-1}) \Pbacksum^b _{t-1}(k^1 _{t-1}, k^4 _{t-1}) \\
        & +  \frac{M-1}{M} \sum_{k^{1:4} _{t-1} \in [N]^4} \prod_{\ell = 1}^2 \beta^\BS _t(k^{2\ell-1} _t, k^{2\ell - 1} _{t-1}) \normweight{k^{2\ell} _{t-1}}{t-1} \Pbacksum^b _{t-1}(k^1 _{t-1}, k^2 _{t-1}) \Pbacksum^b _{t-1}(k^3 _{t-1}, k^4 _{t-1}) \big| \Gfilt{b}{t-1} \bigg] \\
        & = \frac{\boundg^3}{M \Omega^3 _{t-1}} \sum_{k^{1,2,4} _{t-1} \in [N]^3} \Pbacksum^b _{t-1}(k^1 _{t-1}, k^2 _{t-1}) \Pbacksum^b _{t-1}(k^1 _{t-1}, k^4 _{t-1}) \\
        & \hspace{3cm} + \frac{\boundg^3 (M - 1)}{M\Omega^3 _{t-1}} \int \parismuest{M}{b,t-1}(m_t(.,x) \otimes \boldone) \parismuest{M}{b,t-1}(\beta^N _t(x,.) \otimes \boldone) \nu(\rmd x) \eqsp.
   \end{align*}
This yields
\begin{align*}
    & \pE \bigg[\Const{t}^2 \sum_{k^{1:4} _t \in \mathcal{I}^2 _1 \cap \sett^1 _2} \Pbacksum^b _{t}(k^1 _t, k^2 _t) \Pbacksum^b _{t}(k^3 _t, k^4 _t) \big| \Gfilt{b}{t-1} \bigg] \\
    & \hspace{1cm} \leq \frac{\boundg^3\Omega _{t-1}}{M N}\Const{t-1}^2 \sum_{k^{1,2,4} _{t-1} \in [N]^3} \Pbacksum^b _{t-1}(k^1 _{t-1}, k^2 _{t-1}) \Pbacksum^b _{t-1}(k^1 _{t-1}, k^4 _{t-1}) \\
    & \hspace{3cm} + \frac{\boundg^3 (M - 1) \Omega_{t-1}}{M N} \int \parismuest{M}{b,t-1}(m_t(.,x) \otimes \boldone) \parismuest{M}{b,t-1}(\beta^N _t(x,.) \otimes \boldone) \nu(\rmd x) \eqsp,
\end{align*}
which in turn proves \eqref{eq:paris_conv_partitions_b1card1}.
    \end{proof}
    \begin{proposition}
        \label{prop:conv_partitions}
        Let $\As{assp:B}{assp:boundup}$ hold. Let $t \geq 0$. For any $b \in \Bset_t$
        \begin{enumerate}[label=(\roman*)]
            \item \label{item:parispropb0}If $b_t = 0$, then for $p \in \{2,3\}$,
            \begin{equation}
                \simplelim \pE \bigg[ \Const{t}^2 \sum_{k^{1:4} _t \in \mathcal{I}^2 _0 \cap \sett^p _2} \Pbacksum^b _{t}(k^1 _t, k^2 _t) \Pbacksum^b _t(k^3 _t, k^4 _t) \bigg] = 0\eqsp.
            \end{equation}  
            \item \label{item:parispropb1} If $b_t = 1$, then 
            \begin{equation}
                \simplelim \pE \bigg[ \Const{t}^2 \sum_{k^{1:4} _t \in \mathcal{I}^2 _1 \cap \sett^1 _2} \Pbacksum^b _{t}(k^1 _t, k^2 _t) \Pbacksum^b _t(k^3 _t, k^4 _t) \bigg] = 0\eqsp.
            \end{equation}
        \end{enumerate}
        Additionally, if $\As{assp:boundbelow}{}$ holds then the rate of convergence is $\bigo(N^{-1})$.
    \end{proposition}
    \begin{proof}
    We prove \ref{item:parispropb0}-\ref{item:parispropb1} simultaneously by induction. Let $t = 0$ and $b \in \Bset_0$. By definition, $\Pbacksum^b _0 = \backsum^b _0$ and, 
    if $b_0 = 0$,
    \begin{align}
        \label{eq:convratet0_1}
        & \Const{0}^2 \sum_{k^{1:4} _0 \in \mathcal{I}^2 _0 \cap \sett^2 _2} \Pbacksum^b _{0}(k^1 _0, k^2 _0) \Pbacksum^b _0(k^3 _0, k^4 _0) = \frac{2}{N^2(N-1)^2} \sum_{i,j \in [N]^2} \1_{i \neq j} = \bigo(N^{-2}) \eqsp, \\
        \label{eq:convratet0_2}
        & \Const{0}^2 \sum_{k^{1:4} _0 \in \mathcal{I}^2 _0 \cap \sett^3 _2} \Pbacksum^b _{0}(k^1 _0, k^2 _0) \Pbacksum^b_0(k^3 _0, k^4 _0) = \frac{4}{N^2 (N-1)^2} \sum_{i,j,k \in [N]^3} \1_{i \neq j} \1_{i \neq k} = \bigo(N^{-1}) \eqsp.
    \end{align}
    If $b_0 = 1$, then 
    \begin{align}
        \label{eq:convratet0_3}
        \Const{0}^2 \sum_{k^{1:4} _0 \in \mathcal{I}^2 _1 \cap \sett^1 _2} \Pbacksum^b _{0}(k^1 _0, k^2 _0) \Pbacksum^b _t(k^3 _0, k^4 _0) = \frac{1}{N^2} \sum_{i, j \in [N]^2} \1_{i = j} = \bigo(N^{-1}) \eqsp.
    \end{align}
    Let $t > 0$ and assume that both \ref{item:parispropb0}-\ref{item:parispropb1} hold at $t-1$. Define
    \begin{align*}
        D^N _{1,b} & = \pE \bigg[ \Const{t-1}^2 \sum_{k^{1,2,4} _{t-1} \in [N]^3} \Pbacksum^b _{t-1}(k^1 _{t-1}, k^2 _{t-1}) \Pbacksum^b _{t-1}(k^1 _{t-1}, k^4 _{t-1}) \bigg]\eqsp, \\
        D^N _{2,b} & = \pE \bigg[ \Const{t-1}^2 \sum_{k^{1:2} _{t-1} \in [N]^2}  \mathsf{T}^{(1)}_{t-1,b}(k^1 _{t-1},k^2 _{t-1}) \bigg] \eqsp,\\
        D^N _{3,b} & = \pE \bigg[ \Const{t-1}^2 \sum_{k^{1:3} _{t-1} \in [N]^3} \mathsf{T}^{(2)}_{t-1,b}(k^1 _{t-1},k^2 _{t-1},k^3 _{t-1}) \bigg] \eqsp, 
    \end{align*}
    where $\mathsf{T}^{(1)}_{t-1,b}$ and $\mathsf{T}^{(2)}_{t-1,b}$ are defined in Proposition~\ref{prop:paris_conv_partitions}. 
    If $b_t = 0$, then by \ref{item:paris_convb0} in Proposition~\ref{prop:paris_conv_partitions}, using that $\Omega_{t-1}\leq N\boundg$,
    \begin{multline}
        \label{eq:S22_ineq}
         \pE \bigg[ \Const{t} ^2 \sum_{k^{1:4} _{t} \in \mathcal{I}^2 _0 \cap \sett^{2} _2} \Pbacksum^b _t(k^1 _t, k^2 _t) \Pbacksum^b _t(k^3 _t, k^4 _t) \bigg] \\
          \leq \pE \bigg[ \frac{\Omega^2 _{t-1} \boundg^2 (M-1)}{MN(N-1)} \int   \widetilde{\Upsilon}_{b,t}^{N,M}(x,y) \nu^{\otimes 2}(\rmd x, \rmd y ) \bigg] + 2\frac{ \boundg^4}{M} D_{2,b} ^N  \eqsp,
    \end{multline}
    and 
    \begin{multline}
        \label{eq:S32_ineq}
         \pE \bigg[ \Const{t}^2 \sum_{k^{1:4} _{t} \in \mathcal{I}^2 _0 \cap \sett^{3} _2} \Pbacksum^b _t(k^1 _t, k^2 _t) \Pbacksum^b _t(k^3 _t, k^4 _t) \bigg] \\
         \leq \pE \bigg[ \frac{\Omega _{t-1} \boundg^3 (M-1)(N-2)}{MN(N-1)} \int  \widetilde{\Theta}_{b,t}^{N,M}(x) \nu(\rmd x) \bigg]  + \frac{ \boundg^4}{M} D^N _{3,b}  \eqsp,
    \end{multline}
    where $\widetilde{\Theta}_{b,t}^{N,M}$ and $\widetilde{\Upsilon}_{b,t}^{N,M}$ are defined in Proposition~\ref{prop:paris_conv_partitions}. 
    We deal first with $D^N _{2,b}$ and $D^N _{3,b}$. If $b_{t-1} = 0$, then by definition of $\Pbacksum^b _{t-1}$ in \eqref{eqdef:tautilde},
    \begin{equation}
    \begin{aligned}
        \label{eq:DNb0}
        D^N _{2,b} & = \pE \bigg[ \Const{t-1}^2 \sum_{k^{1:4} _{t-1} \in \mathcal{I}^2 _0 \cap \sett^2 _2} \Pbacksum^b _{t-1}(k^1 _{t-1}, k^2 _{t-1}) \Pbacksum^b _{t-1}(k^3 _{t-1}, k^4 _{t-1}) \bigg] \eqsp, \\
        D^N _{3,b} & = \pE \bigg[ \Const{t-1}^2 \sum_{k^{1:4} _{t-1} \in \mathcal{I}^2 _0 \cap \sett^3 _2} \Pbacksum^b _{t-1}(k^1 _{t-1}, k^2 _{t-1}) \Pbacksum^b _{t-1}(k^3 _{t-1}, k^4 _{t-1}) \bigg] \eqsp.
    \end{aligned}
\end{equation}
    If $b_{t-1} = 1$, 
    \begin{equation}
        \label{eq:DNb1}
    \begin{aligned}
         D^N _{2,b} & = 2 \pE \bigg[ \Const{t-1}^2 \sum_{k^{1:4} _{t-1} \in \mathcal{I}^2 _1 \cap \sett^1 _2} \Pbacksum^b _{t-1}(k^1 _{t-1}, k^2 _{t-1}) \Pbacksum^b _{t-1}(k^3 _{t-1}, k^4 _{t-1}) \bigg]\eqsp, \\
         D^N _{3,b} & = 4 \pE \bigg[ \Const{t-1}^2 \sum_{k^{1:4} _{t-1} \in \mathcal{I}^2 _1 \cap \sett^1 _2} \Pbacksum^b _{t-1}(k^1 _{t-1}, k^2 _{t-1}) \Pbacksum^b _{t-1}(k^3 _{t-1}, k^4 _{t-1}) \bigg]\eqsp.
    \end{aligned}
\end{equation}
    In all cases, by the induction hypothesis we get for any $b \in \Bset_t$ with $b_t = 0$,
    \begin{equation*} 
        \label{eq:convTausquare}
         \simplelim D^N _{2,b} = 0, \quad  \simplelim D^N _{3,b} = 0  \eqsp.
    \end{equation*}
    Regarding the first terms in the r.h.s. of inequalities \eqref{eq:S22_ineq}-\eqref{eq:S32_ineq}, they go to zero when $N$ goes to infinity since they are, up to the constant $(M-1)/M \leq 1$, the $\paris$ counterpart of $B_N$ \eqref{eq:defBN} in the proof of Theorem~\ref*{thm:conv} and are treated in the exact same way since $\sup_{N \in \N} \pE \big[ \parismuest{\BS}{b,t}(\boldone)^3 \big] < \infty$ by Proposition~\ref{prop:parisQbound}.
    If $b_t = 1$, then, by Proposition~\ref{prop:paris_conv_partitions}, using that $\Omega_{t-1} \leq N\boundg$ , 
    \begin{align*}
        & \pE \bigg[ \Const{t}^2 \sum_{k^{1:4} _{t} \in \mathcal{I}^2 _1 \cap \sett^{1} _2} \Pbacksum^b _t(k^1 _t, k^2 _t) \Pbacksum^b _t(k^3 _t, k^4 _t) \bigg] \\
                & \hspace{2cm} \leq \pE \bigg[ \frac{\Omega_{t-1} \boundg^3}{N} \int \parismuest{\BS}{b,t-1}(\transitiondens{t}.,x) \otimes \boldone) \parismuest{\BS}{b,t-1}(\beta^N _t(x,.) \otimes 1) \nu(\rmd x) \bigg] + \frac{ \boundg^4}{M} D^N _{1,b} \eqsp.
    \end{align*}
    The first term goes to zero when $N$ goes to infinity similarly to the case $b_t=0$. As for the second term, if $b_{t-1} = 0$ then by definition of $D^N _{1,b}$ 
    \begin{equation*}
        0 \leq D^N _{1,b} \leq D^N _{3,b} \eqsp,
    \end{equation*}
    and if $b_{t-1} = 1$, 
    \begin{equation*}
        0 \leq D^N _{1,b} = \pE \bigg[ \Const{t}^2 \sum_{k^{1:2} _{t-1} \in [N]^2} \Pbacksum^b _{t-1}(k^1 _{t-1}, k^2 _{t-1})^2 \bigg] \leq D^N _{2,b} \eqsp.
    \end{equation*}
    Hence, in both cases $D^N _{1,b}$ goes to zero by \eqref{eq:convTausquare}. This ends the proof of the first claim. 
    
    Now assume that $\As{assp:boundbelow}{}$ holds also. We proceed by induction. At $t = 0$ the rate of convergence is $\bigo(N^{-1})$ by \eqref{eq:convratet0_1}-\eqref{eq:convratet0_2} and \eqref{eq:convratet0_3}. Let $t > 0$ and assume that the rate of convergence in \ref{item:paris_convb0} and \ref{item:paris_convb1} at $t-1$ is $\bigo(N^{-1})$. Assume that $b_t = 0$. By the strong mixing assumption we have that
 \begin{equation}
    \label{eq:betaNstrongmixbound}
    \beta^N _t(x,y) \leq \frac{\boundg \sigma_{+}}{\sigma_{-} \Omega _{t-1}}\eqsp, \quad \forall (x,y) \in \Xset^2 \eqsp.
\end{equation}
Using for example that \[ \int \parismuest{\BS}{b,t}(m_t(.,x) \otimes \boldone) \nu(\rmd x) = \int \parismuest{\BS}{b,t}(m_t(.,x) \otimes m_t(.,y)) \nu^{\otimes 2}(\rmd x, \rmd y) = \parismuest{\BS}{b,t}(\boldone \otimes \boldone) \eqsp, \]
and then bounding $\beta^N _t$ using \eqref{eq:betaNstrongmixbound} we get that
\begin{equation}
     \pE \bigg[ \frac{\Omega^2 _{t-1} \boundg^2 (M-1)}{MN(N-1)} \int \widetilde{\Upsilon}^{N,M} _{b,t} (x,y) \nu^{\otimes 2}(\rmd x, \rmd y ) \bigg]
     \leq \frac{\boundg^4 \sigma_{+}^2 (M-1)}{\sigma_{-}^2 MN(N-1)} \big\| \parismuest{\BS}{b,t-1}(\boldone \otimes \boldone) \big\|^2 _2 \eqsp,
\end{equation}
and
$$
     \pE \bigg[ \frac{\Omega _{t-1} \boundg^3 (M-1)(N-2)}{MN(N-1)} \int  \widetilde{\Theta}_{b,t}^{N,M}(x) \nu(\rmd x) \bigg] 
        \leq \frac{4\sigma_{+} \boundg^4 (M-1)(N-2)}{\sigma_{-} MN(N-1)} \big\| \parismuest{\BS}{b,t-1}(\boldone \otimes \boldone) \big\|_2 ^2 \eqsp,
$$
where both bounds are $\bigo(N^{-1})$ by Proposition~\ref{prop:parisQbound}. Going back to \eqref{eq:S22_ineq}-\eqref{eq:S32_ineq}, we obtain
$$
     \pE \bigg[ \Const{t} ^2 \sum_{k^{1:4} _{t} \in \mathcal{I}^2 _0 \cap \sett^{2} _2} \Pbacksum^b _t(k^1 _t, k^2 _t) \Pbacksum^b _t(k^3 _t, k^4 _t) \bigg] 
      \leq \frac{2\boundg^4 \sigma_{+}^2 (M-1)}{\sigma_{-}^2 MN(N-1)} \big\| \parismuest{\BS}{b,t-1}(\boldone \otimes \boldone) \big\|^2 _2
     + 2\frac{\boundg^4}{M} D^N _{2,b} \eqsp, 
$$
and 
\begin{multline*}
     \pE \bigg[ \Const{t}^2 \sum_{k^{1:4} _{t} \in \mathcal{I}^2 _0 \cap \sett^{3} _2} \Pbacksum^b _t(k^1 _t, k^2 _t) \Pbacksum^b _t(k^3 _t, k^4 _t) \bigg] \\
     \leq \frac{4\sigma_{+} \boundg^4 (M-1)(N-2)}{\sigma_{-} MN(N-1)} \big\| \parismuest{\BS}{b,t-1}(\boldone \otimes \boldone) \big\|_2 ^2 + \frac{\boundg^4}{M} D^N _{3,b} \eqsp. \nonumber
\end{multline*}
By \eqref{eq:DNb0}-\eqref{eq:DNb1} and the induction hypothesis, we get that $D^N _{2,b}$ and $D^N _{3,b}$ are both $\bigo(N^{-1})$. Finally, applying Proposition~\ref{prop:parisQbound} we get $\bigo(N^{-1})$ upper bounds. This ends the proof for the case $b_t = 0$ . The case $b_t = 1$ follows the same steps. 
    \end{proof}
    \begin{proposition}
        \label{prop:parisQbound}
        Assume that $\As{assp:B}{}$ holds. For all $M > 1$, $t \in \N$ and $b \in \Bset_t$,
        \begin{equation}
            \label{eq:parisQbound}
            \sup_{N \in \N} \big\| \parismuest{M}{b,t}(\boldone) \big\|_3 < \infty\eqsp.
        \end{equation}
    \end{proposition}
    \begin{proof} We proceed by induction on $t \in \N$. Assume for now that $N \geq 6$ and $M \geq 2$.
        
        Since $\parismuest{\BS}{b,0}(h) = \muest{\BS}{b,0}(h)$ for any $h$, the case $t = 0$ follows from Proposition~\ref{prop:Qbound}. Let $t > 0$ and assume that \eqref{eq:parisQbound} holds at $t-1$. We only treat the case $b_t = 0$. The proof for the case $b_t = 1$ follows using the same steps. Since we have that 
    \begin{equation}
    \begin{aligned}
        \big\| \parismuest{\BS}{b,t}(\boldone) \big\|^3 _3
        & = \pE \bigg[ \Const{t}^3 \sum_{k^{1:6} _t \in [N]^6}\prod_{\ell = 1}^3 \Pbacksum^b _t(k^{2\ell - 1} _t, k^{2\ell} _t)\bigg] \\
        & = \pE \bigg[ \frac{\Const{t}^3}{M^3} \sum_{k^{1:6} _t \in \mathcal{I}^3 _0} \sum_{i^{1:3} \in [M]^3} \prod_{\ell = 1}^3 \Pbacksum^b _{t-1}\big( J^{i^\ell} _{k^{2\ell - 1} _t, t-1}, J^{i^\ell} _{k^{2\ell} _t, t-1}\big) \bigg] \eqsp,
    \end{aligned}
\end{equation}
    the proof proceeds by (i) splitting the sum in three parts with respect to the  cardinal of the triplet $(i^1, i^2 ,i^3) \in [M]^3$, and (ii) bounding each term by $\| \parismuest{\BS}{b,t-1}(\boldone)\|^3 _3$ up to a constant independent of $N$.
        Let $(i^1, i^2, i^3) \in [M]^3$. If $\card{\{i^1, i^2, i^3\}} = 3$, then for all $k^{1:6} _t \in [N]^6$, $( J^{i^\ell} _{k^{2\ell - 1} _t, t-1}, J^{i^\ell} _{k^{2\ell} _t, t-1})_{\ell\in [1:3]}$ are mutually independent conditionally on $\Gfilt{b}{t-1} \vee \particle{1:N}{t}$, defined in \eqref{eq:Gfiltparticle}, and
        \begin{align*}
             \pE \bigg[ \prod_{\ell = 1}^3 & \Pbacksum^b _{t-1}\big( J^{i^\ell} _{k^{2\ell - 1} _t, t-1}, J^{i^\ell} _{k^{2\ell} _t, t-1}\big)  \bigg| \Gfilt{b}{t-1} \bigg] \\
            & = \pE \bigg[ \prod_{\ell = 1}^3 \pE \bigg[ \Pbacksum^b _{t-1}\big( J^{i^\ell} _{k^{2\ell - 1} _t, t-1}, J^{i^\ell} _{k^{2\ell} _t, t-1}\big) \bigg| \Gfilt{b}{t-1} \vee \particle{1:N}{t} \bigg] \bigg| \Gfilt{b}{t-1} \bigg] \\
            & = \pE \bigg[ \prod_{\ell = 1}^3 \sum_{k^{2\ell-1 : 2\ell} _{t-1} \in [N]^2} \beta^\BS _t(k^{2\ell-1} _t, k^{2\ell-1} _{t-1}) \beta^\BS _t(k^{2\ell} _t, k^{2\ell} _{t-1}) \Pbacksum^b _{t-1}\big(k^{2\ell - 1} _{t-1}, k^{2\ell} _{t-1}\big) \bigg| \Gfilt{b}{t-1} \bigg] \\
            & = \sum_{k^{1:6} _{t-1} \in [N]^6} \prod_{\ell = 1}^3 \Pbacksum^b _{t-1}\big(k^{2\ell - 1} _{t-1}, k^{2\ell} _{t-1}\big) \pE \bigg[ \prod_{n = 1}^6 \beta^\BS _t(k^n _t, k^n _{t-1}) \bigg| \F{t-1} \bigg] \eqsp. 
        \end{align*}
    Hence, by Proposition~\ref{lem:prodbeta}, Lemma~\ref{lem:cardinals} and similarly to the proof of Proposition~\ref{prop:Qbound},
    \begin{align*}
         \sum_{k^{1:6} _t \in \mathcal{I}^3 _0} & \pE \bigg[ \prod_{\ell = 1}^3 \Pbacksum^b _{t-1}\big( J^{i^\ell} _{k^{2\ell - 1} _t, t-1}, J^{i^\ell} _{k^{2\ell} _t, t-1}\big)  \bigg| \Gfilt{b}{t-1} \bigg] \\
        & = \sum_{k^{1:6} _{t-1} \in [N]^6} \prod_{\ell = 1}^3 \Pbacksum^b _{t-1}\big(k^{2\ell - 1} _{t-1}, k^{2\ell} _{t-1}\big) \sum_{k^{1:6} _t \in \mathcal{I}^3 _0} \pE \bigg[ \prod_{n = 1}^6 \beta^\BS _t(k^n _t, k^n _{t-1}) \bigg| \F{t-1} \bigg] \\
        & \leq \sum_{k^{1:6} _{t-1} \in [N]^6} \prod_{\ell = 1}^3 \Pbacksum^b _{t-1}\big(k^{2\ell - 1} _{t-1}, k^{2\ell} _{t-1}\big) \sum_{p = 2}^{6} \sum_{k^{1:6} _t \in \mathcal{I}^3 _0 \cap \sett^p _3} \frac{\boundg^p}{\Omega^p _{t-1}} \\
        & \lesssim \bigg( \sum_{p = 2}^{6} \frac{N^p }{\Omega^p _{t-1}} \bigg) \sum_{k^{1:6} _{t-1} \in [N]^6} \prod_{\ell = 1}^3 \Pbacksum^b _{t-1}\big(k^{2\ell - 1} _{t-1}, k^{2\ell} _{t-1}\big) \eqsp,
    \end{align*}
    where $\lesssim$ means less than or equal up to a multiplicative constant independent of $N$. Consequently,
    \begin{align*}
         \pE \bigg[ \frac{\Const{t}^3}{M^3} &\sum_{k^{1:6} _t \in \mathcal{I}^3 _0} \sum_{i^{1:3} \in [M]^3} \1_{i^1 \neq i^2 \neq i^3} \prod_{\ell = 1}^3 \Pbacksum^b _{t-1}\big( J^{i^\ell} _{k^{2\ell - 1} _t, t-1}, J^{i^\ell} _{k^{2\ell} _t, t-1}\big) \bigg| \Gfilt{b}{t-1} \bigg] \\
        &\hspace{1cm} = \frac{\Const{t}^3}{M^3}  \sum_{k^{1:6} _t \in \mathcal{I}^3 _0} \sum_{i^{1:3} \in [M]^3} \1_{i^1 \neq i^2 \neq i^3} \pE \bigg[ \prod_{\ell = 1}^3 \Pbacksum^b _{t-1}\big( J^{i^\ell} _{k^{2\ell - 1} _t, t-1}, J^{i^\ell} _{k^{2\ell} _t, t-1}\big) \bigg| \Gfilt{b}{t-1} \bigg] \\
        & \hspace{1cm}\lesssim \frac{M(M-1)(M-2)}{M^3} \bigg( \sum_{p = 2}^{6} \frac{N^p }{\Omega^p _{t-1}} \bigg) \Const{t}^3  \sum_{k^{1:6} _{t-1} \in [N]^6} \prod_{\ell = 1}^3 \Pbacksum^b _{t-1}\big(k^{2\ell - 1} _{t-1}, k^{2\ell} _{t-1}\big) \eqsp,
    \end{align*} 
    and
    \begin{multline*}
        \Const{t}^3 \bigg(\sum_{p = 2}^{6} \frac{N^p }{\Omega^p _{t-1}} \bigg)  = \Const{t-1}^3 \bigg( \sum_{p = 2}^{6} \frac{N^p \Omega^6 _{t-1}}{\Omega^p _{t-1} N^3 (N-1)^3} \bigg)\\
         \leq \Const{t-1}^3 \bigg( \sum_{p = 2}^6 \frac{\boundg^{6-p} N^6}{N^3 (N-1)^3}\bigg) \lesssim \Const{t-1}^3 \eqsp.
    \end{multline*}
    Therefore, 
    \begin{equation}
        \label{eq:Q3card3}
         \pE \bigg[ \frac{\Const{t}^3}{M^3} \sum_{k^{1:6} _t \in [N]^6} \sum_{i^{1:3} \in [M]^3} \1_{i^1 \neq i^2 \neq i^3} \prod_{\ell = 1}^3 \Pbacksum^b _{t-1}\big( J^{i^\ell} _{k^{2\ell - 1} _t, t-1}, J^{i^\ell} _{k^{2\ell} _t, t-1}\big) \bigg]
         \lesssim \big\| \parismuest{\BS}{b,t-1}(\boldone)\big\|_3 ^3 \eqsp.
    \end{equation}
    Assume now that $\card{\{i^1, i^2, i^3\}} = 2$ and that $i^1 = i^2$. For all $(k^1 _t, \cdots, k^6 _t) \in [N]^6$, conditionally on $\Gfilt{b}{t-1} \vee \particle{1:N}{t}$, $\Pbacksum^b _{t-1}\big( J^{i^3} _{k^{5} _t, t-1}, J^{i^3} _{k^{6} _t, t-1} \big)$ is independent from $\Pbacksum^b _{t-1}\big( J^{i^1} _{k^{1} _t, t-1}, J^{i^1} _{k^{2} _t, t-1} \big)$, $ \Pbacksum^b _{t-1}\big( J^{i^2} _{k^{3} _t, t-1}, J^{i^2} _{k^{4} _t, t-1} \big)$, hence, using \eqref{def:Iset}, 
    \begin{align}
        \sum_{k^{1:6} _t \in \mathcal{I}^3 _0} & \pE \bigg[ \prod_{\ell = 1}^3 \Pbacksum^b _{t-1}\big( J^{i^\ell} _{k^{2\ell - 1} _t, t-1}, J^{i^\ell} _{k^{2\ell} _t, t-1}\big)  \bigg| \Gfilt{b}{t-1} \vee \particle{1:N}{t} \bigg] \\
       & = \bigg\{ \sum_{k^{1:4} _t \in \mathcal{I}^2 _0} \pE \bigg[  \prod_{\ell = 1}^2 \Pbacksum^b _{t-1}\big( J^{i^\ell} _{k^{2\ell -1} _t, t-1}, J^{i^\ell} _{k^{2\ell} _t, t-1} \big) \bigg| \Gfilt{b}{t-1} \vee \particle{1:N}{t} \bigg] \bigg\} \nonumber \\
       & \hspace{2cm} \times \bigg\{ \sum_{k^{5:6} _t \in \mathcal{I}^1 _0}  \pE \bigg[ \Pbacksum^b _{t-1}\big( J^{i^3} _{k^{5} _t, t-1}, J^{i^3} _{k^{6} _t, t-1} \big) \bigg| \Gfilt{b}{t-1} \vee \particle{1:N}{t} \bigg] \bigg\} \nonumber \\
       & = \big( F^N _{2,b} + F^N _{3,b} + F^N _{4,b} \big) F^N _{1,b} \eqsp, \nonumber
    \end{align}
with
\begin{align*}
     F^N _{1,b} & \eqdef \sum_{k^{5:6} _t \in \mathcal{I}^1 _0}  \pE \bigg[ \Pbacksum^b _{t-1}\big( J^{i^3} _{k^{5} _t, t-1}, J^{i^3} _{k^{6} _t, t-1} \big) \bigg| \Gfilt{b}{t-1} \vee \particle{1:N}{t} \bigg] \\
    & = \sum_{k^{5:6} _t \in [N]^2} \1_{k^5 _t \neq k^6 _t} \sum_{k^{5:6} _{t-1} \in [N]^2} \beta^\BS _t(k^5 _t, k^5 _{t-1}) \beta^\BS _t(k^6 _t, k^6 _{t-1}) \Pbacksum^b _{t-1}(k^5 _{t-1}, k^6 _{t-1})  \eqsp,
\end{align*}
and
\begin{align*}
    & \sum_{k^{1:4} _t \in \mathcal{I}^2 _0} \pE \bigg[  \prod_{\ell = 1}^2 \Pbacksum^b _{t-1}\big( J^{i^\ell} _{k^{2\ell -1} _t, t-1}, J^{i^\ell} _{k^{2\ell} _t, t-1} \big) \bigg| \Gfilt{b}{t-1} \vee \particle{1:N}{t} \bigg] \\
     & \hspace{2cm} = \sum_{p = 2}^4 \sum_{k^{1:4} _t \in \mathcal{I}^2 _0 \cap \sett^p _2} \pE \bigg[  \prod_{\ell = 1}^2 \Pbacksum^b _{t-1}\big( J^{i^\ell} _{k^{2\ell -1} _t, t-1}, J^{i^\ell} _{k^{2\ell} _t, t-1} \big) \bigg| \Gfilt{b}{t-1} \vee \particle{1:N}{t} \bigg] \\
     & \hspace{2cm} = F^N _{2,b} + F^N _{3,b} + F^N _{4,b} \eqsp,
\end{align*}
where
\begin{align*}
    & F^N _{2,b} \eqdef \sum_{k^{1:2} _t \in [N]^2} \1_{k^2 _t \neq k^2 _t} \sum_{k^{1:2} _{t-1} \in [N]^2}  \beta^\BS _t(k^1 _t, k^1 _{t-1}) \beta^\BS _t(k^2 _t, k^2 _{t-1})\mathsf{T}^{(1)}_{t-1,b}(k^1 _{t-1}, k^2 _{t-1})\\
    & F^N _{3,b} \eqdef \sum_{k^{1:3} _t \in [N]^3} \1_{k^1 _t \neq k^2 _t \neq k^3 _t}
    \sum_{k^{1:3} _{t-1} \in [N]^3} \prod_{\ell = 1}^3 \beta^\BS _t(k^\ell _t, k^\ell _{t-1})\mathsf{T}^{(2)}_{t-1,b}(k^1 _{t-1}, k^2 _{t-1},k^3 _{t-1})\\
    & F^N _{4,b} \eqdef \sum_{k^{1:4} _t \in [N]^4} \1_{k^1 _t \neq k^2 _t \neq k^3 _t \neq k^4 _t}  \sum_{k^{1:4} _{t-1} \in [N]^4} \prod_{\ell = 1}^4 \beta^\BS _t(k^\ell _t, k^\ell _{t-1}) \Pbacksum^b _{t-1}(k^1 _{t-1}, k^2 _{t-1}) \Pbacksum^b _{t-1}(k^3 _{t-1}, k^4 _{t-1}) \eqsp,
\end{align*}
where $\mathsf{T}^{(1)}_{t-1,b}$ and $\mathsf{T}^{(2)}_{t-1,b}$ are defined in Proposition~\ref{prop:paris_conv_partitions}. 
    We now upperbound each $\pE[ F^N _{1,b} F^N _{i,b}| \Gfilt{b}{t-1}]$ for $i \in [2:4]$. 
    \vspace{.2cm}\\
     Consider first the case $\pE[ F^N _{1,b} F^N _{4,b}| \Gfilt{b}{t-1}]$.   Let $S_6 \eqdef \{k^{1:6} \in [N]^6: k^1 _t \neq k^2 _t \neq k^3 _t \neq k^4 _t, k^5 _t \neq k^6 _t \} \eqsp.$ Then, $S_6 \subset \big( \mathcal{I}^3 _0 \cap \sett^4 _3 \big) \sqcup \big( \mathcal{I}^3 _0 \cap \sett^5 _3 \big) \sqcup \big( \mathcal{I}^3 _0 \cap \sett^6 _3 \big)$ and 
    \begin{align}
        \label{eq:boundF1F4}
         \pE \big[ F^N _{1,b} F^N _{4,b} \big| \Gfilt{b}{t-1} \big] \nonumber & = \sum_{k^{1:6} _{t-1} \in [N]^6} \prod_{\ell = 1}^3 \Pbacksum^b _{t-1}(k^{2\ell - 1} _{t-1}, k^{2\ell} _{t-1})  \sum_{k^{1:6} _t \in \mathsf{S}_6}  \pE \bigg[ \prod_{n = 1}^6 \beta^\BS _t(k^n _t, k^n _{t-1}) \bigg| \Gfilt{b}{t-1}\bigg] \nonumber \\
        & \leq \sum_{k^{1:6} _{t-1} \in [N]^6} \prod_{\ell = 1}^3 \Pbacksum^b _{t-1}(k^{2\ell - 1} _{t-1}, k^{2\ell} _{t-1}) \sum_{p = 4}^6 \sum_{k^{1:6} _t \in \mathcal{I}^3 _0 \cap \sett^p _3} \pE \bigg[ \prod_{n = 1}^6 \beta^\BS _t(k^n _t, k^n _{t-1}) \bigg| \Gfilt{b}{t-1}\bigg] \nonumber \\
        & \lesssim \bigg( \sum_{p = 4}^6 \frac{N^p \boundg^p}{\Omega^p _{t-1}} \bigg)\sum_{k^{1:6} _{t-1} \in [N]^6} \prod_{\ell = 1}^3 \Pbacksum^b _{t-1}(k^{2\ell - 1} _{t-1}, k^{2\ell} _{t-1}) \eqsp,
    \end{align}
    by Proposition~\ref{lem:prodbeta} and Lemma~\ref{lem:cardinals}.
    \vspace{.2cm}\\
   Consider now the case $\pE[ F^N _{1,b} F^N _{3,b}| \Gfilt{b}{t-1}]$ and define 
    $\mathsf{S}_5 \eqdef \{k^{1:5} _t \in [N]^5: k^1 _t \neq k^2 _t \neq k^3 _t, k^4 _t \neq k^5 _t \}$. Then,
    \begin{multline*}
        \pE \big[ F^N _{1,b} F^N _{3,b} \big| \Gfilt{b}{t-1} \big]  \\
        = \sum_{k^{1:5} _{t-1} \in [N]^5}  \mathsf{T}^{(2)}_{t-1,b}(k^1 _{t-1}, k^2 _{t-1}, k^3 _{t-1}) \Pbacksum^b _{t-1}(k^4 _{t-1}, k^5 _{t-1}) \sum_{k^{1:5} _t \in \mathsf{S}_5}  \pE \bigg[ \prod_{n = 1}^5 \beta^\BS _t(k^n _t, k^n _{t-1}) \bigg| \Gfilt{b}{t-1}\bigg] \eqsp.
    \end{multline*}
    Proceeding similarly as in Proposition~\ref{lem:prodbeta} and Lemma~\ref{lem:cardinals}, it can be shown that 
    \begin{align*}
        \sum_{k^{1:5} _t \in \mathsf{S}_5}  \pE \bigg[ \prod_{n = 1}^5 \beta^\BS _t(k^n _t, k^n _{t-1}) \bigg| \Gfilt{b}{t-1}\bigg] \leq \frac{N^5 \boundg^5 }{\Omega^5 _{t-1}} + \frac{2N^4\boundg^4}{\Omega^4 _{t-1}} + \frac{N^3\boundg^3}{\Omega^3 _{t-1}} \eqsp.
    \end{align*}
    Finally, by noting that
    \begin{align*}
        \sum_{k^{1:5} _{t-1} \in [N]^5} \mathsf{T}^{(2)}_{t-1,b}(k^1 _{t-1}, k^2 _{t-1}, k^3 _{t-1}) \Pbacksum^b _{t-1}(k^4 _{t-1}, k^5 _{t-1}) \leq 4 \sum_{k^{1:6} _{t-1} \in [N]^6} \prod_{\ell = 1}^3 \Pbacksum^b _{t-1}(k^{2\ell - 1} _{t-1}, k^{2\ell} _{t-1}) 
    \end{align*}
    we get
    \begin{align}
        \label{eq:boundF1F3}
        \pE \big[ & F^N _{1,b} F^N _{3,b} \big| \Gfilt{b}{t-1} \big] \nonumber \\
        & \leq \bigg( \frac{N^5 \boundg^5 }{\Omega^5 _{t-1}} + \frac{2N^4\boundg^4}{\Omega^4 _{t-1}} + \frac{N^3\boundg^3}{\Omega^3 _{t-1}} \bigg) 
        \sum_{k^{1:5} _{t-1} \in [N]^5}  \mathsf{T}^{(2)}_{t-1,b}(k^1 _{t-1}, k^2 _{t-1}, k^3 _{t-1}) \Pbacksum^b _{t-1}(k^4 _{t-1}, k^5 _{t-1}) \nonumber \\
        & \leq 4 \bigg( \frac{N^5 \boundg^5 }{\Omega^5 _{t-1}} + \frac{2N^4\boundg^4}{\Omega^4 _{t-1}} + \frac{N^3\boundg^3}{\Omega^3 _{t-1}} \bigg) \sum_{k^{1:6} _{t-1} \in [N]^6} \prod_{\ell = 1}^3 \Pbacksum^b _{t-1}(k^{2\ell - 1} _{t-1}, k^{2\ell} _{t-1}) \eqsp.
    \end{align}
    Consider now the case $\pE[ F^N _{1,b} F^N _{2,b}| \Gfilt{b}{t-1}]$.  In the same way as for the previous case, we obtain
    \begin{equation}
        \label{eq:boundF1F2}
        \pE \big[  F^N _{1,b} F^N _{2,b} \big|\Gfilt{b}{t-1} \big] \leq 2 \left( \frac{N^4 \boundg^4}{\Omega^4 _{t-1}} + \frac{2N^3 \boundg^3}{\Omega^3 _{t-1}} + \frac{N^2 \boundg^2}{\Omega^2 _{t-1}}\right) \sum_{k^{1:6} _{t-1} \in [N]^6} \prod_{\ell = 1}^3 \Pbacksum^b _{t-1}(k^{2\ell - 1} _{t-1}, k^{2\ell} _{t-1}) \eqsp.
    \end{equation} 
    Thus, combining \eqref{eq:boundF1F2}, \eqref{eq:boundF1F3} and \eqref{eq:boundF1F4}, we obtain
    \begin{align*}
         \pE \bigg[ \Const{t}^3 & \sum_{k^{1:6} _t \in [N]^6} \sum_{i^{1:3} \in [M]^3} \1_{i^1 = i^2, i^3 \notin \{i^1, i^2\}} \prod_{\ell = 1}^3 \Pbacksum^b _{t-1}\big( J^{i^\ell} _{k^{2\ell - 1} _t, t-1}, J^{i^\ell} _{k^{2\ell} _t, t-1}\big) \bigg] \\
        & \lesssim \pE \bigg[ \bigg( \sum_{p = 2}^6 \frac{\Omega^6_{t-1} N^p \boundg^p}{\Omega^p_{t-1} N^3(N-1)^3}\bigg) \parismuest{\BS}{b,t-1}(\boldone)^3 \bigg] \lesssim \big\| \parismuest{\BS}{b,t-1}(\boldone) \big\|^3 _3\eqsp.
    \end{align*}
    The other triplets $(i^1, i^2, i^3)$ for which $\card{\{i^1, i^2, i^3\}} = 2$ are handled similarly and are bounded by $\big\| \parismuest{\BS}{b,t-1}(\boldone) \big\|^3 _3$ up to some multiplicative constant independent of $N$. We finally obtain
    \begin{equation}
        \label{eq:Q3card2}
        \pE \bigg[ \frac{\Const{t}^3}{M^3}  \sum_{k^{1:6} _t \in [N]^6} \sum_{i^{1:3} \in [M]^3} \1_{\card{\{i^1, i^2, i^3\}} = 2} \prod_{\ell = 1}^3 \Pbacksum^b _{t-1}\big( J^{i^\ell} _{k^{2\ell - 1} _t, t-1}, J^{i^\ell} _{k^{2\ell} _t, t-1}\big) \bigg] \lesssim \big\| \parismuest{\BS}{b,t-1}(\boldone) \big\|^3 _3 \eqsp.
   \end{equation}
   It now remains to treat the case $\card{\{i^1, i^2, i^3\}} = 1$. Let $(k^1, \cdots, k^d) \in [N]^d$. Denote by $\pos_d(k^{1:d})$ the set of elements in $[N]^d$ with positions of equal elements similar to those of the equal elements in $k^{1:d}$. For example, 
   \begin{align*}
        \pos_3((2,1,1)) & = \{(j,i,i) \mid  (i,j) \in [N]^2, i \neq j \}\eqsp, \\
        \pos_3((1,1,2)) & = \{(i,i,j) \mid  (i,j) \in [N]^2, i \neq j \}\eqsp, \\
        \pos_3((1,2,3)) & = \{(i,j,k) \mid (i,j,k) \in [N]^3, \card{\{i,j,k\}} = 3 \} \eqsp.
   \end{align*}
   Let $p \in [2:6]$ and $(k^1 _t, \cdots, k^6 _t) \in \mathcal{I}^3 _0 \cap \sett^p _3$. Without loss of generality, assume that the first $p$ elements of $k^{1:6} _t$ are all different. 
   \begin{align*}
       \pE \bigg[ \prod_{\ell = 1}^3  \Pbacksum^b _{t-1}\big( J^{i^\ell} _{k^{2\ell - 1} _t, t-1},& J^{i^\ell} _{k^{2\ell} _t, t-1}\big)  \bigg| \Gfilt{b}{t-1} \bigg] \\
       &  =  \sum_{k^{1:6} _{t-1} \in \pos_6(k^{1:6} _t)} \prod_{\ell = 1}^3 \Pbacksum^b _{t-1}(k^{2\ell - 1} _{t-1}, k^{2\ell} _t) \pE \bigg[ \prod_{n = 1}^p \beta^\BS _t(k^n _t, k^n _{t-1}) \bigg| \Gfilt{b}{t-1} \bigg] \\
       & \leq \frac{\boundg^p}{\Omega_{t-1}^p} \sum_{k^{1:6} _{t-1} \in \pos_6(k^{1:6} _t)} \prod_{\ell = 1}^3 \Pbacksum^b _{t-1}(k^{2\ell - 1} _{t-1}, k^{2\ell} _t) \\
       & \leq \frac{\boundg^p}{\Omega_{t-1}^p} \sum_{k^{1:6} _{t-1} \in [N]^6} \prod_{\ell = 1}^3 \Pbacksum^b _{t-1}(k^{2\ell - 1} _{t-1}, k^{2\ell} _t) \eqsp.
        \end{align*}
    Consequently, by Lemma~\ref{lem:cardinals}, 
    \begin{align*}
        \sum_{k^{1:6} _t \in \mathcal{I}^3 _0 \cap \sett^p _3} \pE \bigg[ \prod_{\ell = 1}^3 &  \Pbacksum^b _{t-1}\big( J^{i^\ell} _{k^{2\ell - 1} _t, t-1}, J^{i^\ell} _{k^{2\ell} _t, t-1}\big)  \bigg| \Gfilt{b}{t-1} \bigg]\\
         & \leq \sum_{k^{1:6} _t \in \mathcal{I}^3 _0 \cap \sett^p _3} \frac{\boundg^p}{\Omega_{t-1}^p} \sum_{k^{1:6} _{t-1} \in [N]^6} \prod_{\ell = 1}^3 \Pbacksum^b _{t-1}(k^{2\ell - 1} _{t-1}, k^{2\ell} _t) \\
         & \lesssim \frac{N^p \boundg^p}{\Omega^p _{t-1}} \sum_{k^{1:6} _{t-1} \in [N]^6} \prod_{\ell = 1}^3 \Pbacksum^b _{t-1}(k^{2\ell - 1} _{t-1}, k^{2\ell} _t) \eqsp,
    \end{align*}
    and,
        \begin{multline}
            \label{eq:Q3card1}
            \pE \bigg[ \frac{\Const{t}^3}{M^3}  \sum_{k^{1:6} _t \in [N]^6} \sum_{i^{1:3} \in [M]^3} \1_{\card{\{i^1, i^2, i^3\}} = 1} \prod_{\ell = 1}^3 \Pbacksum^b _{t-1}\big( J^{i^\ell} _{k^{2\ell - 1} _t, t-1}, J^{i^\ell} _{k^{2\ell} _t, t-1}\big) \bigg] \\
             \lesssim \pE \bigg[ \bigg( \sum_{p = 2}^6 \frac{\Omega_{t-1}^6 N^p \boundg^p}{\Omega_{t-1}^p N^3(N-1)^3}\bigg) \parismuest{\BS}{b,t-1}(\boldone)^3 \bigg] \lesssim \big\| \parismuest{\BS}{b,t-1}(\boldone) \big\|^3 _3  \eqsp.
       \end{multline}
    Finally, combining \eqref{eq:Q3card3}, \eqref{eq:Q3card2} and \eqref{eq:Q3card1} we get
  $$
        \big\| \parismuest{\BS}{b,t}(\boldone) \big\|^3 _3 = \pE \bigg[ \frac{\Const{t}^3}{M^3} \sum_{k^{1:6} _t \in [N]^6} \sum_{i^{1:3} \in [M]^3} \prod_{\ell = 1}^3 \Pbacksum^b _{t-1}\big( J^{i^\ell} _{k^{2\ell - 1} _t, t-1}, J^{i^\ell} _{k^{2\ell} _t, t-1}\big) \bigg] \lesssim \big\| \parismuest{\BS}{b,t-1}(\boldone) \big\|^3 _3 \eqsp,
 $$
    and hence $\sup_{N \geq 6} \big\| \parismuest{\BS}{b,t}(\boldone) \big\| _3 < \infty$ by the induction hypothesis. This ends the proof for the case $b_t = 0$.

    If $M = 2$, then \eqref{eq:Q3card3} is equal to $0$ and \eqref{eq:Q3card2}, \eqref{eq:Q3card1} remain the same. The result then follows. If $N < 6$ then it suffices to truncate the sums over $p$ to obtain the result. 
    \end{proof}
\section{Further algorithmic details}
\subsection{Alternative expression of the genealogy tracing variance estimator}
\label{sec:altexpr}
The expression of the CLE estimator \eqref{eq:CLE} provided in the main paper is different from the expression of the estimator appearing in \cite{olssonvar}. We show here that these are  two expressions of the same quantity. Note first that
\begin{align*}
    \left( \frac{1}{N}\sum_{i = 1}^N h(\particle{i}{t}) - \pred{t}^N (h) \right)^2 & = 0 \\
    & = N^{-2}\sum_{i, j \in [N]^2} \1_{\eve^i _{t,0} = \eve^j _{t,0}} \lbrace h(\particle{i}{t}) - \pred{t}^N (h)\rbrace \lbrace h(\particle{j}{t}) - \pred{t}^N(h) \rbrace \\
    & \hspace{2cm} + N^{-2}\sum_{i, j \in [N]^2} \1_{\eve^i _{t,0} \neq \eve^j _{t,0}} \lbrace h(\particle{i}{t}) - \pred{t}^N (h)\rbrace \lbrace h(\particle{j}{t}) - \pred{t}^N(h) \rbrace.
\end{align*}
On the other hand,
\begin{align*}
     \sum_{i, j \in [N]^2} \1_{\eve^i _{t,0} = \eve^j _{t,0}} & \lbrace h(\particle{i}{t}) - \pred{t}^N (h)\rbrace \lbrace h(\particle{j}{t}) - \pred{t}^N(h) \rbrace \\
    & = \sum_{k = 1}^N \sum_{i, j \in [N]^2} \1_{\eve^i _{t,0} = \eve^j _{t,0} = k} \lbrace h(\particle{i}{t}) - \pred{t}^N (h)\rbrace \lbrace h(\particle{j}{t}) - \pred{t}^N(h) \rbrace \\
    & =  \sum_{k = 1}^N \left( \sum_{i = 1}^N \1_{\eve^i _{t, 0} = k} \lbrace h(\particle{i}{t}) - \pred{t}^N (h) \rbrace \right)^2 \eqsp.
\end{align*}
Thus, 
\begin{align*}
    \predasymptvarestim{t}{h} & = -N^{-1} \sum_{i, j \in [N]^2} \1_{\eve^i _{t,0} \neq \eve^j _{t,0}} \lbrace h(\particle{i}{t}) - \pred{t}^N (h)\rbrace \lbrace h(\particle{j}{t}) - \pred{t}^N(h) \rbrace \\
    & = N^{-1} \sum_{k = 1}^N \left( \sum_{i = 1}^N \1_{\eve^i _{t, 0} = k} \lbrace h(\particle{i}{t}) - \pred{t}^N (h) \rbrace \right)^2.
\end{align*}  
where the expression in the second line is that of \cite{olssonvar}. By a similar reasoning, \eqref{eq:truncCLE} is also equivalent to their estimator.
\subsection{Variance estimators for the predictor and filter}
\label{subsec:predfiltervariance}
The asymptotic variances of the predictor and filter \eqref{eq:asymptvarpred}-\eqref{eq:asymptvarfilter} can be expressed using $\asymptvar{\gamma,t}$. Indeed,
    \begin{equation*}
        \frac{\asymptvar{\gamma,t}(h - \pred{t}(h))}{\joint{t}(\boldone)^2} = \sum_{s = 0}^{t} \left\{ \frac{\joint{s} (\mathbf{1}) \joint{s}\big(\Qmarg{s+1:t}[h - \pred{t}(h)]^2\big)}{\joint{t}(\boldone)^2} - \pred{t}(h - \pred{t}(h))^2 \right\} = \asymptvar{\pred, t}(h) \eqsp,
    \end{equation*}
    and using that
    \begin{equation*}
        \frac{\joint{t}(g_t \{h - \filter{t}(h) \})}{\joint{t+1}(\boldone)} = \frac{\joint{t}(g_t h )}{\joint{t+1}(\boldone)} - \filter{t}(h) = 0 \eqsp,
    \end{equation*}
    we get
    \begin{equation}
        \asymptvar{\phi, t}(h) = \frac{\asymptvar{\gamma,t}\big(g_t \{h - \filter{t}(h) \}\big)}{\joint{t+1}(\boldone)^2} \eqsp.
    \end{equation}
Then, replacing  $\joint{t}(h)$ and $\filter{t}(h)$ 
by their empirical approximations 
$\joint{t}^N(h)$ and $\filter{t}^N(h)$, we obtain
    \begin{align}
        \label{eq:suppl:eta_disjoint}
            \asymptvarestim{\eta, t}{\BS}(h) & \eqdef \frac{-N^{t}}{(N-1)^{t+1}} 
                 \sum_{i,j \in [N]^2} \backsum^\zero _{t}(i,j) \big\{ h(\particle{i}{t}) - \pred{t}^N (h) \big\} \big\{ h(\particle{j}{t}) - \pred{t}^N(h) \big\}\eqsp, \\  
        \label{eq:suppl:phi_disjoint}
            \asymptvarestim{\phi,t}{\BS}(h) & \eqdef \frac{-N^{t+2}}{(N-1)^{t+1}} 
                  \sum_{i,j \in [N]^2}  \normweight{i}{t} \normweight{j}{t}  \backsum^\zero _{t}(i,j)\big\{ h(\particle{i}{t}) - \filter{t}^N (h) \big\} \big\{ h(\particle{j}{t}) - \filter{t}^N(h) \big\} \eqsp.
    \end{align}  
    As a consequence of Theorem~\ref*{thm:consistencyVBS}, these estimators are also weakly consistent. 
    \begin{corollary}
        \label{cor:varestimators}
        Let $\As{assp:A}{assp:boundup}$ hold. For any $h \in \bounded{}$, $\asymptvarestim{\eta, t}{\BS}(h) \plim \asymptvar{\eta,t}(h)$ and $\asymptvarestim{\phi, t}{\BS}(h) \plim \asymptvar{\phi,t}(h)$.
    \end{corollary}
    \begin{proof}
        It suffices to note that $\muest{\BS}{b,t}$ and $\mumeasure{b,t}$ are bilinear,  that $\eta^N _t(h) \plim \eta _t(h)$ and to apply Theorem~\ref*{thm:conv} again: 
    \begin{align*}
        & \muest{\BS}{b,t}\left( \{ h - \pred{t}^N(h)\}^{\otimes 2}\right) \\
         & \hspace{1.cm} = \muest{\BS}{b,t}(h^{\otimes 2}) - \pred{t}^N(h) \muest{\BS}{b,t}(h \otimes \boldone) - \pred{t}^N(h) \muest{\BS}{b,t}(\boldone \otimes h) + \pred{t}^N(h)^2 \muest{\BS}{b,t}(\boldone) \\
         & \hspace{1.cm} \plim \mumeasure{b,t}\left( \{ h - \pred{t}(h)\}^{\otimes 2}\right) \eqsp.
    \end{align*}
    Hence, $ \asymptvarestim{\gamma, t}{\BS}(h - \pred{t}^N(h)) \plim \asymptvar{\gamma, t}(h - \pred{t}(h))$ and using the fact that $\joint{t}^N(\boldone)^2 \plim \joint{t}(\boldone)^2$ we get the consistency for the predictive measures. The remaining limit is a straightforward application.
\end{proof}
    \begin{algorithm}
        \caption{Update at step $t+1$ of the variance estimator for the predictor}
        \label{alg:est1alg}
        \begin{algorithmic}
        \Require $\weight{1:N}{t}, \particle{1:N}{t},
         \particle{1:N}{t+1}$ and $\bm{\backsum}^\zero _t$
        \State{Compute $\bm{\beta}^\BS _{t+1}$}
        \If{$\paris$}
        \For{$k \in [1:N]$}
    
             Sample $J^{1:M} _{k,t} \iid \beta^\BS _{t+1}(k, .)$
    
        \EndFor
        \For{$(k,\ell) \in [1:N]^2$}
    
            Set $\backsum^\zero _{t+1}(k,\ell) = \1_{k \neq \ell} \sum_{i = 1}^{M} \backsum^\zero _{t}(
                J^i _{k,t}, J^i _{\ell,t}) / M$
    
        \EndFor
        \Else
    
        \State{Compute $\overline{\bm{\backsum}}^\zero _{t+1} = \bm{\beta}^\BS _{t+1} \bm{\backsum}^\zero _t \bm{\beta}^{\BS \prime} _{t+1}$.}
        \State{Set $\bm{\backsum}^\zero _{t+1} = \overline{\bm{\backsum}}^\zero _{t+1} - \mathrm{Diag}(\overline{\bm{\backsum}}^\zero _{t+1})$.}
        \EndIf
        \State{Compute $\bm{\mathcal{Q}} = \bm{\backsum}^\zero _{t+1} \odot \big[ \big\{ h(\particle{1:N}{t+1}) - \pred{t+1}^N(h) \big\}
         \big\{ h(\particle{1:N}{t+1}) - \pred{t+1}^N(h) \big\}^\top $\big].}\algorithmiccomment{$h$ is applied elementwise}\\
        \Return $- N^t / (N-1)^{t+1} \sum_{i, j \in [N]^2} \bm{\mathcal{Q}}_{i,j}, \quad \bm{\backsum}^\zero _{t+1}$.
        \end{algorithmic}
        \end{algorithm}

\subsection{$\GT$ term by term estimator of the asymptotic variance}
\label{apdx:diffBSGT}
In this section we derive the $\GT$ counterpart of the term by term estimator \eqref{eq:general_termbyterm}. Define for all $t > 0$
\begin{align}
    \label{eq:def_tauGT}
        \backsum^\GT _{b,t} (K^1 _t, K^2 _t) & \eqdef \pE_{\GT} \big[ \intersect{b,t}{K^1 _{0:t} , K^2 _{0:t}} \big| \F{t}, K^1 _t, K^2 _t \big] \eqsp, 
    \end{align}
and $\backsum^\GT _{b,0}(K^1 _0, K^2 _0) \eqdef \1_{K^1 _0 \neq K^2 _0, b_0 = 0} + \1_{K^1 _0 = K^2 _0, b_0 = 1}$.

By the tower property and the definition of $\pE_\GT [ \cdot | \F{t-1} ]$, for all $(k,\ell) \in [N]^2$ and $t > 0$, if $b_t = 0$,
\begin{align*}
    \backsum^\GT _{b,t}(k, \ell) & = \1_{k \neq \ell} \sum_{i, j \in [N]^2} \1_{A^k _t = i, A^\ell _t = j} \backsum^\GT _{b,t-1}(i,j) = \1_{k \neq \ell} \backsum^\GT _{b,t-1}(A^k _{t-1}, A^\ell _{t-1}) \eqsp,
\end{align*}
and if $b_t = 1$,
\begin{align*}
    \backsum^\GT _{b,t}(k, \ell) & = \1_{k = \ell} \sum_{i,j \in [N]^2} \1_{A^k _{t-1} = i }\normweight{j}{t-1} \backsum^\GT _{b,t-1}(i,j) = \1_{k = \ell} \sum_{i = 1}^N \normweight{j}{t-1} \backsum^\GT _{b,t-1}(A^k _{t-1},j) \eqsp.
\end{align*}
Similarly to $\BS$, we have by the tower property
\begin{align}
    \label{eq:QexprGT}
    \muest{\BS}{b, t}(h) & = \prod_{s = 0 }^t N^{b_s} 
    \bigg(\frac{N}{N-1}\bigg)^{1 - b_s} \frac{\joint{t}^N(\mathbf{1})^2}{N^2} \sum_{k, \ell \in [N]^2} \backsum^\GT _{b,t}(k,\ell) h(\particle{k}{t}, \particle{\ell}{t}) \eqsp,
\end{align}
and the term by term estimator is thus 
\begin{equation}
    \label{eq:tbt_exprGT}
    \tbtasymptvarestim{\gamma, t}{\GT}(h) =  \frac{N^{t-1} \joint{t}^N(\boldone)^2}{(N-1)^t} \sum_{k, \ell \in [N]^2} \left\{ S^\GT _t(k, \ell)-\frac{t+1}{N-1}\backsum^\GT _{\zero,t} (k,\ell) \right \} h(\particle{k}{t})h(\particle{\ell}{t}) \eqsp,
\end{equation}
where $S^\GT _t$ is such that for all $(k, \ell) \in [N]^2$,
\begin{align*}
        \label{eq:updatesumesGT}
         S^\GT _t(k,\ell) & = \sum_{s = 0}^t \backsum^\GT _{e_s, t}(k,\ell)  =  \1_{k = \ell} \sum_{i = 1}^N \normweight{j}{t-1} \backsum^\GT _{\zero,t-1}(A^k _{t-1},j) + \1_{k \neq \ell} S^\GT _{t-1}(A^k _{t-1}, A^\ell _{t-1})
         \eqsp,
\end{align*}   
which shows that \eqref{eq:tbt_exprGT} is also updated online in a rather simple way by propagating the matrices $\bm{S}^\GT _t$ and $\bm{\backsum}^\GT _{\zero,t}$.

\section{Technical results}
\label{sec:techres}
\begin{thm}[\textbf{Generalized dominated convergence theorem}]
    \label{thm:GDCT}
    Let $(f_N) _{N \in \N}$ be a sequence of $\sigmaX$-measurable functions and $(g_N)_{N \in \N}$ a sequence of non-negative $\sigmaX$-measurable functions. 
    Assume that the following assumptions hold.
    \begin{enumerate}[label=(\roman*)]
        \item There exists $C > 0$ such that $\left| f_N(x) \right| \leq C g_N(x)$ for all $N \in \N$ and $x \in \Xset$. 
        \item $(g_N)_{N\in \N}$ converges pointwise to $g$ and $\simplelim \int g_N \rmd \nu = \int g \rmd \nu < \infty$.
        \item $(f_N)_{N\in \N}$ converges pointwise to $f$.   
    \end{enumerate}
    Then, $f$ is $\nu$-integrable and $\simplelim \int f_N \rmd \nu = \int f \rmd \nu$.
\end{thm}
\begin{proof}
    The proof can be found in \cite{royden1988real}.
\end{proof}
Theorem~\ref{thm:parishoeff} and Lemma~\ref{lem:genhoeff} are borrowed from \cite{paris} and \cite{doucmoulines} respectively.
\begin{thm}
    \label{thm:parishoeff}
Assume that $\As{assp:B}{assp:boundup}$ hold. Then, for all $s \in \mathbb{N}\eqsp,$ $h_s \in \bounded{s+1}$ and $\big(f_{s}, \tilde{f}_{s}\big) \in \bounded{}^2$, there exist constants $\big(C_{s}, \widetilde{C}_{s}\big) \in\left(\mathbb{R}_{+}^{*}\right)^{2}$, depending on $h_{s}, f_{s}$, and $\tilde{f}_{s}$, such that for all $N \in \N^{*}$ and all $\varepsilon \in \mathbb{R}_{+}^{*}$,
\begin{equation}
    \label{eq:hoeffding1}
    \pP \left( \bigg| N^{-1} \sum_{i=1}^{N} \weight{i}{s} \big\{\big(\bwpath{s}^N[h_s] f_s \big)(\particle{i}{s})+ \tilde{f}_{s}\big(\particle{i}{s}\big)\big\}-\pred{s}\big(\bwpath{s}[h_s] f_{s}+\tilde{f}_{s}\big)\bigg| \geq \varepsilon\right)  \leq C_{s} \exp \left(-\widetilde{C}_{s} N \varepsilon^{2}\right)\eqsp, 
\end{equation}
\begin{equation}
    \label{eq:hoeffding2}
\pP \left(\left| \sum_{i=1}^{N} \normweight{i}{s} \left\{\left(\bwpath{s}^N[h_s] f_s \right)(\particle{i}{s})+ \tilde{f}_{s}\left(\particle{i}{s}\right)\right\}-\filter{s}\left(\bwpath{s}[h_s] f_{s}+\tilde{f}_{s}\right)\right| \geq \varepsilon\right) \leq C_{s} \exp \left(-\widetilde{C}_{s} N \varepsilon^{2}\right)\eqsp.
    \end{equation}
\end{thm}
\begin{lemma}
    \label{lem:genhoeff}
    Assume that $a_{N}, b_{N}$ and $b$ are random variables defined on the same probability space such that there exist positive constants $\beta, B_1, C_1, B_2, C_2$ and $M$ satisfying the following assumptions.
\begin{itemize}
    \item $\left|a_{N} / b_{N}\right| \leq M \eqsp,$ $\pP - a.s.$ and $b \geq \beta$.
    \item For all $\varepsilon>0$ and all $N \geq 1, \mathbb{P}\left(\left|b_{N}-b\right|>\varepsilon\right) \leq B_1 \exp(-C_1 N \varepsilon^{2})$.
    \item For all $\varepsilon>0$ and all $N \geq 1, \mathbb{P}\left(\left|a_{N}\right|>\varepsilon\right) \leq B_2 \exp(-C_2 N \varepsilon^2)$.
\end{itemize}
Then, there exist two positive constants $B_3, C_3$ such that
$$
\mathbb{P}\left(\left|\frac{a_{N}}{b_{N}}\right|>\varepsilon\right) \leq B_3 \exp (-C_3 N \varepsilon^2) \eqsp.
$$
\end{lemma}
\section{Asymptotic variance of the joint predictive distribution}
\label{apdx:asymptvar}
In this section we provide some intuition on \eqref{eq:asymptvarjoint}. Let $h \in \bounded{}$. 
By the law of total variance, 
\begin{equation}
    \label{eq:totalvariance}
    \pV\big[ \joint{t+1}^N (h) \big] = \pV\big[ \pE\big[ \joint{t+1}^N (h) \mid \F{t} \big] \big] + \pE \big[ \pV[\joint{t+1}^N (h) \mid \F{t}\big]\big] \eqsp.
\end{equation}
As $\joint{t+1}^N(\boldone)$ is $\F{t}$-measurable and the particles at time $t+1$ are i.i.d conditionally on $\F{t}$, we have that
\begin{align}
    \label{eq:gammaQ}
    \pE\big[ \joint{t+1}^N (h) \mid \F{t} \big] & = \joint{t+1}^N(\boldone) \sum_{i = 1}^N \frac{\weight{i}{t}}{\Omega_{t}} \transition{t+1}[h](\particle{i}{t}) \nonumber\\
    & = \joint{t}^N(\boldone) N^{-1} \Omega _{t} \sum_{i = 1}^N \frac{\weight{i}{t}}{\Omega_{t}} \transition{t+1}[h](\particle{i}{t}) = \joint{t}^N \big(\Q{t+1}[h]\big) \eqsp.
\end{align}
On the other hand,
\begin{align*}
\pV \big[ \joint{t+1}^N [h] \mid \F{t} \big] & = \joint{t+1}^N (\mathbf{1})^2  \pV \bigg[ \frac{1}{N} \sum_{i = 1}^N h(\particle{i}{t+1}) \bigg| \F{t} \bigg] = N^{-1} \joint{t+1}^N(\boldone)^2 \pV _{\filter{t}^N \transition{t+1}} \big[ h(\particle{}{t+1}) \big] 
\end{align*}
where 
\begin{align*}
    \pV _{\filter{t}^N \transition{t+1}} \big[ h(\particle{}{t+1}) \big] = \filter{t}^N \transition{t+1}\big( \big\{ h - \filter{t}^N \transition{t+1}(h) \big\}^2 \big) \eqsp.
\end{align*}
Therefore, 
\begin{align}
    \label{eq:condvar}
    \pV \big[ \joint{t+1}^N (h) \mid \F{t} \big] & = N^{-2} \joint{t}^N(\boldone)^2 \Omega_t \pred{t}^N \big( \Q{t+1} \big[ \big\{ h - \filter{t}^N \transition{t+1}(h) \big\}^2 \big]\big) \nonumber \\
    & = N^{-1} \joint{t+1}^N(\boldone) \joint{t}^N \big( \Q{t+1} \big[ \big\{ h - \filter{t}^N \transition{t+1}(h) \big\}^2 \big] \big) \eqsp.
\end{align}
Replacing \eqref{eq:gammaQ}-\eqref{eq:condvar} in \eqref{eq:totalvariance}, we get the recursive formula
\begin{align*}
    N\pV\big[ \joint{t+1}^N (h) \big] & = N\pV\big[ \joint{t}^N \big(\Q{t+1}[h]\big) \big] + \pE\big[\joint{t+1}^N (\mathbf{1}) \joint{t}^N \big( \Q{t+1}\big[ \big\{ h - \filter{t}^N\transition{t+1}(h) \big\}^2 \big] \big) \big] \\
    & = \transition{0}\big(\Qmarg{1:t+1}[h]^2\big) - \joint{t+1}(h)^2 \\
    & \hspace{1.3cm} + \sum_{s = 1}^{t+1} \pE\big[\joint{s}^N (\mathbf{1}) \joint{s-1}^N \big( \Q{s}\big[\big\{ \Qmarg{s+1:t+1}[h] - \filter{s-1}^N\big(\transition{s}\big[\Qmarg{s+1:t+1}[h] \big] \big) \big\}^2 \big] \big)\big] \eqsp.
\end{align*}
With multiple applications of \eqref{eq:asconv} in the main paper, we get that
\begin{multline}
    \joint{s}^N (\mathbf{1}) \joint{s-1}^N \big( \Q{s}\big[\big\{ \Qmarg{s+1:t+1}[h] - \filter{s-1}^N\big(\transition{s}\big[\Qmarg{s+1:t+1}[h] \big] \big) \big\}^2 \big] \big) \\ \overset{a.s.}{\longrightarrow} \joint{s} (\mathbf{1}) \joint{s-1} \big( \Q{s}\big[\big\{ \Qmarg{s+1:t+1}[h] - \filter{s-1}\big(\transition{s}\big[\Qmarg{s+1:t+1}[h] \big] \big) \big\}^2 \big] \big) \eqsp,
\end{multline}
and, using that $\joint{s}(\boldone) \filter{s-1}( \transition{s}[h]) = \joint{s}(h)$ and $\joint{s}\big( \Qmarg{s+1:t+1}[h] \big) = \joint{t+1}(h)$ for all $h$ we get
\begin{align*}
    \joint{s} (\mathbf{1}) \joint{s-1} \big( \Q{s}\big[\big\{ \Qmarg{s+1:t+1}[h] & - \filter{s-1}\big(\transition{s}\big[\Qmarg{s+1:t+1}[h] \big] \big) \big\}^2 \big] \big) \\
    & = \joint{s} (\mathbf{1}) \joint{s} \big( \big\{ \Qmarg{s+1:t+1}[h] - \filter{s-1}\big(\transition{s}\big[\Qmarg{s+1:t+1}[h] \big] \big) \big\}^2 \big)\\
     & = \joint{s} (\mathbf{1}) \joint{s}\big[\Qmarg{s+1:t+1}(h)^2\big] - \joint{t+1}(h)^2 \eqsp.
\end{align*}
Finally, by $\As{assp:B}{}$ and the boundedness of $h$, $| \joint{s}^N(\boldone) | \leq \boundg^s $, $| \Qmarg{s+1:t}[h] | \leq \boundg^{t-s} | h |_\infty $, thus
\[ 
    \big| \joint{s}^N (\mathbf{1}) \joint{s-1}^N \big( \Q{s}\big[\big\{ 
        \Qmarg{s+1:t+1}[h] - \filter{s-1}^N\big(\transition{s}
        \big[\Qmarg{s+1:t+1}[h] \big] \big) \big\}^2 \big] \big) 
        \big| \leq 4\boundg^{2t} | h |_\infty \eqsp,
\]
and by the dominated convergence theorem, for any $s \in [0:t]$,
\begin{multline}
\simplelim \pE\big[ \joint{s}^N (\mathbf{1}) \joint{s-1}^N 
\big( \Q{s}\big[\big\{ \Qmarg{s+1:t+1}[h] - \filter{s-1}^N
\big(\transition{s}\big[\Qmarg{s+1:t+1}[h] \big] \big) 
\big\}^2 \big] \big) \big]\\ = \joint{s} (\mathbf{1})
 \joint{s}\big[\Qmarg{s+1:t+1}(h)^2\big] - \joint{t+1}(h)^2 \eqsp.
\end{multline}
and $N\pV\big[ \joint{t+1}^N [h] \big] \longrightarrow 
\sum_{s = 0}^{t+1}  \big\{ \joint{s} (\mathbf{1}) 
\joint{s}\big[\Qmarg{s+1:t+1}(h)^2\big] - \joint{t+1}(h)^2 \big\}$. As argued in Section~\ref*{subsec:disj} of the main paper, $\pE \big[ N \big(\joint{t+1}^N (h) - \joint{t+1}(h) \big)^2 \big]$ converges to the asymptotic variance when $N$ goes to infinity, and since $\joint{t+1}^N(h)$ is an unbiased estimator of $\joint{t+1}(h)$, $N\pV\big[ \joint{t+1}^N [h] \big] = \pE \big[ N \big(\joint{t+1}^N (h) - \joint{t+1}(h) \big)^2 \big]$. Therefore,
\[ 
    \simplelim \pE \big[ N \big(\joint{t+1}^N (h) - 
    \joint{t+1}(h) \big)^2 \big] = \sum_{s = 0}^{t+1}  \big\{ \joint{s}
     (\mathbf{1}) \joint{s}\big[\Qmarg{s+1:t+1}(h)^2\big] - \joint{t+1}(h)^2 \big\} \eqsp,
    \]
    which ends the proof.

\section{Computational time comparison}
\begin{figure}[h!]
    \centering
    \includegraphics[width = 10cm]{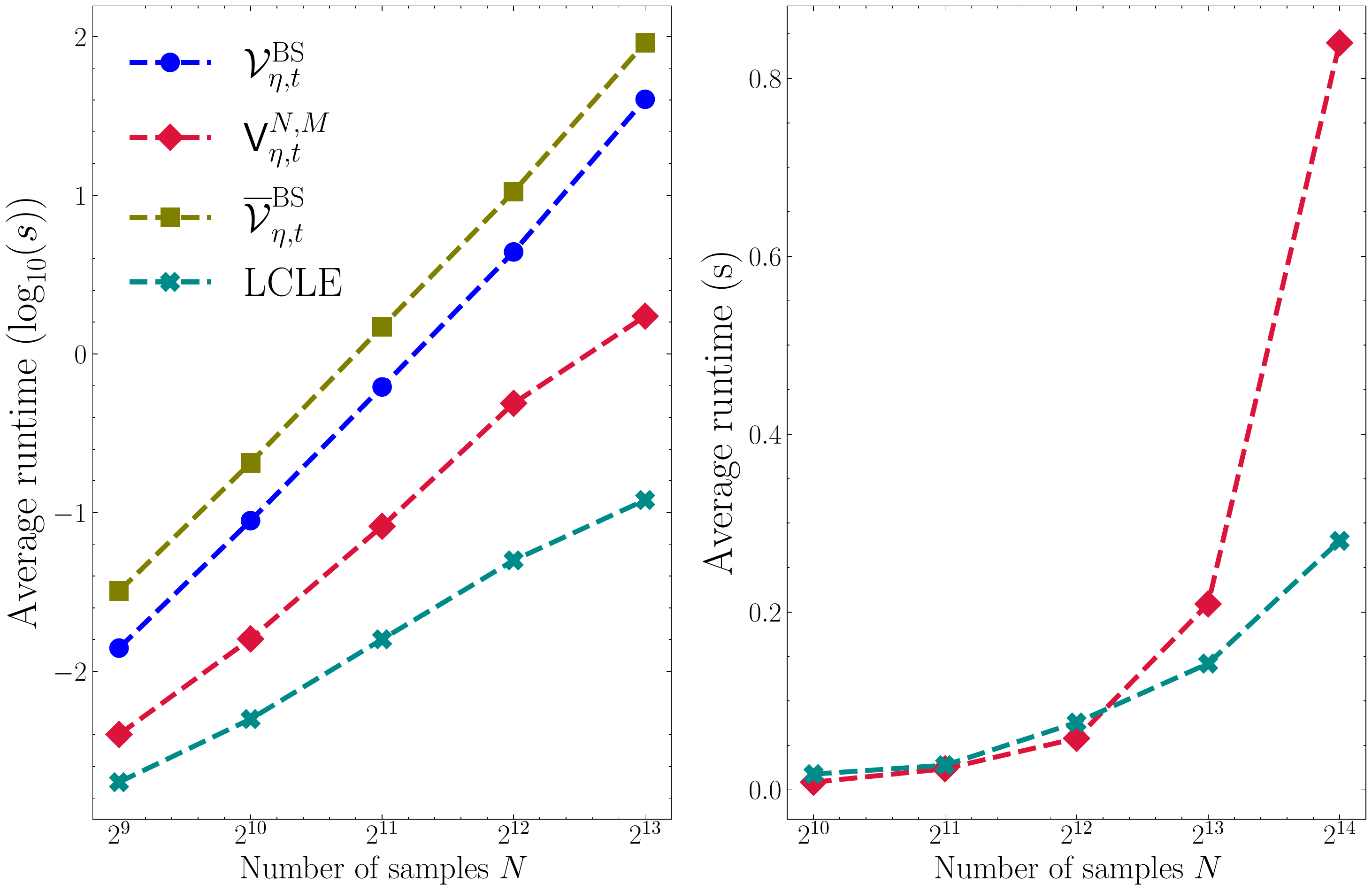}
    \caption{Comparison of time complexity (left) and runtime (right) for the different estimators, per time step. The runtime on the left plot is on CPU and that on the right plot on GPU, only for our most competitive estimator.}
    \label{fig:runtime}
\end{figure}

\end{document}